\documentclass[
onecolumn, aps, pra,
tightenlines,
11pt,
longbibliography,
showpacs,
nofootinbib,
notitlepage,
superscriptaddress]{revtex4-2}

\usepackage{subfiles}
\usepackage{amsmath, amssymb, amsthm, bbold, mathtools}
\usepackage{mathrsfs}
\usepackage{chngcntr}
\usepackage{braket, physics, bm, enumitem}
\usepackage{booktabs}
\usepackage{subcaption}

\makeatletter
\long\def\revtex@flush@makecaption#1#2{%
  \par
  \vskip\abovecaptionskip
  \begingroup
    \footnotesize\rmfamily
    \samepage
    \flushing
    \let\footnote\@footnotemark@gobble
    \@make@capt@title{#1}{#2}\par
  \endgroup
  \vskip\belowcaptionskip
}%
\AtBeginDocument{\let\@makecaption\revtex@flush@makecaption}%
\makeatother
\newcommand{\mlc}[1]{\begin{tabular}[c]{@{}l@{}}#1\end{tabular}}
\newcommand{\thc}[1]{\multicolumn{1}{l}{\mlc{#1}}}

\usepackage{setspace}
\usepackage[colorlinks, allcolors=black]{hyperref}
\usepackage[capitalize, nameinlink]{cleveref}
\usepackage{quantikz, adjustbox}
\usetikzlibrary{matrix}

\usepackage{algorithm, algcompatible}
\usepackage{algpseudocode}

\allowdisplaybreaks

\newtheorem{lem}{Lemma}

\newtheorem{rmk}{Remark}
\newtheorem{cor}{Corollary}
\newtheorem{prop}{Proposition}
\newtheorem{thm}{Theorem}
\newtheorem{exm}{Example}
\newtheorem*{thm*}{Theorem}

\Crefname{prop}{Proposition}{Propositions}
\Crefname{thm}{Theorem}{Theorems}
\Crefname{lem}{Lemma}{Lemmas}
\Crefname{defn}{Definition}{Definitions}
\Crefname{rmk}{Remark}{Remarks}
\Crefname{exm}{Example}{Examples}
\Crefname{alg-line}{Line}{Lines}
\crefname{alg}{Algorithm}{Algorithms}
\Crefname{cor}{Corollary}{Corollaries}
\Crefname{appendix}{Appendix}{Appendices}

\DeclareMathOperator{\polylog}{polylog}
\DeclareMathOperator{\poly}{poly}
\DeclareMathOperator{\diag}{diag}

\DeclarePairedDelimiter{\ceil}{\lceil}{\rceil}

\renewcommand{\epsilon}{\varepsilon}
\renewcommand{\tilde}[1]{\widetilde{#1}}
\renewcommand{\hat}[1]{\widehat{#1}}
\newcommand{\llangle}{\langle\!\langle}
\newcommand{\rrangle}{\rangle\!\rangle}

\mathcode`l="8000
\begingroup
\makeatletter
\lccode`\~=`\l
\DeclareMathSymbol{\lsb@l}{\mathalpha}{letters}{`l}
\lowercase{\gdef~{\ifnum\the\mathgroup=\m@ne \ell \else \lsb@l \fi}}%
\endgroup
\makeatother

\begin{document}

\title{Breaking the Curse of Dimensionality in Quantum PDE Solvers \texorpdfstring{\\}{ }via Gevrey Regularity}

\author{Pooya~Ronagh}
\email{Corresponding author: pooya.ronagh@uwaterloo.ca}
\affiliation{Institute for Quantum Computing, University of Waterloo, Waterloo, ON, N2L 3G1, Canada}
\affiliation{Department of Physics \& Astronomy, University of Waterloo, Waterloo, ON, N2L 3G1, Canada}
\affiliation{Perimeter Institute for Theoretical Physics, Waterloo, ON, N2L 2Y5, Canada}
\affiliation{Microsoft Quantum, Redmond, WA, 98052, USA}

\author{Mariia~Sobchuk}
\affiliation{Institute for Quantum Computing, University of Waterloo, Waterloo, ON, N2L 3G1, Canada}
\affiliation{Department of Physics \& Astronomy, University of Waterloo, Waterloo, ON, N2L 3G1, Canada}

\author{Xiaoran~Li}
\affiliation{Institute for Quantum Computing, University of Waterloo, Waterloo, ON, N2L 3G1, Canada}

\author{Grecia~Castelazo}
\affiliation{Department of Electrical Engineering \& Computer Science, Massachusetts Institute of Technology, Cambridge, MA 02139, USA}
\affiliation{Department of Computer Science, University of California,
Santa Barbara, CA 93106, USA}

\author{Ala~Shayeghi}
\affiliation{Institute for Quantum Computing, University of Waterloo, Waterloo, ON, N2L 3G1, Canada}
\affiliation{Department of Applied Mathematics, University of Waterloo, Waterloo, ON, N2L 3G1, Canada}
\affiliation{National Research Council Canada, Waterloo, ON, N2L 3G1, Canada}

\date{\today}

\begin{abstract}
We connect different degrees of smoothness of real-valued periodic functions to the number of qubits required for their high-precision Fourier-basis amplitude encodings. Our central observation is that the Gevrey hierarchy, which stratifies the space between smooth and analytic functions, provides a natural class for high-precision quantum algorithms. We use Gevrey regularity to analyze a quantum algorithm for solving general linear partial differential equations (PDEs) with periodic boundary conditions. We then extend the solver to PDEs with Dirichlet boundary conditions using thin Gevrey collars and a regularized inverse. Our resource bounds specify precision-dependent Fourier cutoffs and account for restricted discrete stability, elementary gate costs, and restriction to the Dirichlet domain. Finally, we apply our approach to isotropic and anisotropic Poisson equations with periodic and Dirichlet boundary conditions.
\end{abstract}

\maketitle

\tableofcontents

\section{Introduction}
\label{sec:intro}

Quantum computing is anticipated to provide striking computational advantages for specific problems, such as period finding in integer rings, which lies at the core of Shor's factoring algorithm. Such problems possess particular structures that a quantum computer can exploit \cite{aaronson2009need}. This raises a key question: what kind of structure enables such advantages, and do these structures arise in meaningful real-world applications? For discrete problems, namely those involving Boolean functions $f:\{0, 1\}^n \to \{0, 1\}$, this question has been studied extensively, yielding deep insights into the mechanisms of quantum speedup, and in particular the crucial role of the quantum Fourier transform (QFT) in algorithms exhibiting exponential speedups \cite{aaronson2009need, childs2010quantum, chailloux2018note, aaronson2020quantum}. Far less is understood, however, in the continuous-domain setting, where a comparable systematic effort to guide the search for impactful quantum applications has been lacking.

Several previous works have used quantum linear algebra algorithms \cite{harrow2009quantum, childs2017quantum, costa2022optimal} to solve differential equations. Among these, many leverage the QFT to solve the equation in the Fourier (or plane-wave) basis \cite{leyton2008quantum, childs2021high, tong2021fast, linden2022quantum, liu2023efficient}, or generalize to Euclidean domains by preparing the solution on a pseudo-spectral lattice (such as the Chebyshev lattice) \cite{childs2021high}, or impose all-zero Dirichlet boundary conditions on a hypercube \cite{cao2013quantum, liu2021efficient, liu2023efficient}. Na\"ively, however, one still requires a discretization of size $\mathcal O((1/\epsilon)^d)$ in either of the Euclidean or Fourier domains to achieve a target accuracy $\epsilon>0$ on the solution (in the $\ell^2$ norm), leading to an exponential dependence on $d$.

Classical numerical methods for general PDEs can improve on this na\"ive bound to $\mathcal O(\log^d(1/\epsilon))$-type complexities at the price of stronger structural assumptions such as isotropic analyticity or bounded mixed derivatives (see \cref{tab:classical-pde-methods}). In every such case, however, $d$ still appears as an exponent rather than as a bounded polynomial factor, so the well-known \emph{curse of dimensionality} persists.

\begin{table}[t]
\footnotesize
\centering
\begin{tabular}{p{3.1cm} p{5.2cm} p{3.6cm} p{3.9cm}}
\toprule
Algorithm &
Mathematical assumptions &
Complexity in $\epsilon$ and $d$ &
Reference
\\
\midrule
\mlc{Uniform-grid \\ FDM/FEM/FVM} &
Finitely differentiable &
$\mathrm{poly}((1/\epsilon)^d)$ &
folklore
\\[2ex]

\mlc{\\ Adaptive h--p-FEM \\ on geometrically \\ graded meshes} &
\mlc{\\ Analytic on a polygonal domain \\ 
with possibly finitely many \\
corner singularities} &
$\mathrm{poly}(\log^d(1/\epsilon))$ &
\cite[Sec 6]{babuvska1987hp}
\\[2ex]

\mlc{Spectral methods \\ on full tensor grids} &
\mlc{\\ Isotropically analytic in each \\ 
variable (Bernstein ellipse / \\
periodic strip)} &
$\mathrm{poly}(\log^d(1/\epsilon))$ &
\mlc{\\
\cite[Thms 3.1, 4.4]{tadmor1986exponential} \\
\cite[Rmk 27]{gheorghiu2007spectral} \\
\cite[Rmk 1.10]{shen2011spectral}}
\\[2ex]

\mlc{\\ Sparse-grid FEM }&
\mlc{\\ Bounded mixed second derivatives \\
(dominating mixed smoothness)} &
$\epsilon^{-1/2}\log^{3(d-1)/2}(1/\epsilon)$ &
\mlc{\\
\cite[Eq 2.3, 2.9]{zenger1991sparse} \\
\cite[Thm 3.8, Lem 3.13]{bungartz2004sparse}}
\\[2.5ex]

\mlc{\\ Sparse spectral \\ methods (hyperbolic \\ cross)} &
\mlc{\\ Bounded mixed derivatives up to \\
a fixed order $m$ (Jacobi-weighted \\
Korobov space)} &
$\epsilon^{-1/m}\log^{d-1}(1/\epsilon)$ &
\cite[Thm 8.10]{shen2011spectral}
\\[2ex]

\mlc{\\ \textbf{This work}} &
\mlc{\\ Gevrey regularity ($s \geq 0$)} &
$(sd)^{\ceil{\mathcal J}s}\log^{\ceil{\mathcal J}s + 1}(\frac{1}{\epsilon})$ &
\cref{thm:linear-pde-complexity}
\\[2.5ex]
\bottomrule
\end{tabular}
\caption{Classical complexity of solving a general $d$-dimensional PDE to accuracy $\epsilon$, contrasted with the quantum algorithm of this work. For each classical method we report its sharpest published rate together with the structural assumption required to attain it, stating the strongest claim for each method. In our own row, $s$ is the Gevrey order of the solution and $\ceil{\mathcal J}$ is the PDE order. The displayed quantum rate counts oracle queries for periodic solution-state preparation, with the other parameters of \cref{thm:linear-pde-complexity} held fixed; \cref{thm:dirichlet-complexity} gives the Dirichlet counterpart.}
\label{tab:classical-pde-methods}
\end{table}

In this paper, we discuss regularity conditions that allow quantum algorithms to alleviate this curse of dimensionality. We study the quantum computational hardness of preparing amplitude encodings of a broad class of periodic real functions, whose smoothness is much weaker than analyticity. This gives a full picture of how the various regularity classes affect the time and space complexity of preparing such encodings (\cref{tab:fourier-tailedness}). We also construct efficient block encodings for the Fourier approximations of all higher-order differential operators, which is critical to achieving full $\poly(d)$ complexity. These block encodings give concrete gate and qubit counts for general linear time-independent PDEs with periodic boundary conditions (\cref{thm:linear-pde-complexity}). A thin-collar Gevrey extension then certifies that a periodic Fourier trial space of polylogarithmic cutoff resolves Dirichlet problems on hypercubes, on which we collocate the equation and its boundary condition, while controlling the cost of extracting the solution state (\cref{thm:dirichlet-complexity}). We envisage many high-impact applications of our framework as future research directions (see \cref{sec:discussion}). However, in the present paper, we demonstrate our framework on isotropic and anisotropic Poisson equations with periodic and Dirichlet boundary conditions.

\begin{table}[t]
\small
\centering
\begin{tabular}{p{3.9cm} p{3.85cm} p{4.45cm} p{3.7cm}}
\toprule
Regularity class &
\thc{Decay rate of Fourier \\ coefficients ($|\hat f_\omega|$)}&
\thc{Truncation-error bound \\ $\bar{\mathcal T}_N$ at cutoff $N$}&
\thc{Cutoff $N$ to guarantee \\ target precision $\epsilon$}\\
\midrule
Continuous $(\mathscr{C}^{0})$ &
$o(1)$ &
$o(1)$ &
no uniform rate
\\[2ex]

$p$-smooth $(\mathscr{C}^{p})$ &
$C\norm{\omega}_\infty^{-p}$ &
$\frac{2^d C}{\sqrt{2p - d}} N^{-p + d/2}$ &
$\left(\frac{4^d C^2}{(2p - d)\epsilon^2}\right)^{\frac{1}{2p - d}}$
\\[2ex]

$(p, q)$-H\"older $(\mathscr{C}^{p, q})$ &
$C\norm{\omega}_\infty^{-p - q}$ &
$\frac{2^d C}{\sqrt{2p + 2q - d}} N^{-p - q + d/2}$ &
$\left(\frac{4^d C^2}{(2p + 2q - d)\epsilon^2}\right)^{\frac{1}{2p + 2q - d}}$
\\[2.5ex]

smooth $(\mathscr{C}^\infty)$ &
sub-polynomial &
sub-polynomial &
super-logarithmic
\\[2ex]

$s$-\text{Gevrey} ($\mathscr{G}^s$ for $s > 1$) &
$C e^{-\frac34(r\norm{\omega}_\infty)^{\frac1s}}$ &
$2^d C \sqrt{\frac{4s}{3r^{1/s}} N^{d - 1/s}} e^{-\frac34(rN)^{\frac1s}}$ &
$\frac1r\left(\frac43\log\frac{4^{d + 1} C^2s}{3r^{1/s}\epsilon^2}\right)^s$
\\[2ex]

\text{Analytic} ($\mathscr{C}^\omega$ or $\mathscr{G}^1$) &
$C e^{-\frac34 r\norm{\omega}_\infty}$ &
$2^d C \sqrt{\frac{4N^{d - 1}}{3r}}\, e^{-\frac34 rN}$ &
$\frac4{3r}\log\frac{4^{d + 1} C^2}{3r\epsilon^2}$
\\[2ex]

Entire ($\mathscr{E}$) &
\mlc{faster than $Ce^{-r\norm{\omega}_\infty}$, \\ $\forall\, r>0$} &
\mlc{faster than $Ce^{-rN}$, \\ $\forall\, r>0$} &
sub-linear in $\log \frac1\epsilon$
\\[2ex]

\mlc{Entire of finite order \\ ($\mathscr{G}^s$ for $0 < s < 1$)} &
$C e^{-\frac{s}{2}(r\norm{\omega}_\infty)^{\frac1s}}$ &
$2^d C \sqrt{\frac{2}{r^{1/s}} N^{d - 1/s}}\, e^{-\frac{s}{2}(rN)^{\frac1s}}$ &
$\frac1r\left(\frac2s\log\frac{2 \cdot 4^d C^2}{r^{1/s}\epsilon^2}\right)^s$
\\[2ex]

Finite bandwidth ($s=0$) &
$0$ for $\norm{\omega}_\infty > 1/r$ &
$0$ for $N\geq 1/r$ &
\mlc{$N=\lceil 1/r\rceil$ \\ (exact, any $\epsilon\geq0$)}
\\[1ex]
\bottomrule
\end{tabular}
\caption{Relating quantum computational hardness of problems regarding a periodic $d$-variate function $f: \mathbb T^d \to \mathbb R$ to its smoothness. The second column provides the decay law of Fourier coefficient $|\hat f_\omega|$ in each of the regularity classes enumerated as a function of the infinite norm of the Fourier mode, $\norm{\omega}_\infty$. The third column gives the bound $\bar{\mathcal T}_N$ on the truncation error (tailedness), $\mathcal T_N := \|\hat f - \hat f_N\|_2 \leq \bar{\mathcal T}_N$. Here the integer $N$ is a cutoff value for the infinite norm of the Fourier modes. The last column gives sufficient cutoffs for $\bar{\mathcal T}_N \leq \epsilon$. The Gevrey coefficient bounds hold when $(r\norm{\omega}_\infty)^{1/s} \geq 16$; the tail bounds and logarithmic cutoffs hold in the regimes specified in \cref{prop:tail-decay-rate}. Stronger rates and their parameter dependence are discussed in \cref{rmk:gevrey-coefficient-onset,rmk:sigma-to-r-subtlety} for the Gevrey and analytic cases, respectively. The algebraic tail estimates require $p > d/2$ and $p + q > d/2$, respectively. So for a quantum register storing an amplitude encoding of $f$ in the Fourier basis, this translates to $\mathcal O(d\log N)$ qubits required. Throughout, $r$ measures the scale of the decay (e.g.\ the radius of convergence for analytic functions, or the growth rate of an entire extension into the complex plane, where it no longer corresponds to a distance to a singularity). All columns use the same parameters; in the Gevrey rows, including $s = 1$, $C$ and $r$ are the derivative-bound parameters from \eqref{eq:gevrey-smoothness-definition}. A detailed discussion and derivations are provided in \cref{sec:fourier}.}
\label{tab:fourier-tailedness}
\end{table}

Our results pertain to general linear PDEs; nonetheless, they contribute to the prior body of work on high-precision quantum algorithms \cite{cao2013quantum, childs2021high, tong2021fast}. These references, however, focus on specific types of PDEs, such as the Poisson or constant-coefficient second-order elliptic equations, and leave some important solution-regularity dependencies implicit in their claimed complexity. For instance, \cite{childs2021high} states a complexity result for its spectral method that requires a uniform bound on the solution's derivatives at \emph{arbitrarily high} orders. Such a requirement is considerably stronger than any finite-order smoothness, or even analyticity, and fails for simple periodic functions such as $\sin(kx)$ with any $k > 1$. We also note that the need for ``sufficient smoothness'' has been mentioned in previous papers \cite{cao2013quantum, childs2021high, leng2025operator}. However, here we quantify its effect on Fourier cutoffs and PDE-solver resources. We suspect that high-precision $\poly(d)$ quantum algorithms for all the non-Gevrey function classes (i.e., the first four rows of \cref{tab:fourier-tailedness}) may be fundamentally obstructed (\cref{rmk:smooth-dft-no-go}).

The paper is organized as follows. \Cref{sec:fourier-results} develops the Fourier-analytic framework and the truncation-error bounds. \Cref{sec:periodic-boundary-conditions} presents the discretization error, quantum algorithm, and resource bounds for linear PDEs with periodic boundary conditions, and \cref{sec:dirichlet} adapts the construction to Dirichlet boundary conditions. \Cref{sec:poisson} treats periodic and Dirichlet Poisson equations. We conclude in \cref{sec:discussion}.

\section{Fourier Analytic Results}
\label{sec:fourier-results}

\subsection{Aliasing sums and Gevrey smoothness}

Our Fourier representation uses $2\pi$-periodic $d$-variate functions $f$, viewed on the flat torus $\mathbb T^d := \mathbb R^d / (2\pi \mathbb Z)^d$; for Dirichlet problems these are extensions of the physical solution. Periodicity of $f$ lets us manipulate it efficiently through its Fourier expansion $\sum_{\omega \in \mathbb Z^d} \hat f_\omega \, e^{i\langle \omega, x\rangle}$. The Fourier coefficients are defined as $\hat f_\omega = \mathbb E_{x \in \mathbb T^d}[f(x)e^{-i\langle \omega, x\rangle}]$. A computationally more accessible set of coefficients associated with $f$ is given by the \emph{discrete} Fourier transform coefficients:
\begin{equation}
\tilde f_\omega= \mathbb E_{x \in \Gamma_N}[f(x) e^{-i \langle \omega, x\rangle}],
\end{equation}
obtained by averaging $f$ over the grid
\begin{equation}
\label{eq:spatial-grid}
\Gamma_N:= \left\{\frac{2\pi n}{2N+1}: n \in \mathbb Z, |n|\leq N \right\}^d.
\end{equation}
In fact, we can obtain an amplitude encoding of the discrete Fourier transform (DFT) of $f$ from the amplitude encoding of $f$ itself on the discrete spatial lattice:
\begin{equation}
\label{eq:DFT-transform}
F_N^{\otimes d}: f_N \mapsto (2N+1)^{d/2} \tilde f_N.
\end{equation}
Here $f_N$ and $\tilde f_N$ denote the vectors $(f(x))_{x\in \Gamma_N}$ and $(\tilde f_\omega)_{\omega \in \Lambda_N}$, respectively, and the discrete Fourier coefficient modes $\omega$ reside on another integer lattice with $(2N+1)^d$ grid points:
\begin{equation}
\label{eq:frequency-grid}
\Lambda_N= \{\omega\in \mathbb Z^d: \norm{\omega}_\infty \leq N \}.
\end{equation}
One can choose the discrete Fourier modes to be all non-negative; however, centering this grid at the origin makes it possible to prove $\tilde f_N \to \hat f$ strongly in the $\ell^2$ norm as $N \to \infty$ under the regularity assumptions of \cref{prop:dft-distance-bound}. We also note that the DFT operator \eqref{eq:DFT-transform} and the $d$-fold QFT matrix are identical unitary matrices, and we have:
\begin{equation}
\ket{f_N}
:= \frac{1}{\sqrt{\llangle f^2 \rrangle_{\Gamma_N}}}\sum_{x\in \Gamma_N} f(x) \ket{x}
\xrightarrow{\quad F_N^{\otimes d} \quad}
|\tilde f_N\rangle:= \frac{1}{\sqrt{\langle f^2 \rangle_{\Gamma_N}}} \sum_{\omega\in \Lambda_N} \tilde f_\omega\,\ket{\omega}.
\end{equation}
Here and elsewhere, the notations $\langle - \rangle_{S}$ and $\llangle - \rrangle_S$ stand, respectively, for the average and the sum over a set $S$.

\begin{table}[t]
\footnotesize
\begin{center}
\begin{tabular}{l l r}
\toprule
$f(x)$
& $\displaystyle = \sum_{\omega\in\mathbb{Z}^d} \hat f_\omega \,
e^{i\langle \omega, x\rangle}$
& Fourier transform / series
\\[3.5ex]
$\hat f_\omega$
& $\displaystyle = \frac{1}{(2\pi)^d}\int_{\mathbb T^d}
f(x) e^{-i\langle \omega,x\rangle} \, dx
= \mathbb E_{\mathbb T^d} [f(x)e^{-i \langle \omega, x \rangle}]$
& Fourier coefficients
\\[4.2ex]
$\tilde f_\omega$
& $\displaystyle = \mathbb E_{\Gamma_N} [f(x) e^{-i \langle \omega,x\rangle}]$
& Discrete Fourier coefficients
\\[3.2ex]
$\lvert \tilde f_N \rangle$
& $\displaystyle = \frac{1}{\|\tilde f_N\|_2} \sum_{\omega \in \Lambda_N} \tilde f_\omega\, \ket{\omega}
=F_N^{\otimes d} \ket{f_N}$
& Discrete / quantum Fourier transform
\\[4ex]
$\tilde f_\omega$
& $\displaystyle = \sum_{m \in \mathbb Z^d} \hat f_{\omega + (2N+1)m}$
& Aliasing sums\\
\bottomrule
\end{tabular}
\end{center}
\caption{A summary of the main Fourier-analytic objects used throughout this work, for a $2\pi$-periodic function $f:\mathbb T^d \to \mathbb R$. From top to bottom: the Fourier series of $f$; its Fourier coefficients $\hat f_\omega$, written as the average of $f$ against the plane wave $e^{-i\langle\omega,x\rangle}$ over the torus; the discrete coefficients $\tilde f_\omega$, the same average over the finite grid $\Gamma_N$; the amplitude encoding $\ket{\tilde f_N}=F_N^{\otimes d}\ket{f_N}$ produced by the $d$-fold quantum Fourier transform; and the aliasing relation, which expresses each discrete coefficient as a sum of continuous coefficients at frequencies spaced by $2N+1$. We give full definitions and derivations in \cref{sec:fourier}.}
\label{tab:Fourier-summary}
\end{table}

We can better understand the role of differentiability as follows. We write $\hat f_N:= (\hat f_\omega)_{\omega \in \Lambda_N}$ and $\tilde f_N= (2N+1)^{-d/2} F_N^{\otimes d} (f_N)$ for the finitely supported vectors of the truncated Fourier and discrete Fourier coefficients of $f$ in the Hilbert space $\ell^2= \ell^2(\mathbb Z^d)$, respectively. The infinitely supported vector $\hat f$ also resides in $\ell^2$ when $f$ is itself in $L^2$ by Parseval's theorem. We note that in general the continuous and discrete Fourier coefficients are not identical. Instead, the former ``roll up'' to form the latter in the so-called \emph{aliasing sums},
\begin{equation}
\label{eq:aliasing-sum-intro}
\tilde f_\omega= \sum_{m \in \mathbb Z^d} \hat f_{\omega + (2N+1)m}.
\end{equation}
To ensure the approximation $\hat f_N \simeq \tilde f_N$, we require fast decay of the high-frequency Fourier coefficients, so that the $m=0$ term $\hat f_\omega$ dominates the right-hand sum in \eqref{eq:aliasing-sum-intro}. We refer the reader to \cref{sec:fourier} for a thorough account of this framework. A summary of the concepts introduced therein is provided in \cref{tab:Fourier-summary}.

Gevrey smoothness \cite{gevrey1918}, namely
\begin{equation}
\label{eq:gevrey-smoothness-definition}
\|D^\alpha f\|_\infty \leq C\, (\alpha!)^s / r^{|\alpha|}, \qquad \text{for all} \quad \alpha \in \mathbb Z_{\geq 0}^d,
\end{equation}
for some $s \geq 0$ is equivalent to exponentially fast decay of $|\hat f_\omega|$ as $\norm{\omega} \to \infty$ (\cref{prop:all-decay-rates}). This also implies that the entire \emph{tail} of the Fourier series, $\mathcal T_N := \|\hat f- \hat f_N\|_2$, decays exponentially fast in $N$ (\cref{prop:tail-decay-rate}). \cref{tab:fourier-tailedness} shows the tailedness of the Fourier expansion of various classes of functions. The less smooth $f$ is, the larger the tail of its Fourier expansion. For example, for $p$-smooth functions the tail decays only polynomially, $\mathcal T_N \in \mathcal O(N^{-p})$, and therefore a Fourier cutoff $\mathcal O((1/\epsilon)^{1/p})$ is required to guarantee a precision $\epsilon$ in approximating $f$.

For analytic functions (i.e., when $s = 1$), the derivative bound guarantees a holomorphic extension on the strip $\norm{\mathrm{Im}\, z}_\infty < r$. The parameter $C$ bounds the real-domain derivatives, while a contour estimate at height $\sigma < r$ uses the supremum $M(\sigma)$ over shifts in one coordinate at a time (\cref{rmk:sigma-to-r-subtlety}). This gives, in the univariate case, the classical exponential-accuracy behavior of Fourier and Chebyshev differencing of analytic functions \cite{tadmor1986exponential,gottlieb1997gibbs}, which \cref{prop:all-decay-rates} extends to the full Gevrey hierarchy and to arbitrary dimension $d$ (\cref{rmk:classical-analytic-precedent}). Beyond functions with complex singularities lie the \emph{entire} functions, whose decay rate is super-exponential. There, $r$ no longer measures a distance to a singularity, but instead the growth rate of the entire extension, while $C$ retains its role as an overall size constant. Each class in \cref{tab:fourier-tailedness} strictly contains the one below it. Notably, there are smooth functions that are not Gevrey for any $s$ (\cref{exm:smooth-not-gevrey}) and there are entire functions that are not Gevrey for any $s < 1$ (\cref{exm:entire-not-finite-order}).

\subsection{Discretization error and cutoff selection}

For any $2\pi$-periodic $d$-variate function $f$, and any derivative index $\alpha \in \mathbb Z^d_{\geq 0}$, the Fourier transform of $D^\alpha f$ has Fourier coefficients $(i\omega)^\alpha \hat f_\omega$. This means that differentiation is a diagonal operator in the Fourier domain:
\begin{equation}
\label{eq:diff-from-fourier-results}
D^\alpha f = \mathcal F^{-1} \diag\left\{(i\omega)^\alpha: \omega \in \mathbb Z^d\right\} \mathcal F (f)
\end{equation}
where $\mathcal F$ denotes the Fourier transform as an operator. It is therefore natural to ask whether the truncation of $\mathcal F$ and the diagonal operator up to a cutoff mode $N$ approximate $D^\alpha$ well. Therefore, we define the approximate derivative operator
\begin{equation}
\label{eq:approx-derivative-op-results}
\tilde D^\alpha_N = F_N^{\dagger \otimes d} \diag\left\{(i\omega)^\alpha: \omega \in \Lambda_N\right\} F_N^{\otimes d}
\end{equation}
and seek to bound the distance between the approximate derivative applied to the discretization of $f$ and the discretization of the true derivative: $\norm{\tilde D^\alpha_N f_N - (D^\alpha f)_N}_2$.

\begin{thm}
\label{thm:fourier-approximation}
Let $f$ be a $2\pi$-periodic real-valued $d$-variate function. If $f$ is merely continuous, then $\hat f_\omega \to 0$ as $\norm{\omega}_\infty \to \infty$. If $f$ is moreover from one of the regularity classes of \cref{tab:fourier-tailedness} then $\hat f_\omega$ has the decay shown in the second column, in the regimes of \cref{prop:all-decay-rates}. The truncation error satisfies $\mathcal T_N := \norm{\hat f - \hat f_N}_2 \leq \bar{\mathcal T}_N$ in the regimes of \cref{prop:tail-decay-rate}, where $\bar{\mathcal T}_N$ is the bound in the third column of the same table.
\end{thm}

The proofs are provided in \cref{prop:all-decay-rates} and \cref{prop:tail-decay-rate}. It follows immediately that for all of the regularity classes in \cref{tab:fourier-tailedness} an $\epsilon$-accurate truncated Fourier expansion $\hat f_N$ exists when $N$ is chosen according to the fourth column of this table. However, the same observation does not hold when comparing the truncated QFT (or discrete Fourier transform) and $\hat f$.

\begin{thm}
\label{thm:high-precision-dft}
Let $f$ be a continuous $2\pi$-periodic real-valued $d$-variate function from a regularity class in \cref{tab:fourier-tailedness}, with $N \geq 1$. For the Gevrey rows with parameters $C$ and $r$, in the respective cutoff regimes of \cref{prop:dft-distance-bound,prop:dft-derivative-bound}, we have
\begin{equation}
\norm{\tilde f_N - \hat f}_2 \leq 2^{d/2 + 2}\bar{\mathcal T}_N.
\end{equation}
For any multi-index $\alpha \in \mathbb Z^d_{\geq 0}$ with $D^\alpha f$ continuous, we also have
\begin{equation}
\norm{\tilde D_N^\alpha f_N - (D^\alpha f)_N}_2
\leq 2^{(d + 5)/2}(2N + 1)^{d/2}N^{|\alpha|}\bar{\mathcal T}_N.
\end{equation}
For the H\"older rows, these bounds hold under the respective assumptions $p + q \geq 2d$ and $p + q \geq 2d + \frac32|\alpha|$, with $q = 0$ in the $p$-smooth row. For finite bandwidth, both errors are zero once $N \geq 1/r$.
\end{thm}

These results are proven in \cref{prop:dft-distance-bound,prop:dft-derivative-bound}, the latter also giving a grid-uniform estimate. The bounds use the tabulated $\bar{\mathcal T}_N$, not the actual tail $\mathcal T_N$. Sufficient H\"older smoothness gives algebraic convergence, while Gevrey smoothness gives exponential accuracy (see \cref{rmk:smooth-dft-no-go}).

\begin{thm}
\label{thm:high-precision-dft-gevrey}
Let $f$ be an $s$-Gevrey, $2\pi$-periodic, and $d$-variate real-valued function with Gevrey parameters $C$ and $r$, and $s \geq 1$. We have
\begin{align}
\norm{\tilde{f}_{N} - \hat{f}}_2
\leq C e^{-\frac12(rN)^{1/s}},
\end{align}
provided that $N \in \tilde \Omega((ds)^s/r)$. Moreover, for any multi-index $\alpha \in \mathbb Z^d_{\geq 0}$ we have
\begin{align}
\norm{\tilde D^\alpha_N f_N - (D^\alpha f)_N}_2
\leq C e^{-\frac12(rN)^{1/s}},
\end{align}
provided that $N \in \tilde \Omega((s(d + |\alpha|))^s/r)$.
\end{thm}

\begin{rmk}
For $0 < s < 1$, the same arguments give both error bounds with $Ce^{-\frac{s}{4}(rN)^{1/s}}$ in place of $Ce^{-\frac12(rN)^{1/s}}$. The cutoffs use $r \mapsto (s/2)^s r$, with $ds$ replaced by $d$ in \eqref{eq:N-bound-for-good-gevrey-fourier} and $s(d + |\alpha|)$ replaced by $\max(d, s(d + |\alpha|))$ in \eqref{eq:N-bound-for-good-gevrey-fourier-diff}. This follows from Case~8 of \cref{prop:all-decay-rates} and the remarks following the Fourier error corollaries.
\end{rmk}

These results are proven in \cref{cor:qft-func-err,cor:qft-diff-err}.

Since we are frequently concerned about the distance between unit vectors, we note that the following inequality holds \cite[Lemma 13]{berry2017quantum}:
\begin{equation}
\label{eq:distance-scaling}
\norm{\frac{a}{\norm{a}}- \frac{b}{\norm{b}}}\leq \frac{2\norm{a-b}}{\max(\norm{a}, \norm{b})}.
\end{equation}
For example, since $F_N$ is a unitary transformation, we may be interested in its action on the normalized amplitude-encoding states $\ket{f_N}$. The following statement follows immediately from \cref{thm:high-precision-dft-gevrey}, \eqref{eq:distance-scaling}, and Parseval's theorem for the Fourier and discrete Fourier transforms.

\begin{cor}
\label{cor:high-precision-qft}
Let $f$ be an $s$-Gevrey, $2\pi$-periodic, and $d$-variate real-valued function with $s \geq 1$. With the Fourier output embedded into $\ell^2(\mathbb Z^d)$ by zero extension, we have
\begin{align}
\norm{F_{N}^{\otimes d}\ket{f_{N}} - \ket{\hat f}}
\leq \frac{2C}{\sqrt{\max(\langle f^2 \rangle_{\Gamma_N}, \langle f^2 \rangle_{\mathbb T^d})}} e^{-\frac12(rN)^{1/s}},
\end{align}
provided that $N \in \tilde \Omega((ds)^s/r)$. Moreover, for any multi-index $\alpha \in \mathbb Z^d_{\geq 0}$ we have
\begin{align}
\norm{\ket{\tilde D^\alpha_N f_{N}} - \ket{(D^\alpha f)_N}}
\leq \frac{2C}{\sqrt{\llangle (D^\alpha f)^2 \rrangle_{\Gamma_N}}} e^{-\frac12(rN)^{1/s}},
\end{align}
provided that $N \in \tilde \Omega((s(d + |\alpha|))^s/r)$.
\end{cor}

\section{Linear PDEs with Periodic Boundary Conditions}
\label{sec:periodic-boundary-conditions}

Consider a general linear PDE with periodic boundary conditions in the form:
\begin{equation}
\label{eq:linear-PDE}
\sum_{\alpha\in \mathcal J}g_{\alpha}(x)D^{\alpha}u(x)= \eta(x).
\end{equation}
\begin{sloppypar}
\noindent Here the finite index set $\mathcal J \subset \mathbb{Z}^d_{\geq 0}$ determines the differential operators $D^\alpha$, including $D^0 = I$ for a zeroth-order term. The function $\eta$ is called the \emph{source} term in what follows. We assume all $g_\alpha$, $\eta$, and the solution $u$ are $2\pi$-periodic. Moreover, we assume oracle access to the coefficient functions $g_\alpha$ through a family of block encodings $\{U_{g_{\alpha,N}}\}_{N \in \mathbb N}$ of the diagonal matrices $\diag\{g_{\alpha,N}\}$ (up to normalization factors), one for each truncation threshold $N$, and to the discretized source function $\eta_N$ through a family of state-preparation oracles $\{U_{\eta_N}\}_{N \in \mathbb N}$ satisfying $U_{\eta_N}\ket0=\ket{\eta_N}$. Controlled versions and inverses are included in this access model. If we are only given access to bit-oracles of the form
\begin{equation}
\label{eq:pde-coeff-oracle}
O_{g_\alpha}:\ket x\ket 0^{\otimes b} \mapsto \ket x \ket{g_\alpha (x)},
\end{equation}
they can be used to construct block encodings $U_{g_\alpha}$ with a target accuracy $\delta > 0$ using $\mathcal O(\log^{2.5}(1/\delta)+d\log N)$ elementary gates\footnote{We take arbitrary single-qubit rotations and CNOT gates as elementary gates.}, $\mathcal O(\log^{2.5}(1/\delta))$ ancilla qubits, and two queries to $O_{g_{\alpha}}$ as shown in \cite[Lemma 48]{gilyen2019quantum}.
\end{sloppypar}

In \cref{sec:choosing-N} we show how the truncation threshold $N$ for the norms of the Fourier modes controls the discretization error. Given $N$, we consider the grid discretizations $\Gamma_N$ and $\Lambda_N$ of the spatial and frequency domains, respectively, as introduced earlier in \eqref{eq:spatial-grid} and \eqref{eq:frequency-grid}. We seek to prepare an approximation of the normalized state $\ket{u_N} = \norm{u_N}_2^{-1}\sum_{x \in \Gamma_N}u(x)\ket{x}$ by discretizing the operator
\begin{equation}
\label{eq:linear-PDE-op}
\mathcal L = \sum_{\alpha\in \mathcal J}g_{\alpha}(x)D^{\alpha}
\end{equation}
as the matrix $\mathbb L_N$ defined in \eqref{eq:def-discrete-L}, and solving
\begin{equation}
\label{eq:linear-PDE-discretized}
\mathbb L_N\ket{u_N} \propto \ket{\eta_N}.
\end{equation}

For homogeneous linear PDEs, i.e., when $\eta = 0$, the equation can be solved on a quantum computer via quantum singular value thresholding \cite{leng2025operator}. However, in this paper we focus on the inhomogeneous case ($\eta \neq 0$) and tackle it using a state-of-the-art quantum linear system algorithm (QLSA) \cite{costa2022optimal, tong2021fast, dalzell2024shortcut, low2026quantum}. Each of these algorithms has slightly different assumptions and trade-offs, so for simplicity we choose the approach of \cite{dalzell2024shortcut} in our exposition, with query complexity $\mathcal O(\kappa \log(1/\epsilon))$ to both the block encoding of $\mathbb L_N$ and the state-preparation oracle for $\ket{\eta_N}$. Here $\kappa = \sigma_{\max} / \sigma_{\min}^+$ is the effective condition number of $\mathbb L_N$ restricted to the orthogonal complement of its kernel, where $\sigma_{\min}^+$ denotes its smallest nonzero singular value; \cite{dalzell2024shortcut} is defined precisely for such matrices with a nontrivial kernel, returning the minimum-Euclidean-norm solution, i.e., the Moore--Penrose pseudo-inverse solution $\mathbb L_N^+ \eta_N$, whenever $\eta_N$ has no component in $\ker(\mathbb L_N^T)$. More precisely, the QLSA's query complexity depends on $\alpha/\sigma_{\min}^+$ for the normalization factor $\alpha$ achieved by the block encoding. Our LCU construction satisfies $\alpha \leq |\mathcal J|\ceil{g}N^{\ceil{\mathcal J}}$ (see the proof of \cref{thm:linear-pde-complexity}).

The relationship between $\sigma_{\min}^+$ and the singular values of $\mathcal L$ is complicated so we confine the statement of our results to be in terms of $\sigma_{\min}^+$, which is what the applications of \cref{sec:poisson} require. Even if the continuous operator $\mathcal L$ is gapped, more technical conditions such as ellipticity are required so that $\sigma_{\min}^+$ approaches the smallest positive singular value of $\mathcal L$ as $N \to \infty$ (see \cref{sec:discussion} for further discussions). Note that \cite{dalzell2024shortcut} is not necessarily optimal in queries to the source state oracle (see \cite[Theorem 33]{tong2021fast} and \cite{low2026quantum}). This can be consequential whenever the state-preparation oracle for the QLSA is itself an expensive subroutine. The end-to-end PDE solver is summarized in \cref{alg:linear-pde}.

\begin{algorithm}[H]
\caption{Quantum algorithm for linear PDEs}
\label{alg:linear-pde}
\begin{algorithmic}[1]

\Require \,
Target precision $\epsilon > 0$; Families of block encodings $\{U_{g_{\alpha, N}}\}_{N \in \mathbb N}$; source state oracles $\{U_{\eta_N}\}_{N \in \mathbb N}$;\newline
\hspace*{7.3mm} Parameter $s > 0$ and the promise that \eqref{eq:linear-PDE} has an $s$-Gevrey minimum-norm solution with parameters $C$ and $r$; The discrete target-orthogonality hypothesis of \cref{thm:linear-pde-complexity} at the chosen cutoff.

\State Choose a power-of-two cutoff $N \geq 2$ with discretization error at most $\epsilon/2$ according to \cref{cor:ultimate-N} for $s \geq 1$, or according to \cref{rmk:ultimate-N-small-s} for $0 < s < 1$.
\State Prepare the block encoding of $\mathbb L_N$ defined in \eqref{eq:def-discrete-L}.
\State Apply the QLSA of \cite{dalzell2024shortcut} to obtain a state $\ket{v}$ within $\epsilon/2$ of $\ket{\mathbb{L}^{+}_N\eta_N}$, including errors in state preparation and block encoding.

\Ensure $N \in \mathbb{N}$, and $\ket{v} \in \mathbb{C}^{(2N + 1)^d}$ such that $\norm{\ket{v} - \ket{u_N}} \leq \epsilon$, with $u$ the minimum-norm solution~to~\eqref{eq:linear-PDE}.
\end{algorithmic}
\end{algorithm}

Throughout, the target precision $\epsilon$ in \cref{alg:linear-pde} is meant to bound the total error accumulated from every source of approximation: the Fourier truncation error from choosing a finite $N$, the accuracy of the state-preparation oracle $U_{\eta_N}$, the accuracy of the block encodings $U_{g_{\alpha,N}}$, and the QLSA's own inversion error.

\subsection{Choosing the Fourier cutoff}
\label{sec:choosing-N}

We define the QFT approximation of $\mathcal L$ by
\begin{equation}
\label{eq:def-discrete-L}
\mathbb L_N= \sum_\alpha \diag\{g_{\alpha,N}\} \tilde D_N^\alpha,
\end{equation}
which is accurate for $s$-Gevrey functions when $N$ is sufficiently large as stated in the next corollary. In the rest of this exposition we use the notation
\begin{equation}
\ceil{g}= \max_x \left(\sum_\alpha |g_\alpha(x)|\right),
\quad \text{ and } \quad
\ceil{\mathcal J} = \max_{\alpha \in \mathcal J} |\alpha|,
\end{equation}
respectively for bounding the coefficients of the linear PDE in \eqref{eq:linear-PDE}, and the order of it.

\begin{cor}
\label{cor:approx-linear-pde-op}
Given the linear differential operator $\mathcal L = \sum_{\alpha\in \mathcal J}g_{\alpha}D^{\alpha}$ we have
\begin{equation}
\norm{\mathbb L_N u_N - (\mathcal L u)_N}
\leq C \ceil{g} e^{-\frac12(rN)^{1/s}},
\end{equation}
provided that $u$ is $s$-Gevrey with $s \geq 1$ and $N \in \tilde \Omega((s(d + \ceil{\mathcal J}))^s/r)$.
\end{cor}

\begin{proof}
We have $(g_\alpha D^\alpha u)_N = \diag\{g_{\alpha, N}\}(D^\alpha u)_N$. By the grid-uniform estimate in \cref{cor:qft-diff-err},
\begin{align}
\norm{\mathbb L_Nu_N - (\mathcal Lu)_N}_2
& \leq (2N + 1)^{d/2}\max_{x \in \Gamma_N}
\sum_{\alpha \in \mathcal J}|g_\alpha(x)|\, |\tilde D_N^\alpha u_N(x) - D^\alpha u(x)| \nonumber\\
& \leq \ceil{g}\, C e^{-\frac12(rN)^{1/s}},
\end{align}
proving the claim.
\end{proof}

\begin{thm}
\label{thm:linear-pde-precision}
Suppose that the minimum-norm solution $u$ to the inhomogeneous linear PDE $\mathcal Lu=\eta$ is $s$-Gevrey with $s \geq 1$ and convergence radius $r > 0$, and that $u_N \perp \ker(\mathbb L_N)$. Then, given $v_N= \mathbb L_N^+\eta_N$, we have
\begin{equation}
\norm{u_N - v_N}_2 \leq \frac{C \ceil{g}}{\sigma_{\min}^+(\mathbb L_N)} e^{-\frac12(rN)^{1/s}}
\end{equation}
provided that $N \in \tilde \Omega((s(d + \ceil{\mathcal J}))^s/r)$.
\end{thm}

\begin{proof}
First, let $ \xi_N$ be a vector of discrepancies from applying the discretized operator on the true discretized solution:
\begin{equation}
\mathbb L_N u_N= \eta_N+ \xi_N.
\end{equation}
Then since $u$ is the true solution, $(\mathcal Lu)_N = \eta_N$, and therefore by \cref{cor:approx-linear-pde-op}
\begin{equation}
\norm{\xi_N}= \norm{\mathbb L_N u_N - (\mathcal L u)_N}
\leq C \ceil{g} e^{-\frac12(rN)^{1/s}}.
\end{equation}
Since $\mathbb L_N^+\mathbb L_N$ projects onto $\ker(\mathbb L_N)^\perp$ and $u_N$ lies in this subspace,
\begin{equation}
u_N- v_N= \mathbb L_N^+(\mathbb L_Nu_N-\eta_N)= \mathbb L_N^+\xi_N.
\end{equation}
Therefore,
\begin{equation}
\norm{u_N- v_N}
= \norm{\mathbb L_N^{+}\xi_N}
\leq \frac{\norm{\xi_N}}{\sigma_{\min}^+(\mathbb L_N)},
\end{equation}
completing the proof.
\end{proof}

\begin{rmk}
Without the orthogonality assumption on the discretized solution, the same calculation gives
\begin{equation}
u_N-\mathbb L_N^+\eta_N= \mathbb L_N^+\xi_N + P_{\ker(\mathbb L_N)}u_N.
\end{equation}
Thus the vector-error bound acquires the additional term $\norm{P_{\ker(\mathbb L_N)}u_N}$.
\end{rmk}

For solutions to linear PDEs, \eqref{eq:distance-scaling} is useful for bounding the distance $\norm{\ket{u_N} - \ket{v_N}}$ in the notation of \cref{thm:linear-pde-precision}:
\begin{equation}
\label{eq:distance-bound-pde}
\norm{\ket{u_N} - \ket{v_N}}_2 \leq \frac{2C \ceil{g}}{\sigma_{\min}^+(\mathbb L_N) \sqrt{\llangle u^2 \rrangle_{\Gamma_N}}} e^{-\frac12(rN)^{1/s}}.
\end{equation}
This results in the following statement.

\begin{cor}
\label{cor:ultimate-N}
Suppose that the minimum-norm solution $u$ to the inhomogeneous linear PDE $\mathcal Lu=\eta$ is $s$-Gevrey with $s \geq 1$ and parameters $r$ and $C$, and that $u_N \perp \ker(\mathbb L_N)$. With $w = s(d + \ceil{\mathcal J})$, an $\epsilon > 0$ discretization error in the normalized solution is guaranteed by any positive integer $N$ satisfying
\begin{equation}
\label{eq:PDE-truncation-bound}
N \geq \frac{2^s}{r}\max\Biggl\{\log\frac{2\sqrt2\, C\ceil{g}}{\epsilon\, \sigma_{\min}^+(2N + 1)^{d/2}\norm{\hat u}}, \log\frac{4C}{\norm{\hat u}}, 32w\log\!\left(\frac{64w}{r^{1/s}}\right)\Biggr\}^s.
\end{equation}
\end{cor}

\begin{proof}
We first eliminate the dependence on the discretized-solution norm $\sqrt{\llangle u^2 \rrangle_{\Gamma_N}}$ appearing in \eqref{eq:distance-bound-pde}, replacing it with the norm $\norm{\hat u}$ of the true solution's Fourier coefficients $\hat u= (\hat u_\omega)_{\omega \in \mathbb Z^d}$, using the Parseval identity $\langle u^2 \rangle_{\mathbb T^d}= \norm{\hat u}^2$.

By Parseval's identity for the discrete Fourier transform, $\sqrt{\langle u^2\rangle_{\Gamma_N}} = \norm{\tilde u_N}_2$. The reverse triangle inequality and \cref{cor:qft-func-err} therefore yield
\begin{equation}
\label{eq:u2-average-error}
\left|\sqrt{\langle u^2\rangle_{\Gamma_N}} - \norm{\hat u}_2\right|
\leq \norm{\tilde u_N - \hat u}_2 \leq Ce^{-\frac12(rN)^{1/s}}.
\end{equation}
The second term of \eqref{eq:PDE-truncation-bound} makes this error at most $\norm{\hat u}/4$, while its last term ensures the validity of Fourier errors as per \eqref{eq:N-bound-for-good-gevrey-fourier} and \eqref{eq:N-bound-for-good-gevrey-fourier-diff}. Thus $\sqrt{\langle u^2\rangle_{\Gamma_N}} \geq 3\norm{\hat u}/4 \geq \norm{\hat u}/\sqrt2$, and hence
\begin{equation}
\label{eq:llangle-lower-bound}
\sqrt{\llangle u^2 \rrangle_{\Gamma_N}} = (2N+1)^{d/2} \sqrt{\langle u^2 \rangle_{\Gamma_N}} \geq \frac1{\sqrt2} (2N+1)^{d/2} \norm{\hat u}.
\end{equation}

Substituting \eqref{eq:llangle-lower-bound} into \eqref{eq:distance-bound-pde} yields
\begin{equation}
\label{eq:distance-bound-pde-continuous}
\norm{\ket{u_N} - \ket{v_N}}_2 \leq \frac{2\sqrt2\, C \ceil{g}}{\sigma_{\min}^+(\mathbb L_N)\, (2N + 1)^{d/2} \norm{\hat u}}\, e^{-\frac12(rN)^{1/s}},
\end{equation}
whose right-hand side no longer references the discretized solution $u_N$, but only the norm $\norm{\hat u}$ of the true solution's Fourier coefficients. Requiring the right-hand side of \eqref{eq:distance-bound-pde-continuous} to be at most $\epsilon$ gives the first term of \eqref{eq:PDE-truncation-bound}, completing the proof.
\end{proof}

\begin{rmk}
\label{rmk:ultimate-N-u}
Since $\langle u^2 \rangle_{\mathbb T^d} = \frac1{(2\pi)^d}\int_{\mathbb T^d} u(x)^2\, dx$, writing $\norm{u} := \left(\int_{\mathbb T^d} u(x)^2\, dx\right)^{1/2}$ for the standard $L^2(\mathbb T^d)$ norm of the solution in the original (spatial) domain rather than the Fourier domain, we have $\norm{\hat u} = \norm{u}/(2\pi)^{d/2}$. Substituting throughout, \eqref{eq:PDE-truncation-bound} is equivalently stated, for readers with more natural access to the norm of their solution in the spatial domain, as
\begin{equation}
\label{eq:PDE-truncation-bound-u}
N \geq \frac{2^s}{r}\max\Biggl\{\log\frac{2\sqrt2\, C\ceil{g}(2\pi)^{d/2}}{\epsilon\, \sigma_{\min}^+(2N + 1)^{d/2}\norm{u}}, \log\frac{4C(2\pi)^{d/2}}{\norm{u}}, 32w\log\!\left(\frac{64w}{r^{1/s}}\right)\Biggr\}^s.
\end{equation}
\end{rmk}

\begin{rmk}
\label{rmk:ultimate-N-small-s}
For $0 < s < 1$, the preceding normalized-state, PDE residual, and discretization bounds hold with $e^{-\frac{s}{4}(rN)^{1/s}}$ in place of $e^{-\frac12(rN)^{1/s}}$. Their cutoff conditions use $r \mapsto (s/2)^s r$ and the maximum-based weights specified in the $0 < s < 1$ extensions of the Fourier error corollaries. The original Gevrey parameters $C$ and $r$ remain unchanged in the error bounds.
\end{rmk}

\begin{rmk}
\label{rmk:complex-valued}
The Fourier and PDE error estimates also hold for complex-valued functions under the same regularity and orthogonality assumptions. The proofs use absolute values, Hermitian inner products, and Parseval's identity, so the same $C$ and $r$ apply without splitting into real and imaginary parts. In the norm formulas we replace $f^2$, $u^2$, and $(D^\alpha f)^2$ by $|f|^2$, $|u|^2$, and $|D^\alpha f|^2$, respectively. This extension is needed for applications such as the Schr\"odinger equation and other quantum-mechanical problems involving complex-valued wavefunctions.
\end{rmk}

\subsection{Block encoding of differential operators}

For executing \cref{alg:linear-pde}, the matrix $\mathbb L_N$ in \eqref{eq:def-discrete-L} is constructed by creating block encodings for each summand $\diag\{g_{\alpha,N}\} \tilde D_N^\alpha$ in \eqref{eq:def-discrete-L}, and performing a linear combination of unitaries (LCU) to construct
\begin{equation}
\label{eq:lcu-of-L}
\sum_{\alpha\in\mathcal J} \ketbra{c_\alpha}{c_\alpha} \otimes \diag\{g_{\alpha,N}\} \tilde D_N^\alpha
\end{equation}
and obtain $\mathbb L_N$ from uncomputing the first register, then inverting $\mathbb L_N$. Here, the first register stores an integer counter $c_\alpha$ that enumerates the elements $\alpha \in \mathcal J$. We do this in two steps by re-writing \eqref{eq:lcu-of-L} as
\begin{align}
\label{eq:lcu-two-stage}
    \left(\sum_{\alpha\in \mathcal J} \ketbra{c_\alpha}{c_\alpha}\otimes \diag\{g_{\alpha,N}\}\right)\left(\sum_{\alpha\in \mathcal J} \ketbra{c_\alpha}{c_\alpha}\otimes \tilde D^\alpha_N \right).
\end{align}
We provide a custom encoding for the differential term (the second term), but implement a conventional LCU for the diagonal term (the first term) using $\mathcal O(1)$ queries to each $U_{g_{\alpha,N}}$. The full circuit for implementing \eqref{eq:lcu-two-stage} is provided in \cref{fig:L-circuit}.

\begin{figure}
\centering
\begin{adjustbox}{max width = \textwidth}

\begin{tikzpicture}[
    thick,
    smallgate/.style={draw, minimum size=0.6cm, inner sep=1pt, font=\small, fill=white},
    bigbox/.style={draw, minimum width=2.2cm, minimum height=2.0cm, inner sep=4pt, font=\large, fill=white},
    tallbox/.style={draw, minimum width=2.2cm, inner sep=4pt, font=\large, fill=white},
    ugate/.style={draw, minimum width=1.2cm, minimum height=1.8cm, inner sep=2pt, font=\small, fill=white},
    metergate/.style={draw, minimum width=0.55cm, minimum height=0.55cm, inner sep=0pt, fill=white,
        path picture={
            \draw ([shift={(0.05,0.05)}]path picture bounding box.south west) 
                  to[out=50,in=130] 
                  ([shift={(-0.05,0.05)}]path picture bounding box.south east);
            \draw ([yshift=0.05cm]path picture bounding box.south) -- ++(60:0.22cm);
        }},
    wdots/.style={fill=white, inner sep=3pt, font=\small, outer sep=0pt},
    control/.style={circle, fill=black, minimum size=5pt, inner sep=0pt},
    every node/.style={font=\small},
]

\def\ytop{0}         
\def\yalpha{-1.2}    
\def\yalphabot{-2.4} 
\def\yx{-4.0}        
\def\yxbot{-5.2}     
\def\yanc{-6}      


\pgfmathsetmacro{\xH}{2.0}

\pgfmathsetmacro{\xQROM}{4.5}
\pgfmathsetmacro{\QROMw}{2.2} 

\pgfmathsetmacro{\xAlphaOut}{\xQROM+1.0}

\pgfmathsetmacro{\xSum}{8.5}
\pgfmathsetmacro{\Sumw}{2.4}

\pgfmathsetmacro{\xQROMd}{12.5}

\pgfmathsetmacro{\xMeasAlpha}{15}

\pgfmathsetmacro{\xMeasAnc}{10.7}

\pgfmathsetmacro{\xUa}{16}
\pgfmathsetmacro{\xUb}{18}
\pgfmathsetmacro{\xUc}{21}

\pgfmathsetmacro{\xHend}{23.0}

\pgfmathsetmacro{\xMeasEnd}{24.5}

\pgfmathsetmacro{\wend}{25.6}


\draw (0.8, \ytop) -- (\xMeasEnd, \ytop);

\draw (\xAlphaOut, \yalpha) -- (\xMeasAlpha+0.3, \yalpha);
\draw (\xAlphaOut, \yalphabot) -- (\xMeasAlpha+0.3, \yalphabot);

\pgfmathsetmacro{\alphaInStart}{\xQROM-1.1}
\draw (\alphaInStart-0.8, \yalpha) -- (\alphaInStart, \yalpha);
\draw (\alphaInStart-0.8, \yalphabot) -- (\alphaInStart, \yalphabot);

\draw (-0.2, \yx) -- (\wend, \yx);
\draw (-0.2, \yxbot) -- (\wend, \yxbot);

\pgfmathsetmacro{\ancStart}{\xSum-1.2}
\draw (\ancStart-1.0, \yanc) -- (\xMeasAnc+0.3, \yanc);

\node[anchor=east, font=\small] at (\ancStart-1.0, \yanc) {$\ket{0}$};


\node[anchor=east] at (0.85, \ytop) {$\ket{0}$};

\node[anchor=east] at (\alphaInStart-0.8, \yalpha) {$\ket{0}$};
\node at (\alphaInStart-0.8, {(\yalpha+\yalphabot)/2 + 0.1}) {$\vdots$};
\node[anchor=east] at (\alphaInStart-0.8, \yalphabot) {$\ket{0}$};

\node[anchor=south west, font=\scriptsize] at (\xAlphaOut+0.3, \yalpha+0.05) {$\ket{\bar{\alpha}_1}$};
\node at (\xAlphaOut+0.7, {(\yalpha+\yalphabot)/2 + 0.1}) {$\vdots$};
\node[anchor=north west, font=\scriptsize] at (\xAlphaOut+0.3, \yalphabot-0.05) {$\ket{\bar{\alpha}_d}$};

\node[anchor=east] at (-0.2, \yx) {$\ket{x_1}$};
\node at (-0.6, {(\yx+\yxbot)/2 + 0.1}) {$\vdots$};
\node at (\wend, {(\yx+\yxbot)/2 + 0.1}) {$\vdots$};
\node[anchor=east] at (-0.2, \yxbot) {$\ket{x_d}$};


\node[draw, minimum width=1.3cm, minimum height=0.9cm, inner sep=1pt, font=\small, fill=white] at (\xH, \ytop) {Prep};

\pgfmathsetmacro{\QROMtop}{\ytop+0.5}
\pgfmathsetmacro{\QROMbot}{\yalphabot-0.5}
\pgfmathsetmacro{\QROMleft}{\xQROM-1.1}
\pgfmathsetmacro{\QROMright}{\xQROM+1.1}
\node[draw, fill=white, minimum width=2.2cm, minimum height={(\QROMtop-\QROMbot)*1cm}, 
      inner sep=0pt, font=\Large] 
      at (\xQROM, {(\QROMtop+\QROMbot)/2}) {QROM};

\pgfmathsetmacro{\Sumtop}{\yalpha+0.3}
\pgfmathsetmacro{\Sumbot}{\yanc-0.3}
\node[draw, fill=white, minimum width=3.0cm, minimum height={(\Sumtop-\Sumbot)*1cm},
      inner sep=4pt, font=\normalsize, align=center]
      at (\xSum, {(\Sumtop+\Sumbot)/2}) {$\displaystyle\sum\ketbra{\bar\alpha}{\bar\alpha}\otimes\tilde D^\alpha_N$};

\pgfmathsetmacro{\QROMdtop}{\ytop+0.5}
\pgfmathsetmacro{\QROMdbot}{\yalphabot-0.5}
\node[draw, fill=white, minimum width=2.4cm, minimum height={(\QROMdtop-\QROMdbot)*1cm},
      inner sep=0pt, font=\Large]
      at (\xQROMd, {(\QROMdtop+\QROMdbot)/2}) {$\text{QROM}^\dagger$};

\node[metergate] at (\xMeasAlpha, \yalpha) {};
\node at (\xMeasAlpha, {(\yalpha+\yalphabot)/2 + 0.1}) {$\vdots$};
\node[metergate] at (\xMeasAlpha, \yalphabot) {};

\node[metergate] at (\xMeasAnc, \yanc) {};

\node[control] (ctrl1) at (\xUa, \ytop) {};
\draw (ctrl1) -- (\xUa, {\yx+0.25});
\node[ugate] at (\xUa, {(\yx+\yxbot)/2}) {$U_{g_{1,N}}$};

\node[control] (ctrl2) at (\xUb, \ytop) {};
\draw (ctrl2) -- (\xUb, {\yx+0.25});
\node[ugate] at (\xUb, {(\yx+\yxbot)/2}) {$U_{g_{2,N}}$};

\node[control] (ctrl3) at (\xUc, \ytop) {};
\draw (ctrl3) -- (\xUc, {\yx+0.25});
\node[ugate] at (\xUc, {(\yx+\yxbot)/2}) {$U_{g_{|\mathcal J|,N}}$};

\pgfmathsetmacro{\xUdots}{(\xUb+\xUc)/2}
\node[wdots] at (\xUdots, \yx) {$\cdots$};
\node[wdots] at (\xUdots, \yxbot) {$\cdots$};
\node[wdots] at (\xUdots, \ytop) {$\cdots$};

\node[draw, minimum width=1.3cm, minimum height=0.9cm, inner sep=1pt, font=\small, fill=white] at (\xHend, \ytop) {Prep$^\dagger$};

\node[metergate] at (\xMeasEnd, \ytop) {};

\pgfmathsetmacro{\LNleft}{0.0}
\pgfmathsetmacro{\LNright}{\wend-0.5}
\pgfmathsetmacro{\LNtop}{\ytop+1.0}
\pgfmathsetmacro{\LNbot}{\yanc-0.8}

\draw[dashed] (\LNleft, \LNtop) rectangle (\LNright, \LNbot);
\node[font=\Large] at ({(\LNleft+\LNright)/2}, \LNbot-0.6) {$\mathbb L_N$};

\end{tikzpicture}

\end{adjustbox}
\caption{The circuit for implementing $\mathbb L_N$ using the LCU in \eqref{eq:lcu-two-stage}. The \texttt{Prep} and \texttt{Prep}$^\dagger$ boxes prepare and un-prepare $\sum_{\alpha\in\mathcal J}\sqrt{\beta_\alpha/\norm{\beta}_1}\ket{c_\alpha}$, where $c_\alpha$ is the address and $\beta_\alpha := \max_x|g_\alpha(x)|\,N^{|\alpha|}$ is the weight. The normalized derivative blocks are implemented by \cref{fig:D-alpha-circuit}. Block-encoding workspaces are omitted.}
\label{fig:L-circuit}
\end{figure}
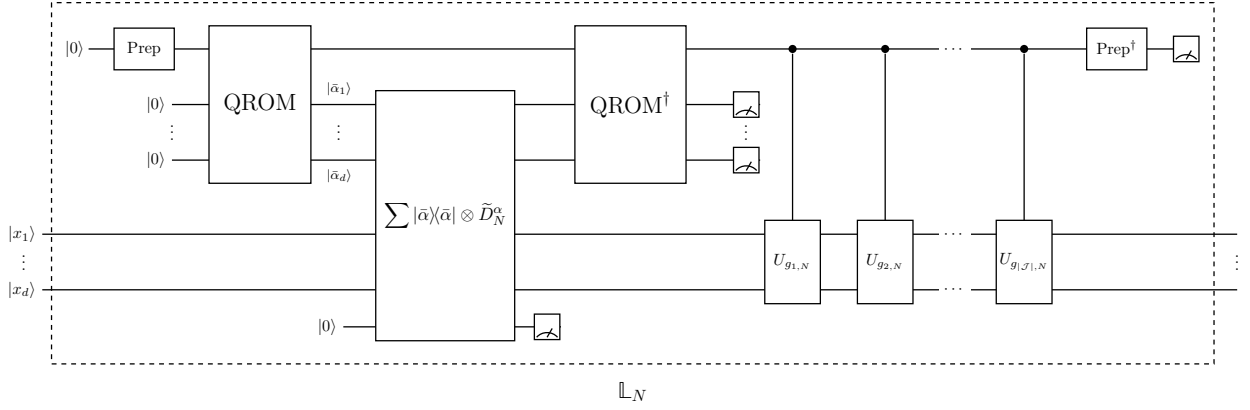

First, we note that $\lceil\log|\mathcal J|\rceil$ qubits suffice for the index register. The \texttt{Prep} circuit uses $\mathcal O(|\mathcal J|)$ rotations to prepare $\sum_{\alpha\in\mathcal J}\sqrt{\beta_\alpha/\norm{\beta}_1}\ket{c_\alpha}$, where $\beta_\alpha := \max_x|g_\alpha(x)|\,N^{|\alpha|}$ is the normalization factor of the $\alpha$-th term's block encoding. This factor results from composing the coefficient block $U_{g_{\alpha,N}}$ with the normalized Fourier multipliers $\prod_i (i\omega_i/N)^{\alpha_i}$ (\cref{app:R_N-calculation}). For each multi-index $\alpha$ we use $\bar \alpha$ to denote an associated $d$-tuple we name the \emph{sum-index} defined as follows.
\begin{equation}
\label{eq:sum-index}
\alpha= (\alpha_1, \dots, \alpha_d) \mapsto \bar \alpha:= (\alpha_1, \alpha_1 + \alpha_2, \dots, \sum_{i=1}^d \alpha_i)
\end{equation}
A multi-index or a sum-index register can be encoded using $d\lceil\log(1+\ceil{\mathcal J})\rceil$ qubits. The QROM in \cref{fig:L-circuit} uses $\mathcal O(d|\mathcal J|\log|\mathcal J|\log\ceil{\mathcal J})$ elementary gates to append $\ket{\bar\alpha}$ to the address $\ket{c_\alpha}$, preserving its amplitude.

Recall that each differential operator $\tilde D^\alpha$ is
\begin{align}
\label{eqn:approximate-differential}
\tilde D^\alpha_N= (F_N^{\otimes d})^{-1} \diag\left( \left\{\prod_{j =1}^{d} (i\omega_j)^{\alpha_j} \right\}_{\omega\in \Lambda_N} \right) F_N^{\otimes d}.
\end{align}
The quantum Fourier transforms, $F_N = \frac{1}{\sqrt{2N + 1}}\sum_{m, n \in \{-N, \ldots, N\}} e^{-\frac{2\pi inm}{2N + 1}}\ket n \bra m$, are implemented exactly using $\mathcal O(\log^2 N)$ elementary gates per register via the arbitrary-modulus construction of \cite{mosca2003exact}. The circuits use $\mathcal O(\log N)$ auxiliary qubits, which are returned to zero and reused between applications. The odd-grid embedding is specified in \cref{app:R_N-calculation}. Let $R_N$ denote the cost of block-encoding the diagonal multiplier $\diag\{(i\omega)_{\omega \in \{-N,\ldots,N\}}\}$, excluding Fourier transforms. The block encoding of $\tilde D^\alpha$ then costs $\mathcal O(d\log^2 N + |\alpha|R_N)$ gates.

Naively, $\mathcal O(|\mathcal J|)$ operators of this type must be implemented, leading to a gate complexity of $\mathcal O(|\mathcal J|\ceil{\mathcal J} R_N)$. However, in \cref{fig:D-alpha-circuit} we provide a more efficient controlled block with successful operator
\begin{equation}
\sum_{\alpha \in \mathcal J} \ketbra{\bar\alpha}{\bar\alpha} \otimes \frac{\tilde D^\alpha_N}{N^{|\alpha|}}
\label{eq:sum-of-diffs}
\end{equation}
using $\mathcal O((\ceil{\mathcal J}+1) \log^2 N + \ceil{\mathcal J} (d \log(1+\ceil{\mathcal J}) \log N + R_N))$ elementary gates. We use a conditional swap network with spatial derivative blocks $\Delta_N$, each with successful block $\tilde\partial_N/N$, where
\begin{equation}
\label{eq:single-diag}
\tilde\partial_N= F_N^{-1} \diag\{(i \omega)_{\omega\in \{-N, \cdots, N\}}\} F_N.
\end{equation}
In \cref{app:R_N-calculation} we show that $R_N$ is in the order of $\log N$.

We prepare the routing register in $\ket\chi = F_N^\dagger\ket{-N}$ by applying $F_N^\dagger$ to the all-zero bit string, which represents the frequency $-N$ in the encoding of \cref{app:R_N-calculation}. The swap network in \cref{fig:D-alpha-circuit} routes coordinate $i$ through $\alpha_i$ derivative stages and restores every data register. Since $(\tilde\partial_N/N)\ket\chi = -i\ket\chi$, each of the remaining $\ceil{\mathcal J}-|\alpha|$ stages acts on the idle state with phase $-i$ and unit success amplitude. We cancel their product with the phase gate
\begin{equation}
\Phi\ket{\bar\alpha_d} = i^{\ceil{\mathcal J}-\bar\alpha_d}\ket{\bar\alpha_d},
\qquad \bar\alpha_d = |\alpha|.
\label{eq:derivative-idle-phase}
\end{equation}
This uses $S^\dagger$ on the least significant bit of $\bar\alpha_d$ and $Z$ on its next bit, when present, together with the common phase $i^{\ceil{\mathcal J}}$. We finally apply $F_N$ to the routing register to return $\ket\chi$ to the zero bit string. This results in the implementation of $\tilde D_N^\alpha/N^{|\alpha|}$ on the success branch for each $\alpha$. The \texttt{Prep} weights and normalization $N^{|\alpha|}$ are unchanged. Including preparation and unpreparation of $\ket\chi$, the circuit uses at most $\ceil{\mathcal J}+1$ QFT and inverse-QFT pairs.

We compose the derivative blocks coherently, without intermediate measurements, reusing one auxiliary workspace and a counter initialized to zero. After block $j$, controlled on a nonzero workspace, we interchange counter values $0$ and $j+1$. A branch whose first failure occurs at stage $j$ keeps the value $j+1$, since subsequent updates exchange zero with different stage values. Projection onto zero workspace and counter therefore gives the product of the successful blocks. This uses $\mathcal O(\log N + \log(1+\ceil{\mathcal J}))$ auxiliary qubits and the same order of additional gates per stage. The coefficient block uses a separate workspace.

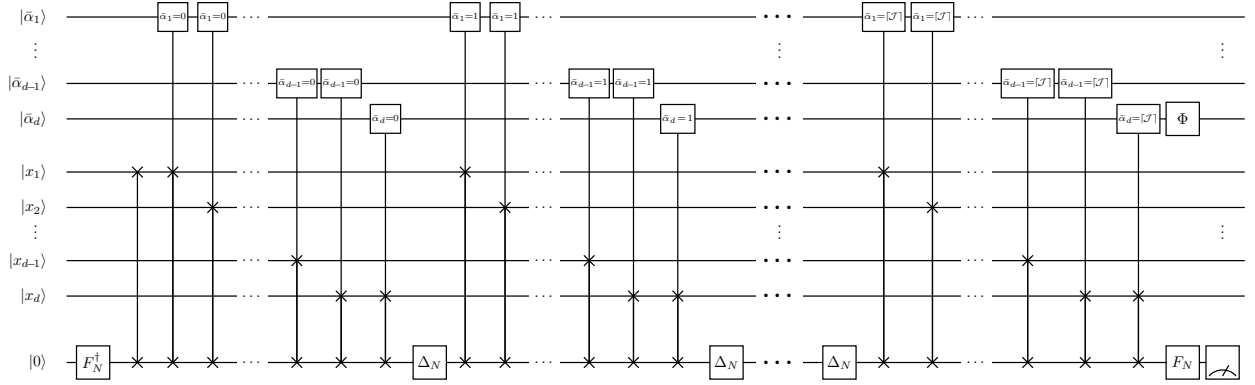
\begin{figure}
\centering
\begin{adjustbox}{max width = \textwidth}

\begin{tikzpicture}[
    thick,
    gate/.style={draw, minimum size=0.65cm, inner sep=1pt, font=\tiny, fill=white},
    dngate/.style={draw, minimum size=0.75cm, inner sep=1pt, font=\small, fill=white},
    metergate/.style={draw, minimum width=0.75cm, minimum height=0.75cm, inner sep=0pt, fill=white,
        path picture={
            \draw ([shift={(0.07,0.07)}]path picture bounding box.south west)
                  to[out=50,in=130]
                  ([shift={(-0.07,0.07)}]path picture bounding box.south east);
            \draw ([yshift=0.07cm]path picture bounding box.south) -- ++(60:0.28cm);
        }},
    wdots/.style={fill=white, inner sep=3pt, font=\small, outer sep=0pt},
    bigdots/.style={fill=white, inner sep=4pt, font=\Large, outer sep=0pt},
    every node/.style={font=\small},
]

\def\ya{0}
\def\yam{-1.5}
\def\yad{-2.3}
\def\yx{-3.5}
\def\yxtwo{-4.3}
\def\yxm{-5.5}
\def\yxd{-6.3}
\def\yanc{-7.8}

\def\cs{0.12}

\newcommand{\swappair}[3]{%
    \draw (#1-\cs, #2-\cs) -- (#1+\cs, #2+\cs);%
    \draw (#1-\cs, #2+\cs) -- (#1+\cs, #2-\cs);%
    \draw (#1-\cs, #3-\cs) -- (#1+\cs, #3+\cs);%
    \draw (#1-\cs, #3+\cs) -- (#1+\cs, #3-\cs);%
    \draw (#1, #2) -- (#1, #3);%
}

\node[anchor=east] at (-0.3, \ya) {$\ket{\bar\alpha_1}$};
\node[anchor=east] at (-0.3, \yam) {$\ket{\bar\alpha_{d\!-\!1}}$};
\node[anchor=east] at (-0.3, \yad) {$\ket{\bar\alpha_d}$};
\node[anchor=east] at (-0.3, \yx) {$\ket{x_1}$};
\node[anchor=east] at (-0.3, \yxtwo) {$\ket{x_2}$};
\node[anchor=east] at (-0.3, \yxm) {$\ket{x_{d\!-\!1}}$};
\node[anchor=east] at (-0.3, \yxd) {$\ket{x_d}$};
\node[anchor=east] at (-0.3, \yanc) {$\ket{0}$};


\pgfmathsetmacro{\sAstart}{1.6}          
\pgfmathsetmacro{\sAa}{\sAstart+0.8}     
\pgfmathsetmacro{\sAb}{\sAa+0.9}         
\pgfmathsetmacro{\sAdotsA}{\sAb+0.9}     
\pgfmathsetmacro{\sAc}{\sAdotsA+1.0}     
\pgfmathsetmacro{\sAd}{\sAc+1.0}         
\pgfmathsetmacro{\sAe}{\sAd+1.0}         
\pgfmathsetmacro{\sDNa}{\sAe+1.0}        

\pgfmathsetmacro{\sBa}{\sDNa+0.8}        
\pgfmathsetmacro{\sBb}{\sBa+0.9}         
\pgfmathsetmacro{\sBdotsA}{\sBb+0.9}     
\pgfmathsetmacro{\sBc}{\sBdotsA+1.0}     
\pgfmathsetmacro{\sBd}{\sBc+1.0}         
\pgfmathsetmacro{\sBe}{\sBd+1.0}         
\pgfmathsetmacro{\sDNb}{\sBe+1.1}        

\pgfmathsetmacro{\dotB}{\sDNb+1.2}

\pgfmathsetmacro{\sDNpre}{\dotB+1.35}

\pgfmathsetmacro{\sCa}{\sDNpre+1.0}      
\pgfmathsetmacro{\sCb}{\sCa+1.1}         
\pgfmathsetmacro{\sCdotsA}{\sCb+1.0}     
\pgfmathsetmacro{\sCc}{\sCdotsA+1.15}     
\pgfmathsetmacro{\sCd}{\sCc+1.3}         
\pgfmathsetmacro{\sCe}{\sCd+1.2}         
\pgfmathsetmacro{\qftout}{\sCe+1.0}
\pgfmathsetmacro{\meas}{\qftout+0.9}        

\pgfmathsetmacro{\wend}{\meas+0.5}

\draw (0, \ya) -- (\wend, \ya);
\draw (0, \yam) -- (\wend, \yam);
\draw (0, \yad) -- (\wend, \yad);
\draw (0, \yx) -- (\wend, \yx);
\draw (0, \yxtwo) -- (\wend, \yxtwo);
\draw (0, \yxm) -- (\wend, \yxm);
\draw (0, \yxd) -- (\wend, \yxd);
\pgfmathsetmacro{\ancend}{\meas-0.375}
\draw (0, \yanc) -- (\ancend, \yanc);

\node[dngate] at (0.6, \yanc) {$F_N^\dagger$};
\node[dngate] at (\qftout, \yanc) {$F_N$};
\node[dngate] at (\qftout, \yad) {$\Phi$};


\swappair{\sAstart}{\yx}{\yanc}

\node[gate] at (\sAa, \ya) {$\bar\alpha_1\!\!=\!\!0$};
\draw (\sAa, {\ya-0.325}) -- (\sAa, \yanc);
\swappair{\sAa}{\yx}{\yanc}

\node[gate] at (\sAb, \ya) {$\bar\alpha_1\!\!=\!\!0$};
\draw (\sAb, {\ya-0.325}) -- (\sAb, \yanc);
\swappair{\sAb}{\yxtwo}{\yanc}

\node[gate] at (\sAc, \yam) {\tiny$\bar\alpha_{d\!-\!1}\!\!=\!\!0$};
\draw (\sAc, {\yam-0.325}) -- (\sAc, \yanc);
\swappair{\sAc}{\yxm}{\yanc}

\node[gate] at (\sAd, \yam) {\tiny$\bar\alpha_{d\!-\!1}\!\!=\!\!0$};
\draw (\sAd, {\yam-0.325}) -- (\sAd, \yanc);
\swappair{\sAd}{\yxd}{\yanc}

\node[gate] at (\sAe, \yad) {$\bar\alpha_d\!\!=\!\!0$};
\draw (\sAe, {\yad-0.325}) -- (\sAe, \yanc);
\swappair{\sAe}{\yxd}{\yanc}

\node[dngate] at (\sDNa, \yanc) {$\Delta_N$};


\node[gate] at (\sBa, \ya) {$\bar\alpha_1\!\!=\!\!1$};
\draw (\sBa, {\ya-0.325}) -- (\sBa, \yanc);
\swappair{\sBa}{\yx}{\yanc}

\node[gate] at (\sBb, \ya) {$\bar\alpha_1\!\!=\!\!1$};
\draw (\sBb, {\ya-0.325}) -- (\sBb, \yanc);
\swappair{\sBb}{\yxtwo}{\yanc}

\node[gate] at (\sBc, \yam) {\tiny$\bar\alpha_{d\!-\!1}\!\!=\!\!1$};
\draw (\sBc, {\yam-0.325}) -- (\sBc, \yanc);
\swappair{\sBc}{\yxm}{\yanc}

\node[gate] at (\sBd, \yam) {\tiny$\bar\alpha_{d\!-\!1}\!\!=\!\!1$};
\draw (\sBd, {\yam-0.325}) -- (\sBd, \yanc);
\swappair{\sBd}{\yxd}{\yanc}

\node[gate] at (\sBe, \yad) {$\bar\alpha_d\!=\!1$};
\draw (\sBe, {\yad-0.325}) -- (\sBe, \yanc);
\swappair{\sBe}{\yxd}{\yanc}

\node[dngate] at (\sDNb, \yanc) {$\Delta_N$};

\foreach \yy in {0, -1.5, -2.3, -3.5, -4.3, -5.5, -6.3, -7.8} {
    \node[bigdots] at (\dotB, \yy) {$\boldsymbol{\cdots}$};
}

\node[dngate] at (\sDNpre, \yanc) {$\Delta_N$};


\node[gate] at (\sCa, \ya) {\tiny$\bar\alpha_1\!\!=\!\!\ceil{\!\mathcal J}$};
\draw (\sCa, {\ya-0.325}) -- (\sCa, \yanc);
\swappair{\sCa}{\yx}{\yanc}

\node[gate] at (\sCb, \ya) {\tiny$\bar\alpha_1\!\!=\!\!\ceil{\!\mathcal J}$};
\draw (\sCb, {\ya-0.325}) -- (\sCb, \yanc);
\swappair{\sCb}{\yxtwo}{\yanc}

\node[gate] at (\sCc, \yam) {\tiny$\bar\alpha_{d\!-\!1}\!\!=\!\!\ceil{\!\mathcal J}$};
\draw (\sCc, {\yam-0.325}) -- (\sCc, \yanc);
\swappair{\sCc}{\yxm}{\yanc}

\node[gate] at (\sCd, \yam) {\tiny$\bar\alpha_{d\!-\!1}\!\!=\!\!\ceil{\!\mathcal J}$};
\draw (\sCd, {\yam-0.325}) -- (\sCd, \yanc);
\swappair{\sCd}{\yxd}{\yanc}

\node[gate] at (\sCe, \yad) {\tiny$\bar\alpha_d\!\!=\!\!\ceil{\!\mathcal J}$};
\draw (\sCe, {\yad-0.325}) -- (\sCe, \yanc);
\swappair{\sCe}{\yxd}{\yanc}

\node[metergate] at (\meas, \yanc) {};

\foreach \xp in {-0.7, \dotB, \meas} {
    \node at (\xp, -0.65) {$\vdots$};
    \node at (\xp, -4.75) {$\vdots$};
}

\foreach \yy in {\ya, \yam, \yad, \yx, \yxtwo, \yxm, \yxd, \yanc} {
    \node[wdots] at (\sAdotsA, \yy) {$\dots$};
    \node[wdots] at (\sBdotsA, \yy) {$\dots$};
    \node[wdots] at (\sCdotsA, \yy) {$\dots$};
}

\end{tikzpicture}

\end{adjustbox}
\caption{Circuit for the controlled derivative block in \eqref{eq:sum-of-diffs}. The endpoint transforms prepare and unprepare the routing state $\ket\chi = F_N^\dagger\ket{-N}$. The swap network routes each spatial coordinate through the unconditional derivative block $\Delta_N$, which includes its own QFT pair, and then restores it. The phase gate $\Phi$ of \eqref{eq:derivative-idle-phase} cancels the idle stages' phases. The auxiliary workspace and controlled failure-counter updates are not shown.}
\label{fig:D-alpha-circuit}
\end{figure}

\subsection{Quantum algorithm and resource bounds}

The matrix inversion step contributes an additional multiplicative factor of $\mathcal O(\kappa_{\mathbb L_N} \log (1/\epsilon))$ to the total elementary gates and query counts, where $\kappa_{\mathbb L_N} = \alpha_{\mathrm{LCU}}/\sigma_{\min}^+(\mathbb L_N)$ is governed by the normalization factor $\alpha_{\mathrm{LCU}}$ actually achieved by our LCU-based block encoding of $\mathbb L_N$ (see the proof below), rather than by the operator's own spectral norm $\sigma_{\max}(\mathbb L_N)$. The end-to-end complexity can now be stated in the following theorem.

\begin{thm}
\label{thm:linear-pde-complexity}
Suppose that the minimum-norm solution $u$ to the general linear PDE of the form \eqref{eq:linear-PDE} with periodic boundary conditions is $s$-Gevrey with radius $r$ for some $s \geq 1$, and that $u_N \perp \ker(\mathbb L_N)$.
Then \cref{alg:linear-pde} solves the PDE with $\mathcal O(1)$ success probability by returning a normalized state $\ket{v}\in\mathbb{C}^{(2N+1)^d}$ which is $\epsilon$-close in $\ell^2$-norm to $\frac{u_N}{\norm{u_N}}$ using
\begin{equation}
\label{eq:total-gate-count}
\tilde{\mathcal O} \left(\frac{\ceil{g}\,|\mathcal J| \,(|\mathcal J|  +
\ceil{\mathcal J})}{\sigma_{\min}^+ r^{\ceil{\mathcal J}}}
d^{\ceil{\mathcal J} s + 1}
\left(s \log \left(\frac{1}{\norm{\hat u}\, \epsilon} \right) \right)^{\ceil{\mathcal J}s}
\log(\frac{1}\epsilon)
\right)
\end{equation}
elementary gates, $\mathcal O(d (\log N + \log\ceil{\mathcal J}) + \log(|\mathcal J|\ceil{g}N^{\ceil{\mathcal J}}/(\epsilon\sigma_{\min}^+)))$ qubits (excluding the coefficient and source oracle qubits), and
\begin{equation}
\label{eq:total-query-count}
\tilde{\mathcal O} \left(\frac{\ceil{g}\, |\mathcal J| \, (sd)^{\ceil{\mathcal J} s}}{\sigma_{\min}^+ r^{\ceil{\mathcal J}}}
\log^{\ceil{\mathcal J}s} \left(\frac{1}{\norm{\hat u}\, \epsilon}\right) \log(\frac{1}\epsilon)\right)
\end{equation}
queries to oracles $U_{g_{\alpha, N}}$ and $U_{\eta_N}$. For sufficiently small $\epsilon > 0$, a sufficient asymptotic cutoff is
\begin{equation}
\label{eq:complexity-cutoff-asymptotic}
N \in \tilde \Omega\!\left(\frac{(s(d + \ceil{\mathcal J}))^s}{r}\log^s\!\left(\frac{C|\mathcal J|\ceil{g}^2N^{\ceil{\mathcal J}}}{\epsilon^2\, (\sigma_{\min}^+)^2\norm{\hat u}}\right)\right).
\end{equation}
\end{thm}

\begin{proof}
As per the discussion above, the elementary gate count for the block encoding of $\mathbb L_N$ is
\begin{equation}
\label{eq:LN-gate-count}
\mathcal O\left(
d|\mathcal J| \log|\mathcal J| \log(1+\ceil{\mathcal J}) + (\ceil{\mathcal J}+1) \log^2 N +
\ceil{\mathcal J} \big(d \log(1+\ceil{\mathcal J}) \log N + R_N\big)
\right)
\end{equation}
In \cref{app:R_N-calculation} we show that $R_N = \mathcal O(\log N)$. Suppressing logarithmic factors in $N$, $|\mathcal J|$, and the differential order, the per-query gate cost is therefore $\tilde{\mathcal O}(d(|\mathcal J| + \ceil{\mathcal J}))$, including \texttt{Prep}. Since $\|\tilde D_N^\alpha\| \leq N^{|\alpha|}$, the operator triangle inequality gives
\begin{equation}
\label{eq:sigma-max-bound}
\sigma_{\max}(\mathbb L_N) \leq \sum_{\alpha \in \mathcal J}\max_x|g_\alpha(x)|\, N^{|\alpha|}
\leq |\mathcal J|\ceil{g}N^{\ceil{\mathcal J}}.
\end{equation}
This is an upper bound, not an equality between the spectral norm and the LCU normalization. A properly weighted \texttt{Prep} circuit achieves normalization factor
\begin{equation}
\label{eq:alpha-lcu-bound}
\alpha_{\mathrm{LCU}} = \sum_{\alpha \in \mathcal J} \max_x|g_\alpha(x)|\, N^{|\alpha|} \leq |\mathcal J| \, \ceil{g}\, N^{\ceil{\mathcal J}},
\end{equation}
using $N^{|\alpha|} \leq N^{\ceil{\mathcal J}}$ and $\max_x|g_\alpha(x)| \leq \ceil{g}$ termwise. The QLSA's block-encoding condition parameter is therefore
\begin{equation}
\label{eq:kappa-lcu-bound}
\kappa_{\mathbb L_N} := \alpha_{\mathrm{LCU}}/\sigma_{\min}^+(\mathbb L_N) \leq |\mathcal J|\ceil{g}N^{\ceil{\mathcal J}}/\sigma_{\min}^+(\mathbb L_N).
\end{equation}
We allocate $\epsilon/2$ to discretization and the remaining $\epsilon/2$ to the quantum solve, including finite-precision state preparation. The resulting query bound is $\tilde{\mathcal O}(\ceil{g}|\mathcal J|N^{\ceil{\mathcal J}}\log(1/\epsilon)/\sigma_{\min}^+)$, with per-query gate cost $\tilde{\mathcal O}(d(|\mathcal J| + \ceil{\mathcal J}))$ after including \texttt{Prep}.

To handle an inconsistent source, let $P_N=\mathbb L_N\mathbb L_N^+$ be the orthogonal projector onto the range and write $\xi_N=\mathbb L_Nu_N-\eta_N$. The consistent system with source $P_N\eta_N$ has the same pseudoinverse solution. If $\norm{\xi_N}\leq\theta\sigma_{\min}^+\norm{u_N}$ with $\theta<1/2$, then
\begin{equation}
\label{eq:source-projection-error}
\norm{(I-P_N)\eta_N}\leq\norm{\xi_N},
\quad
\norm{\eta_N}\geq(1-\theta)\sigma_{\min}^+\norm{u_N},
\quad \text{and} \quad
\norm{\ket{\eta_N}-\ket{P_N\eta_N}}\leq\frac{2\theta}{1-\theta}.
\end{equation}
Consequently, $U_{\eta_N}$ is within $4\theta$ in operator norm of a source oracle for $\ket{P_N\eta_N}$. For the source $P_N\eta_N$, \cite[Theorem 4]{dalzell2024shortcut} guarantees ensemble trace-distance error $c\epsilon^2$ with expected query complexity $\mathcal O(\kappa_{\mathbb L_N}\log(1/\epsilon))$, where $c>0$ is sufficiently small. A query budget $Q$ equal to a sufficiently large constant multiple of this expectation retains constant completion probability and conditional mean infidelity $\mathcal O(\epsilon^2)$. The telescoping bound for at most $Q$ oracle calls shows that replacing the projected-source oracle by $U_{\eta_N}$ changes the joint output, including the completion flag, by $\mathcal O(Q\theta)$ in trace distance. Choosing $\theta=c\epsilon^2/Q$ and implementing each controlled oracle and inverse to operator-norm error $\mathcal O(\epsilon^2/Q)$ preserves these bounds. Markov's inequality then gives a pure-state error at most $\epsilon/2$, up to global phase, with constant probability.

By \cref{cor:ultimate-N}, with $w = s(d + \ceil{\mathcal J})$, a sufficient cutoff is
\begin{align}
N \geq \frac1r\max\Biggl\{&2\log\!\left(\frac{4\sqrt2\, C\ceil{g}}{\epsilon\, \sigma_{\min}^+(\mathbb L_N)\norm{\hat u}}\right),2\log\!\left(\frac{4C}{\norm{\hat u}}\right), 64w\log\!\left(\frac{64w}{r^{1/s}}\right)\Biggr\}^s.
\label{eq:complexity-cutoff}
\end{align}
Here we have dropped the favorable factor $(2N + 1)^{d/2}$ from the first denominator. For the unprojected source oracle, we use this bound with $\epsilon$ replaced by $c\epsilon^2/(\kappa_{\mathbb L_N}\log(1/\epsilon))$; the residual estimate in \cref{thm:linear-pde-precision} then supplies the required $\theta$. Any available lower bound on $\sigma_{\min}^+(\mathbb L_N)$ can be used in its place. At fixed regularity parameters and PDE order, $w = \mathcal O(ds)$, and choosing the least sufficient positive integer cutoff gives
\begin{equation}
N \in \tilde{\mathcal O}\!\left(\frac1r\left(ds\log\!\left(\frac{C}{\norm{\hat u}} + \frac{C|\mathcal J|\ceil{g}^2N^{\ceil{\mathcal J}}}{\epsilon^2\, (\sigma_{\min}^+)^2\norm{\hat u}}\right)\right)^s\right).
\end{equation}
Combining this cutoff with the query and gate bounds above proves the stated small-error complexities.
\end{proof}

\begin{rmk}
Using $\norm{\hat u} = \norm{u}/(2\pi)^{d/2}$ (\cref{rmk:ultimate-N-u}) and $\log\left(\frac{(2\pi)^{d/2}}{X}\right) = O(d\log(1/X))$ for $\log(1/X)\geq1$, \eqref{eq:total-gate-count} and \eqref{eq:total-query-count} yield
\begin{equation}
\tilde{\mathcal O} \left(\frac{\ceil{g}\, |\mathcal J|\, (|\mathcal J| + \ceil{\mathcal J})}{\sigma_{\min}^+ r^{\ceil{\mathcal J}}} d^{2\ceil{\mathcal J} s + 1} \left(s \log \left(\frac{1}{\norm{u}\, \epsilon} \right) \right)^{\ceil{\mathcal J}s} \log(\frac{1}\epsilon) \right)
\end{equation}
elementary gates and
\begin{equation}
\tilde{\mathcal O} \left(\frac{\ceil{g}\, |\mathcal J|\, (sd^2)^{\ceil{\mathcal J} s}}{\sigma_{\min}^+ r^{\ceil{\mathcal J}}} \log^{\ceil{\mathcal J}s} \left(\frac{1}{\norm{u}\, \epsilon}\right) \log(\frac{1}\epsilon)\right)
\end{equation}
queries, in terms of the $\ell^2$-norm of the solution.
\end{rmk}

\begin{rmk}
\label{rmk:source-projection}
Continuous minimum norm, $u \perp \ker(\mathcal L)$, does not by itself imply the discrete target-orthogonality condition $u_N \perp \ker(\mathbb L_N)$ in \cref{thm:linear-pde-complexity}. This condition must therefore be imposed or justified separately in applications; otherwise the vector-error bound includes $\norm{P_{\ker(\mathbb L_N)}u_N}$.
\end{rmk}

\begin{rmk}
For $0 < s < 1$, the same query and gate bounds in terms of $N$ apply. We use $r \mapsto (s/2)^s r$ in \eqref{eq:complexity-cutoff}, with $w = \max(d, s(d + \ceil{\mathcal J}))$. This retains the order dependence of the finite-order entire regime.
\end{rmk}

\section{Extension to Dirichlet Boundary Conditions}
\label{sec:dirichlet}

In this section we demonstrate how our framework can be used to handle non-periodic boundary conditions. We consider \eqref{eq:linear-PDE} on $\Omega=[-1,1]^d$, with a generic boundary condition $u=b$ on $\partial\Omega$. This is not a reduction to the periodic case: no periodic PDE is imposed on $[-\pi,\pi]^d$, and \cref{thm:linear-pde-complexity} is not invoked as a stated result. The periodic hypercube instead supplies a trial space. We show that $u$ admits a periodic extension $\bar u$ whose Gevrey parameters we control, so that a Fourier cutoff $N$ that is polylogarithmic in the target precision already resolves $u$ on $\Omega$, and we then collocate the original equation and its boundary condition on that trial space (see \cref{fig:dirichlet-extension}). Controlling the extension's regularity and the probability of obtaining the solution state on $\Omega$ are the keys to this construction. The same trial space underlies the classical Fourier extension and Fourier continuation methods \cite{boyd2002comparison, bruno2010high, adcock2014numerical} and immersed boundary smooth extension \cite{stein2016immersed}. The cutoff $N$ specifies the Fourier trial space on $\Gamma_N$. The PDE evaluation and solution-state output use the finer grid $\Gamma_M|_\Omega$, defined below.

\subsection{Gevrey extensions and ramp construction}

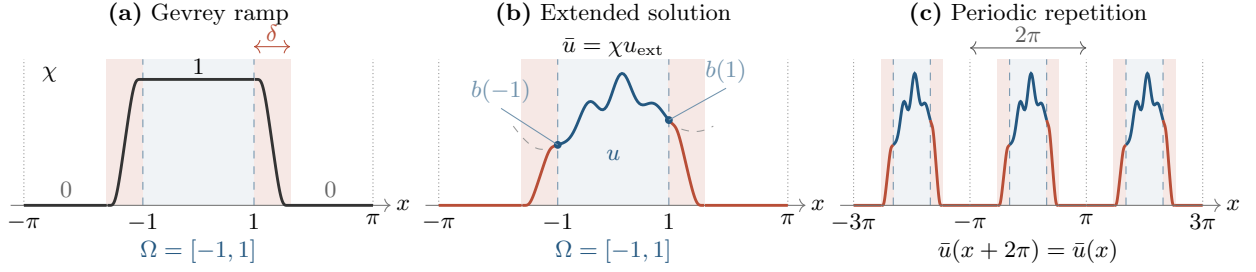
\begin{figure}[t]
\centering
\begingroup
\definecolor{dirichletblue}{RGB}{35,87,128}
\definecolor{dirichletred}{RGB}{184,78,55}
\pgfmathdeclarefunction{dirichletcut}{1}{%
	\pgfmathparse{exp(-1/(1-#1))/(exp(-1/#1)+exp(-1/(1-#1)))}%
}
\pgfmathdeclarefunction{dirichletprofile}{1}{%
	\pgfmathparse{0.72+0.2*cos(135*#1)+0.1*sin(100*#1)+0.12*sin(540*#1)*exp(-2*(#1)^2)}%
}
\begin{adjustbox}{max width=\textwidth}
\begin{tikzpicture}[
	font=\footnotesize,
	physical/.style={draw=dirichletblue!65,thin,dashed},
	periodic/.style={draw=black!45,thin,densely dotted},
	profile/.style={draw=dirichletblue,line width=1.1pt},
	taper/.style={draw=dirichletred,line width=1.1pt},
	axis/.style={draw=black!55,thin,->},
	every node/.style={inner sep=2pt}]

\foreach \panel in {0,1} {
	\begin{scope}[shift={(5.6*\panel,0)},x=0.75cm,y=1.7cm]
		\fill[dirichletblue!7] (-1,0) rectangle (1,1.16);
		\fill[dirichletred!13] (-1.65,0) rectangle (-1,1.16);
		\fill[dirichletred!13] (1,0) rectangle (1.65,1.16);
		\foreach \boundary in {-1,1}
			\draw[physical] (\boundary,0) -- (\boundary,1.16);
		\foreach \boundary in {-pi,pi}
			\draw[periodic] (\boundary,0) -- (\boundary,1.16);
		\draw[axis] (-3.38,0) -- (3.46,0) node[right] {$x$};
		\foreach \position/\tick in {-pi/-\pi,-1/-1,1/1,pi/\pi}
			\node[below] at (\position,-0.025) {$\tick$};
		\node[dirichletblue] at (0,-0.36) {$\Omega=[-1,1]$};
	\end{scope}
}

\begin{scope}[x=0.75cm,y=1.7cm]
	\node at (0,1.52) {\textbf{(a)} Gevrey ramp};
	\node at (-2.68,1.04) {$\chi$};
	\draw[black!80,line width=1.1pt] (-pi,0) -- (-1.65,0)
		plot[domain=0.98:0.02,samples=55,variable=\t]
			({-1-0.65*\t},{dirichletcut(\t)})
		-- (-1,1) -- (1,1)
		plot[domain=0.02:0.98,samples=55,variable=\t]
			({1+0.65*\t},{dirichletcut(\t)})
		-- (1.65,0) -- (pi,0);
	\node[above] at (0,1.0) {$1$};
	\draw[<->,dirichletred] (1,1.26) -- (1.65,1.26)
		node[midway,above] {$\delta$};
	\node[black!60] at (-2.38,0.13) {$0$};
	\node[black!60] at (2.38,0.13) {$0$};
\end{scope}

\begin{scope}[shift={(5.6,0)},x=0.75cm,y=1.7cm]
	\node at (0,1.52) {\textbf{(b)} Extended solution};
	\node at (0,1.25) {$\bar u=\chi u_{\mathrm{ext}}$};
	\foreach \side in {-1,1} {
		\draw[black!40,thin,dashed]
			plot[domain=1:1.8,samples=35,variable=\t]
				({\side*\t},{dirichletprofile(\side*\t)});
		\draw[taper] (\side*pi,0) -- (\side*1.65,0)
			plot[domain=0.98:0.02,samples=55,variable=\t]
				({\side*(1+0.65*\t)},
				 {dirichletcut(\t)*dirichletprofile(\side*(1+0.65*\t))})
			-- (\side,{dirichletprofile(\side)});
		\fill[dirichletblue] (\side,{dirichletprofile(\side)}) circle[radius=1.5pt];
	}
	\draw[profile] plot[domain=-1:1,samples=65] (\x,{dirichletprofile(\x)});
	\node[dirichletblue] at (0,0.4) {$u$};
	\draw[dirichletblue!65,thin] (-1,{dirichletprofile(-1)}) -- (-2.02,0.77)
		node[above] {$b(-1)$};
	\draw[dirichletblue!65,thin] (1,{dirichletprofile(1)}) -- (2.03,0.93)
		node[above] {$b(1)$};
\end{scope}

\begin{scope}[shift={(11.2,0)},x=0.25cm,y=1.7cm]
	\node at (0,1.52) {\textbf{(c)} Periodic repetition};
	\foreach \copy in {-1,0,1} {
		\begin{scope}[shift={(2*pi*\copy,0)}]
			\fill[dirichletblue!7] (-1,0) rectangle (1,1.16);
			\fill[dirichletred!13] (-1.65,0) rectangle (-1,1.16);
			\fill[dirichletred!13] (1,0) rectangle (1.65,1.16);
			\foreach \side in {-1,1} {
				\draw[physical] (\side,0) -- (\side,1.16);
				\draw[taper] (\side*pi,0) -- (\side*1.65,0)
					plot[domain=0.98:0.02,samples=55,variable=\t]
						({\side*(1+0.65*\t)},
						 {dirichletcut(\t)*dirichletprofile(\side*(1+0.65*\t))})
					-- (\side,{dirichletprofile(\side)});
			}
			\draw[profile] plot[domain=-1:1,samples=65] (\x,{dirichletprofile(\x)});
		\end{scope}
	}
	\foreach \boundary in {-3*pi,-pi,pi,3*pi}
		\draw[periodic] (\boundary,0) -- (\boundary,1.16);
	\draw[axis] (-10.15,0) -- (10.4,0) node[right] {$x$};
	\foreach \position/\tick in {-3*pi/-3\pi,-pi/-\pi,pi/\pi,3*pi/3\pi}
		\node[below] at (\position,-0.025) {$\tick$};
	\draw[<->,black!65] (-pi,1.25) -- (pi,1.25)
		node[midway,above] {$2\pi$};
	\node at (0,-0.36) {$\bar u(x+2\pi)=\bar u(x)$};
\end{scope}
\end{tikzpicture}
\end{adjustbox}
\endgroup
\caption{One-dimensional schematic of the Dirichlet extension. (a) The ramp function equals one on the Dirichlet box (blue) and falls to zero across collars of width $\delta$ (red). The width is chosen as in \cref{lem:dirichlet-extension}, with its three constraints explained in \cref{rmk:collar-selection}. (b) Multiplication by the ramp function preserves the nonconstant solution and its Dirichlet values. The gray dashed curves indicate the untapered continuation $u_{\mathrm{ext}}$. (c) The whole computational cell is repeated, with a smaller horizontal scale. Dashed vertical lines mark Dirichlet boundaries and dotted lines mark periodic cell boundaries.}
\label{fig:dirichlet-extension}
\end{figure}

We construct the ramp functions by repeated averaging, as illustrated in \cref{fig:convolution-cutoffs}. The normalized interval indicator $b_a(t)=a^{-1}\mathbf1_{[0,a]}(t)$ satisfies
\begin{equation}
\label{eq:cutoff-moving-average}
(b_a*f)(t)=\frac1a\int_{t-a}^t f(y)\,dy,
\quad \text{and} \quad (b_a*f)'(t)=\frac{f(t)-f(t-a)}a.
\end{equation}
Convolving two equal-width indicators gives a triangular density; further convolutions add smoothness. For positive widths with $\sum_{j\geq1}a_j=w<\infty$, the infinite convolution is a smooth density supported on $[0,w]$. Indeed, for each derivative order, finitely many factors give a continuous derivative, and the remaining averaging converges uniformly as its total width tends to zero. We define the density and its descending ramp by
\begin{equation}
\label{eq:convolution-ramp}
\rho=b_{a_1}*b_{a_2}*\cdots,
\quad \text{and} \quad \psi(t)=1-\int_{-\infty}^t\rho(y)\,dy.
\end{equation}
Thus $\psi=1$ on $(-\infty,0]$, $\psi=0$ on $[w,\infty)$, and $0\leq\psi\leq1$. All positive-order derivatives vanish at the two endpoints. Applying \eqref{eq:cutoff-moving-average} to distinct convolution factors gives
\begin{equation}
\label{eq:convolution-ramp-derivatives}
\norm{\psi^{(\ell)}}_\infty\leq\frac{2^{\ell-1}}{a_1\cdots a_\ell},
\qquad \text{for} \quad \ell\geq1.
\end{equation}
Here each of the first $\ell-1$ derivatives contributes a difference divided by its interval width; the next density has supremum $1/a_\ell$, and convolution with the remaining normalized densities does not increase that supremum. This makes the derivative growth explicit in the choice of widths.

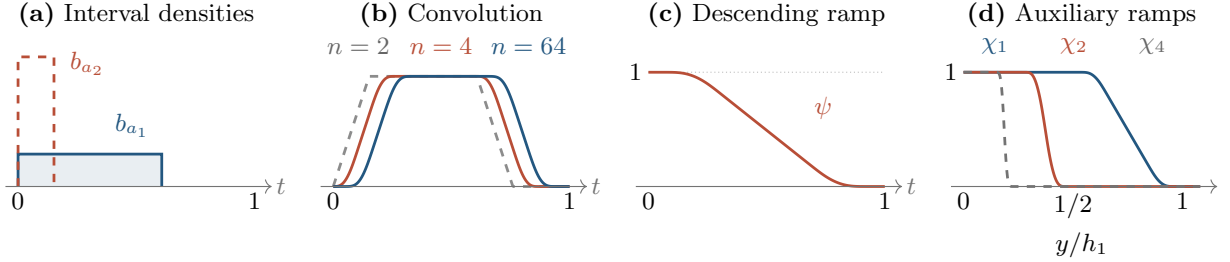
\begin{figure}[t]
\centering
\begingroup
\begin{adjustbox}{width=\textwidth}
\begin{tikzpicture}[
	font=\footnotesize,every node/.style={inner sep=1.5pt},
	axis/.style={black!55,thin,->},line/.style={line width=1pt}]
\definecolor{convolutionblue}{RGB}{35,87,128}
\definecolor{convolutionred}{RGB}{184,78,55}
\path[use as bounding box] (-0.28,-0.9) rectangle (15.5,2.35);
\begin{scope}
	\node at (1.5,2.18) {\textbf{(a)} Interval densities};
	\fill[convolutionblue!10] (0,0) rectangle (1.82378130556208,0.4112335167120566);
	\draw[convolutionblue,line] (0,0) -- (0,0.4112335167120566)
		-- (1.82378130556208,0.4112335167120566) -- (1.82378130556208,0);
	\draw[convolutionred,line,dashed] (0,0) -- (0,1.6449340668482264)
		-- (0.45594532639052,1.6449340668482264) -- (0.45594532639052,0);
	\node[convolutionblue] at (1.45,0.78) {$b_{a_1}$};
	\node[convolutionred] at (0.88,1.56) {$b_{a_2}$};
	\draw[axis] (-0.16,0) -- (3.22,0) node[right] {$t$};
	\node[below] at (0,-0.025) {$0$};
	\node[below] at (3,-0.025) {$1$};
\end{scope}
\begin{scope}[shift={(4,0)}]
	\node at (1.5,2.18) {\textbf{(b)} Convolution};
	\draw[black!45,line,dashed] plot coordinates {
		(0.00000,0.00000) (0.45000,1.37996) (0.46250,1.39819) (1.81250,1.39819) (1.82500,1.39446)
		(2.27500,0.01449) (2.28750,0.00000) (3.00000,0.00000)};
	\draw[convolutionred,line] plot coordinates {
		(0.00000,0.00000) (0.02500,0.00035) (0.03750,0.00117) (0.05000,0.00277) (0.06250,0.00540)
		(0.07500,0.00933) (0.08750,0.01482) (0.10000,0.02213) (0.11250,0.03150) (0.12500,0.04319)
		(0.13750,0.05723) (0.15000,0.07364) (0.16250,0.09242) (0.17500,0.11356) (0.18750,0.13707)
		(0.20000,0.16293) (0.21250,0.19115) (0.22500,0.22152) (0.23750,0.25379) (0.25000,0.28771)
		(0.27500,0.35942) (0.30000,0.43459) (0.47500,0.97099) (0.50000,1.04592) (0.51250,1.08214)
		(0.53750,1.15080) (0.55000,1.18273) (0.56250,1.21270) (0.57500,1.24047) (0.58750,1.26588)
		(0.60000,1.28892) (0.61250,1.30960) (0.62500,1.32792) (0.63750,1.34387) (0.65000,1.35746)
		(0.66250,1.36868) (0.67500,1.37764) (0.68750,1.38457) (0.70000,1.38974) (0.71250,1.39340)
		(0.72500,1.39581) (0.73750,1.39724) (0.75000,1.39794) (0.77500,1.39819) (1.83750,1.39814)
		(1.85000,1.39780) (1.86250,1.39691) (1.87500,1.39522) (1.88750,1.39247) (1.90000,1.38840)
		(1.91250,1.38274) (1.92500,1.37525) (1.93750,1.36565) (1.95000,1.35374) (1.96250,1.33946)
		(1.97500,1.32282) (1.98750,1.30382) (2.00000,1.28245) (2.01250,1.25871) (2.02500,1.23261)
		(2.03750,1.20418) (2.05000,1.17361) (2.06250,1.14116) (2.07500,1.10710) (2.10000,1.03517)
		(2.12500,0.95989) (2.30000,0.42350) (2.32500,0.34870) (2.35000,0.27766) (2.36250,0.24420)
		(2.37500,0.21246) (2.38750,0.18269) (2.40000,0.15515) (2.41250,0.12997) (2.42500,0.10715)
		(2.43750,0.08670) (2.45000,0.06862) (2.46250,0.05289) (2.47500,0.03954) (2.48750,0.02854)
		(2.50000,0.01979) (2.51250,0.01305) (2.52500,0.00804) (2.53750,0.00451) (2.55000,0.00220)
		(2.56250,0.00086) (2.58750,0.00002) (3.00000,0.00000)};
	\draw[convolutionblue,line] plot coordinates {
		(0.00000,0.00000) (0.15000,0.00003) (0.18750,0.00052) (0.21250,0.00206) (0.22500,0.00363)
		(0.23750,0.00601) (0.25000,0.00944) (0.26250,0.01416) (0.27500,0.02039) (0.28750,0.02836)
		(0.30000,0.03825) (0.31250,0.05023) (0.32500,0.06440) (0.33750,0.08083) (0.35000,0.09958)
		(0.36250,0.12065) (0.37500,0.14400) (0.38750,0.16958) (0.40000,0.19729) (0.41250,0.22698)
		(0.42500,0.25848) (0.43750,0.29157) (0.46250,0.36164) (0.48750,0.43530) (0.52500,0.54904)
		(0.63750,0.89356) (0.66250,0.96905) (0.68750,1.04248) (0.71250,1.11217) (0.72500,1.14502)
		(0.73750,1.17623) (0.75000,1.20560) (0.76250,1.23297) (0.77500,1.25819) (0.78750,1.28117)
		(0.80000,1.30185) (0.81250,1.32022) (0.82500,1.33628) (0.83750,1.35008) (0.85000,1.36171)
		(0.86250,1.37127) (0.87500,1.37894) (0.88750,1.38491) (0.90000,1.38940) (0.91250,1.39264)
		(0.92500,1.39487) (0.93750,1.39633) (0.96250,1.39774) (1.01250,1.39819) (1.96250,1.39819)
		(2.01250,1.39763) (2.03750,1.39601) (2.05000,1.39437) (2.06250,1.39190) (2.07500,1.38835)
		(2.08750,1.38350) (2.10000,1.37711) (2.11250,1.36896) (2.12500,1.35887) (2.13750,1.34668)
		(2.15000,1.33230) (2.16250,1.31563) (2.17500,1.29666) (2.18750,1.27537) (2.20000,1.25179)
		(2.21250,1.22600) (2.22500,1.19810) (2.23750,1.16822) (2.25000,1.13656) (2.26250,1.10332)
		(2.27500,1.06873) (2.30000,0.99646) (2.33750,0.88358) (2.47500,0.46304) (2.50000,0.38850)
		(2.52500,0.31689) (2.53750,0.28276) (2.55000,0.25006) (2.56250,0.21902) (2.57500,0.18983)
		(2.58750,0.16267) (2.60000,0.13767) (2.61250,0.11491) (2.62500,0.09445) (2.63750,0.07631)
		(2.65000,0.06047) (2.66250,0.04688) (2.67500,0.03547) (2.68750,0.02609) (2.70000,0.01860)
		(2.71250,0.01278) (2.72500,0.00843) (2.73750,0.00530) (2.75000,0.00315) (2.76250,0.00175)
		(2.78750,0.00042) (2.82500,0.00002) (3.00000,0.00000)};
	\node[black!60] at (0.32,1.78) {$n=2$};
	\node[convolutionred] at (1.37,1.78) {$n=4$};
	\node[convolutionblue] at (2.48,1.78) {$n=64$};
	\draw[axis] (-0.16,0) -- (3.22,0) node[right] {$t$};
	\node[below] at (0,-0.025) {$0$};
	\node[below] at (3,-0.025) {$1$};
\end{scope}
\begin{scope}[shift={(8,0)}]
	\node at (1.5,2.18) {\textbf{(c)} Descending ramp};
	\draw[black!25,densely dotted] (0,1.45) -- (3,1.45);
	\draw[convolutionred,line] plot coordinates {
		(0.00000,1.45000) (0.23750,1.44993) (0.30000,1.44926) (0.33750,1.44803) (0.36250,1.44660)
		(0.40000,1.44325) (0.42500,1.44002) (0.46250,1.43343) (0.48750,1.42777) (0.50000,1.42454)
		(0.51250,1.42104) (0.53750,1.41324) (0.56250,1.40434) (0.58750,1.39436) (0.61250,1.38329)
		(0.62500,1.37734) (0.65000,1.36464) (0.68750,1.34359) (0.72500,1.32024) (0.75000,1.30353)
		(0.77500,1.28600) (0.82500,1.24903) (0.87500,1.21036) (2.03750,0.28644) (2.10000,0.23708)
		(2.13750,0.20800) (2.17500,0.17978) (2.20000,0.16166) (2.22500,0.14424) (2.25000,0.12763)
		(2.27500,0.11195) (2.30000,0.09727) (2.32500,0.08363) (2.35000,0.07107) (2.37500,0.05960)
		(2.40000,0.04921) (2.42500,0.03992) (2.43750,0.03568) (2.46250,0.02802) (2.48750,0.02143)
		(2.51250,0.01591) (2.53750,0.01140) (2.57500,0.00639) (2.60000,0.00407) (2.62500,0.00243)
		(2.65000,0.00134) (2.68750,0.00045) (2.75000,0.00003) (3.00000,0.00000)};
	\node[convolutionred] at (2.22,1.0) {$\psi$};
	\node[left] at (-0.02,1.45) {$1$};
	\draw[axis] (-0.16,0) -- (3.22,0) node[right] {$t$};
	\node[below] at (0,-0.025) {$0$};
	\node[below] at (3,-0.025) {$1$};
\end{scope}
\begin{scope}[shift={(12,0)}]
	\node at (1.5,2.18) {\textbf{(d)} Auxiliary ramps};
	\draw[convolutionblue,line] plot coordinates {
		(0.00000,1.45000) (1.50120,1.44992) (1.53337,1.44903) (1.55481,1.44705) (1.56554,1.44534)
		(1.57626,1.44302) (1.58698,1.44000) (1.59771,1.43619) (1.60843,1.43152) (1.61915,1.42597)
		(1.62987,1.41950) (1.64060,1.41211) (1.65132,1.40377) (1.66204,1.39451) (1.67277,1.38431)
		(1.68349,1.37318) (1.69421,1.36112) (1.70493,1.34815) (1.71566,1.33429) (1.72638,1.31959)
		(1.74783,1.28794) (1.77999,1.23624) (1.81216,1.18186) (2.32686,0.29875) (2.36975,0.22570)
		(2.40192,0.17319) (2.42337,0.14065) (2.44481,0.11100) (2.45553,0.09743) (2.46626,0.08476)
		(2.47698,0.07301) (2.48770,0.06219) (2.49843,0.05230) (2.50915,0.04335) (2.51987,0.03533)
		(2.53059,0.02824) (2.54132,0.02208) (2.55204,0.01683) (2.56276,0.01245) (2.57349,0.00891)
		(2.58421,0.00613) (2.59493,0.00403) (2.60565,0.00250) (2.62710,0.00078) (2.65927,0.00006)
		(3.00240,0.00000)};
	\draw[convolutionred,line] plot coordinates {
		(0.00000,1.45000) (0.79349,1.44978) (0.80421,1.44933) (0.81494,1.44833) (0.82566,1.44639)
		(0.83638,1.44306) (0.84711,1.43781) (0.85783,1.43011) (0.86855,1.41945) (0.87927,1.40533)
		(0.89000,1.38739) (0.90072,1.36533) (0.91144,1.33898) (0.92217,1.30825) (0.93289,1.27311)
		(0.94361,1.23354) (0.95433,1.18956) (0.96506,1.14116) (0.97578,1.08841) (0.98650,1.03144)
		(0.99723,0.97057) (1.01867,0.83943) (1.06156,0.56494) (1.07229,0.49975) (1.08301,0.43773)
		(1.09373,0.37947) (1.10445,0.32534) (1.11518,0.27553) (1.12590,0.23012) (1.13662,0.18913)
		(1.14735,0.15257) (1.15807,0.12042) (1.16879,0.09267) (1.17951,0.06924) (1.19024,0.05000)
		(1.20096,0.03470) (1.21168,0.02297) (1.22241,0.01437) (1.23313,0.00839) (1.24385,0.00450)
		(1.25457,0.00217) (1.26530,0.00091) (1.27602,0.00032) (1.29747,0.00002) (3.00240,0.00000)};
	\draw[black!55,line,dashed] plot coordinates {
		(0.00000,1.45000) (0.41819,1.44983) (0.42891,1.44881) (0.43964,1.44470) (0.45036,1.43251)
		(0.46108,1.40383) (0.47181,1.34767) (0.48253,1.25358) (0.49325,1.11647) (0.50397,0.94067)
		(0.52542,0.53945) (0.53614,0.35870) (0.54687,0.21499) (0.55759,0.11432) (0.56831,0.05284)
		(0.57903,0.02063) (0.58976,0.00650) (0.60048,0.00154) (0.61120,0.00024) (0.62193,0.00002)
		(3.00240,0.00000)};
	\node[convolutionblue] at (0.4,1.78) {$\chi_1$};
	\node[convolutionred] at (1.4,1.78) {$\chi_2$};
	\node[black!60] at (2.4,1.78) {$\chi_4$};
	\node[left] at (-0.02,1.45) {$1$};
	\draw[axis] (-0.16,0) -- (3.22,0);
	\node[below] at (0,-0.025) {$0$};
	\node[below] at (1.39,-0.025) {$1/2$};
	\node[below] at (2.78,-0.025) {$1$};
	\node[below] at (1.5,-0.52) {$y/h_1$};
\end{scope}
\end{tikzpicture}
\end{adjustbox}
\endgroup
\caption{Convolution construction of the ramp functions. (a) Normalized interval indicators with $a_j=1/(\zeta(2)j^2)$. (b) Densities after $n=2,4,64$ convolutions. (c) The descending ramp obtained by integration. (d) The auxiliary $\chi_k$ from \eqref{eq:auxiliary-cutoff}, shown against $y/h_1$ for $\tau=2$. Here $h_k=\delta/(4e^\tau k^{\tau-1})$, with collar width $\delta$, and the transition intervals are $[h_k/2,h_k]$. Panels (c) and (d) use 64-factor approximations; the definitions use infinite convolutions.}
\label{fig:convolution-cutoffs}
\end{figure}

For $\tau>1$ and a collar width $0<\delta<\pi-1$, we set $D_0=4e^\tau$, $q=\tau/(\tau-1)$, and $h_k=\delta/(D_0k^{\tau-1})$. The auxiliary ramp function $\chi_k$ localizes the $k$-th boundary-derivative term in the extension below. We use the widths
\begin{equation}
\label{eq:auxiliary-cutoff-widths}
a_{k,j}=\frac{h_k}{2Z_k}\max(k,j)^{-\tau},
\quad \text{where} \quad Z_k=\sum_{j\geq1}\max(k,j)^{-\tau}\leq qk^{1-\tau},
\end{equation}
whose sum is $h_k/2$. Writing $\rho_k$ for their convolution, we define
\begin{equation}
\label{eq:auxiliary-cutoff}
\chi_k(y)=1-\int_{-\infty}^{y-h_k/2}\rho_k(t)\,dt.
\end{equation}
It equals one up to $h_k/2$ and zero from $h_k$ onwards. The first $k$ widths are equal and the remaining widths decay as $j^{-\tau}$, so \eqref{eq:convolution-ramp-derivatives} gives
\begin{equation}
\label{eq:auxiliary-cutoff-derivatives}
\norm{\chi_k^{(\ell)}}_\infty
\leq\left(\frac{4D_0q}{\delta}\right)^\ell
\begin{cases}
k^{\tau\ell},&\ell\leq k,\\
k^{\tau k}(\ell!/k!)^\tau,&\ell>k.
\end{cases}
\end{equation}
We denote the piecewise factor by $P_\ell(k)$; the estimate also holds for $\ell=0$.

For the multi-dimensional ramp function $\chi$, we instead use \eqref{eq:convolution-ramp} with $a_j=\delta/(\zeta(\tau)j^\tau)$, whose sum is $\delta$. We define
\begin{equation}
\label{eq:dirichlet-cutoff-function}
\chi(x)=\prod_{j=1}^d\psi(|x_j|-1),
\quad \text{with} \quad
\norm{D^\alpha\chi}_\infty\leq
\left(\frac{2\zeta(\tau)}{\delta}\right)^{|\alpha|}(\alpha!)^\tau.
\end{equation}
The derivative estimate follows from $\prod_{j=1}^\ell j^\tau=(\ell!)^\tau$ in \eqref{eq:convolution-ramp-derivatives}. This ramp function equals one on the Dirichlet hypercube and zero outside $[-1-\delta,1+\delta]^d$. Unlike the auxiliary $\chi_k$, it multiplies the entire extension: $\bar u=\chi u_{\mathrm{ext}}$.

\begin{lem}
\label{lem:dirichlet-extension}
Suppose $|D^\alpha u|\leq C(\alpha!)^s/r^{|\alpha|}$ on $\Omega$, with $s\geq1$. We set $\mu=2^{-d/2}\norm{u}_{L^2(\Omega)}$ and $\tau=\max(s,2)$. For a nonempty set $I\subseteq\{1,\ldots,d\}$ and signs $z\in\{-1,1\}^I$, let $F_{I,z}$ be the face on which $x_i=z_i$ for $i\in I$. We define
\begin{equation}
\label{eq:trace-constant}
\nu=\max\left\{1,\ \sup_{I,z,\,\operatorname{supp}\alpha\subseteq I}
\left(\frac{r^{|\alpha|}\norm{D^\alpha u}_{L^2(F_{I,z})}}
{(\alpha!)^s\norm{u}_{L^2(\Omega)}}\right)^{1/|I|}\right\},
\end{equation}
where the norm on a vertex is absolute value, and take
\begin{equation}
\label{eq:collar-width}
\delta=\min\{r,1,1/(9d\nu^2)\}.
\end{equation}
There is a periodic extension $\bar u=u$ on $\Omega$, vanishing outside $[-1-\delta,1+\delta]^d$, with Gevrey parameters
\begin{equation}
\label{eq:dirichlet-gevrey-parameters}
\bar C=2^d C,
\quad
\bar r=\frac{(\tau-1)\delta}{64\tau e^\tau},
\quad \text{and} \quad \bar s=\tau.
\end{equation}
Its mass satisfies
\begin{equation}
\label{eq:dirichlet-extension-norm}
\norm{\bar u}_{L^2(\mathbb T^d)}
\leq(1+8\nu^2\delta)^{d/2}\norm{u}_{L^2(\Omega)}
\leq e^{4/9}\norm{u}_{L^2(\Omega)}.
\end{equation}
\end{lem}

\begin{proof}
We use the ramp functions constructed above with $\tau$ and $\delta$ as in the statement. For $y\geq0$, we extend across the face $x_1=1$ by
\begin{equation}
u_{\mathrm{ext}}(1+y,x')=u(1,x')+\sum_{k\geq1}\frac{\partial_1^ku(1,x')}{k!}y^k\chi_k(y).
\end{equation}
With $P_\ell(k)$ from \eqref{eq:auxiliary-cutoff-derivatives}, for $j\geq\ell$ we have
$k^{\tau(j-\ell)}P_\ell(k)\leq e^{\tau k}(j!)^\tau$: for $\ell\leq k$ we use $k^j\leq e^kj!$, and for $\ell>k$ we use $k^{j-\ell}\ell!\leq j!$ and $k^k/k!\leq e^k$. The differentiated Taylor terms are consequently bounded by
$CD_0^{-k}(4D_0q/\delta)^j\binom j\ell k^{\tau(j-\ell)}P_\ell(k)$.
Summing the geometric series and the binomial coefficients bounds a mixed derivative of the one-coordinate extension by $2C(j!)^\tau(8D_0q/\delta)^j(\alpha'!)^s/r^{|\alpha'|}$. Thus the transverse radius remains $r$. Uniform convergence of these derivative series matches every boundary jet. The construction at $-1$ is identical after reflection.

Applying these extensions in each coordinate gives $u_{\mathrm{ext}}$ with prefactor $2^dC$ and radius $\delta/(8D_0q)$. Its magnitude at a point with exactly $k$ coordinates outside $[-1,1]$ is at most $2^kC$. We set $\bar u=\chi u_{\mathrm{ext}}$, with $\chi$ from \eqref{eq:dirichlet-cutoff-function}. Its derivative bound, $\zeta(\tau)\leq q$, and the product rule give \eqref{eq:dirichlet-gevrey-parameters}. The extension is zero near the periodic faces, so its periodic repetition is smooth.

On a collar region outside precisely the coordinates in $I$, the extension is a multiple Taylor series of the normal jets on $F_{I,z}$. We apply Minkowski's inequality directly to this series and \eqref{eq:trace-constant}. With $h_0=\delta$, each positive-order factor satisfies
\begin{equation}
(k!)^{s-1}(h_k/r)^k\sqrt{h_k}\leq D_0^{-k}\sqrt\delta.
\end{equation}
The sum in each exterior coordinate is at most $2\sqrt\delta$. Hence the squared mass of that region is at most $(4\nu^2\delta)^{|I|}\norm{u}_{L^2(\Omega)}^2$. Summing over the $2^{|I|}$ choices of faces and all subsets $I$ gives $(1+8\nu^2\delta)^d\norm{u}_{L^2(\Omega)}^2$. Multiplication by $\chi$ does not increase the norm, and $8d\nu^2\delta\leq8/9$ proves the claim.
\end{proof}

\begin{rmk}
\label{rmk:collar-selection}
The cap $\delta\leq r$ controls the boundary Taylor series, $\delta\leq1$ keeps the extension inside the periodic cell, and $\delta\leq1/(9d\nu^2)$ bounds the collar mass. For fixed $r$ and $s$, the radius $\bar r$ is of order $1/(d\nu^2)$. We use the constant
\begin{equation}
\label{eq:collar-mass-bound}
\Lambda=4e^{4/9}
\end{equation}
for the discrete mass bound proved below.
\end{rmk}

\begin{rmk}
\label{rmk:dimension-sensitivity}
The parameter $\nu$ controls mixed normal jets on faces of every codimension, not only individual boundary faces. For a product solution it is bounded by the corresponding one-dimensional trace parameter, independently of $d$. The pointwise Gevrey bound ensures only $\nu\leq\max\{1,C/\mu\}$, which can grow exponentially in $d$. Polynomial dimension bounds therefore retain $\nu$ explicitly. The direct collar-region estimate above does not assume that sequential extensions preserve trace-to-volume ratios.
\end{rmk}

\subsection{Dirichlet discretization and resource bounds}

We write $\mathcal I_Nv$ for the trigonometric interpolant of a vector on $\Gamma_N$. The trial space has bandwidth $N$, while the evaluation grid $\Gamma_M$ satisfies
\begin{equation}
\label{eq:evaluation-cutoff}
M\geq c_0dN^2,
\end{equation}
for a sufficiently large absolute constant $c_0$. We evaluate the PDE on $\Gamma_M|_\Omega$ and the boundary values on each face with tangential coordinates in $\Gamma_M|_{[-1,1]}$. The output map is the restricted interpolation isometry
\begin{equation}
\label{eq:dirichlet-output-map}
R_Nw=\left(\frac{2N+1}{2M+1}\right)^{d/2}
(\mathcal I_Nw)|_{\Gamma_M\cap\Omega},
\quad \text{so that} \quad \norm{R_N}\leq1.
\end{equation}
We use $R_Nw$ for the weighted output and $u_M|_\Omega$ for the unweighted grid values of the exact solution. In this section, $\mathbb L_N$ denotes the rectangular matrix whose interior rows evaluate $\mathcal L\mathcal I_Nv$ with weight $((2N+1)/(2M+1))^{d/2}$, and whose boundary rows evaluate $\mathcal I_Nv$ with weight
$((2N+1)/(2\pi))^{1/2}((2N+1)/(2M+1))^{(d-1)/2}$.
The vector $\eta_N$ contains the source and boundary data with the same weights. Oracle access to $\eta_N$ includes the relative norms needed to combine source and boundary data.

These weights make the two blocks discrete counterparts of $\norm{\mathcal L\mathcal I_Nv}_{L^2(\Omega)}$ and $\norm{\mathcal I_Nv}_{L^2(\partial\Omega)}$ carried at the common scale $h_N^{-d}$, with $h_N=2\pi/(2N+1)$ as above. Scaling the boundary block to be a contraction instead would suppress it by a factor $\Theta(N\sqrt d)$ relative to the interior block, and the stability constant below would inherit that factor without any corresponding gain. The rebalanced boundary block $T$ satisfies
\begin{equation}
\label{eq:dirichlet-boundary-norm}
\norm{T}\leq\left(\frac{d(2N+1)}\pi\right)^{1/2},
\end{equation}
since exact quadrature in the tangential coordinates and the Cauchy--Schwarz inequality over the $2N+1$ normal frequencies bound each face trace by $(2N+1)/(2\pi)$ times the squared norm of the trial vector. For $N\geq d$ this is dominated by the interior normalization whenever $\ceil{\mathcal J}\geq1$, so the stacked operator still has a block encoding with normalization at most $|\mathcal J|\ceil{g}N^{\ceil{\mathcal J}}+1$ after rescaling by a constant.

Because these rows constrain a trial vector only on $\Omega$ and on $\partial\Omega$, the smallest nonzero singular value of $\mathbb L_N$ is not the stability constant of the method. A trigonometric polynomial of degree $N$ concentrated away from $\Omega$ is nearly invisible to every row. For $d=1$ and $\mathcal L=-\partial_x^2$, the degree-$N$ polynomial $p_N(x)=\sin^{2N}(x/2)$ satisfies $\norm{\mathbb L_Np_N}/\norm{p_N}=\mathcal O(N^{3}a^{N-1})$ with $a=\sin^2(1/2)$, and applying it to a single coordinate gives the same bound in every dimension, so $\sigma_{\min}^+(\mathbb L_N)$ decays exponentially in $N$. This is the discrete counterpart of the ill-conditioning of classical Fourier extension \cite{adcock2014numerical}, and it does not degrade the accuracy of the method, because the directions responsible for it are also invisible in the output. We therefore work with the restricted stability constant
\begin{equation}
\label{eq:dirichlet-restricted-gap}
\sigma_\Omega(\mathbb L_N)=\min\left\{\frac{\norm{\mathbb L_Nw}}{\norm{R_Nw}}\ :\ w\perp\ker(\mathbb L_N),\ R_Nw\neq0\right\}\geq\sigma_{\min}^+(\mathbb L_N),
\end{equation}
which is the discrete form of an a priori estimate $\norm{w}_{L^2(\Omega)}\lesssim\norm{\mathcal Lw}_{L^2(\Omega)}+\norm{w}_{L^2(\partial\Omega)}$ of the kind available for well-posed boundary value problems. The next proposition makes that correspondence quantitative, and in particular removes the discretized operator from the hypothesis.

\begin{prop}
\label{prop:restricted-stability}
We write $\mathcal V_N$ for the trigonometric polynomials of degree at most $N$ in each variable, and set
\begin{equation}
\label{eq:continuous-stability-constant}
\Theta_N=\sup\left\{\frac{\norm{v}_{L^2(\Omega)}}{\norm{\mathcal Lv}_{L^2(\Omega)}+\norm{v}_{L^2(\partial\Omega)}}\ :\ v\in\mathcal V_N,\ v|_\Omega\neq0\right\}.
\end{equation}
Suppose $\Theta_N<\infty$. We let $\gamma_N^->0$ satisfy
$h_M^{\,d}\sum_{y\in\Gamma_M|_\Omega}|p(y)|^2\geq(\gamma_N^-)^2\norm{p}_{L^2(\Omega)}^2$
for all $p\in\mathcal L\mathcal V_N$, together with its analogue for the traces of $\mathcal V_N$ on each face. We let $\gamma_N^+$ satisfy
$h_M^d\sum_{x\in\Gamma_M|_\Omega}|v(x)|^2\leq(\gamma_N^+)^2\norm{v}^2_{L^2(\Omega)}$
for all $v\in\mathcal V_N$. Then
\begin{equation}
\label{eq:restricted-stability-bound}
\sigma_\Omega(\mathbb L_N)\geq\frac{\gamma_N^-}{\sqrt2\,\gamma_N^+\,\Theta_N}.
\end{equation}
\end{prop}

\begin{proof}
We write $v=\mathcal I_Nw$. The interior weight gives
$\norm{\mathbb L_Nw}^2_{\mathrm{int}}=h_N^{-d}h_M^{\,d}\sum_{y\in\Gamma_M|_\Omega}|\mathcal Lv(y)|^2$, because $((2N+1)/(2M+1))^dh_M^{-d}=h_N^{-d}$. The boundary weight gives the same identity on each face with $h_M^{\,d-1}$ and the surface measure. Applying the assumed lower interpolation inequalities to both blocks,
\begin{equation}
\norm{\mathbb L_Nw}^2\geq h_N^{-d}(\gamma_N^-)^2\left(\norm{\mathcal Lv}^2_{L^2(\Omega)}+\norm{v}^2_{L^2(\partial\Omega)}\right),
\end{equation}
while \eqref{eq:dirichlet-output-map} and the upper interpolation inequality give $\norm{R_Nw}^2\leq h_N^{-d}(\gamma_N^+)^2\norm{v}^2_{L^2(\Omega)}$. Combining the two displays with \eqref{eq:continuous-stability-constant} and $(a+b)^2\leq2(a^2+b^2)$ proves \eqref{eq:restricted-stability-bound}.
\end{proof}

The constant $\Theta_N$ is controlled by an a priori estimate for the boundary value problem, not by uniqueness alone. We bound it for Poisson in \cref{cor:poisson-dirichlet}.

\begin{rmk}
\label{rmk:mz-constants}
For constant-coefficient operators, $\mathcal L\mathcal V_N\subseteq\mathcal V_N$. The one-dimensional Markov inequality on $[-1,1]$ bounds the derivative of a degree-$N$ trigonometric polynomial in $L^2$ by $\mathcal O(N^2)$ times its $L^2$ norm. Applying a first-order quadrature bound to its absolute square, including the endpoint cells, gives squared-norm distortion $1+\mathcal O(h_MN^2)$. Tensorization with \eqref{eq:evaluation-cutoff} therefore gives $\gamma_N^\pm=e^{\pm\mathcal O(1/c_0)}$, also on the faces, once $c_0$ is sufficiently large. Variable coefficients need the lower inequality on the residual space $\mathcal L\mathcal V_N$ separately; it is not implied by an interpolation bound on $\mathcal V_N$ alone.
\end{rmk}

We regularize the small singular values instead of projecting the source onto an exactly retained subspace. For a scale $t>0$, we use
\begin{equation}
\label{eq:dirichlet-filtered-inverse}
F_t=\mathbb L_N^*(\mathbb L_N\mathbb L_N^*+t^2I)^{-1},
\quad \text{with singular multiplier} \quad
f_t(\sigma)=\frac{\sigma}{\sigma^2+t^2}.
\end{equation}
For a block encoding with normalization $\alpha$, the odd function $q_t(x)=t\alpha x/(2(\alpha^2x^2+t^2))$ is bounded by $1/4$ on $[-1,1]$. Its poles are at $\pm it/\alpha$, so truncation of its Chebyshev series on an ellipse with imaginary semiaxis $t/(2\alpha)$ gives an odd polynomial of degree $\mathcal O((\alpha/t)\log(\alpha/(t\zeta)))$ and uniform error $\zeta$. For $\zeta<1/4$ this polynomial is bounded by one and quantum singular value transformation \cite{gilyen2019quantum} implements an approximation to $(t/2)F_t$. Its successful branch is subsequently restricted by $R_N$ and amplified. This amplification is included in the resource bound below. The complete procedure is summarized in \cref{alg:dirichlet-pde}.

\begin{algorithm}[H]
\caption{Quantum algorithm for linear PDEs with Dirichlet boundary conditions}
\label{alg:dirichlet-pde}
\begin{algorithmic}[1]

\Require \,
Target precision $\epsilon\in(0,1)$; coefficient block encodings and combined source and boundary-data oracles with their relative norms;\newline
\hspace*{7.3mm} The hypotheses and parameter bounds of \cref{thm:dirichlet-complexity}, including a positive lower bound on $\sigma_\Omega(\mathbb L_N)$ and the kernel condition \eqref{eq:dirichlet-kernel-condition} at the chosen cutoff.

\State Choose $N\geq d$ satisfying \eqref{eq:dirichlet-cutoff} and set $M=\lceil c_0dN^2\rceil$ as in \eqref{eq:evaluation-cutoff}.
\State Set $\alpha_N=2(|\mathcal J|\ceil{g}N^{\ceil{\mathcal J}}+1)$ and $t=\epsilon\sigma_\Omega/(32\Lambda)$, with $\Lambda$ from \eqref{eq:collar-mass-bound}.
\State Construct the block encoding of $\mathbb L_N$ with normalization $\alpha_N$ and prepare $\ket{\eta_N}$.
\State Apply quantum singular value transformation with an odd polynomial approximating $q_t$ at $\alpha=\alpha_N$ to uniform error $\zeta\leq\epsilon t/(512\alpha_N\Lambda)$, implementing an approximation to $(t/2)F_t$.
\State Apply $R_N$ coherently by Fourier interpolation to $\Gamma_M$ and flagging the points in $\Omega$.
\State Amplify the joint singular value transformation and interior success flags using $\mathcal O(\alpha_N\Lambda/t)$ calls to the complete preparation circuit and its inverse, with implementation errors allocated as in the proof of \cref{thm:dirichlet-complexity}.

\Ensure $N, M \in \mathbb N$, and a normalized state $\epsilon$-close to $(u_M|_\Omega)/\norm{u_M|_\Omega}$ with constant success probability.
\end{algorithmic}
\end{algorithm}

\begin{thm}
\label{thm:dirichlet-complexity}
Suppose the Dirichlet problem $\mathcal Lu=\eta$ on $\Omega$, $u|_{\partial\Omega}=b$, has a solution as in \cref{lem:dirichlet-extension}, and $\ceil{\mathcal J}\geq1$. We take $\ceil{g}=\max\{1,\max_{x\in\Omega}\sum_{\alpha\in\mathcal J}|g_\alpha(x)|\}$. At the chosen cutoff, assume
\begin{equation}
\label{eq:dirichlet-kernel-condition}
R_NP_{\ker(\mathbb L_N)}\bar u_N=0.
\end{equation}
Given the coefficient and combined-data oracles above and a positive lower bound on $\sigma_\Omega=\sigma_\Omega(\mathbb L_N)$, \cref{alg:dirichlet-pde} with $t=\epsilon\sigma_\Omega/(32\Lambda)$ returns a normalized state $\epsilon$-close to $(u_M|_\Omega)/\norm{u_M|_\Omega}$ with constant success probability. Writing $\alpha_N=2(|\mathcal J|\ceil{g}N^{\ceil{\mathcal J}}+1)$, its query count is
\begin{equation}
\label{eq:dirichlet-query-count}
\tilde{\mathcal O}\!\left(
\frac{\alpha_N^2}{\epsilon^2\sigma_\Omega^2}\log\frac1\epsilon\right).
\end{equation}
The elementary-gate count is $\tilde{\mathcal O}(d(|\mathcal J|+\ceil{\mathcal J}))$ times this bound, and the qubit count is
\begin{equation}
\mathcal O\!\left(d(\log M+\log\ceil{\mathcal J})+
\log\left(2+\frac{\alpha_N}{\epsilon^2\sigma_\Omega}\right)\right),
\end{equation}
excluding oracle workspace. Here $\Lambda$ is the absolute constant in \eqref{eq:collar-mass-bound}, and $\tau$, $\mu$, and $\bar r$ are given in \cref{lem:dirichlet-extension}. With $M=\lceil c_0dN^2\rceil$, a sufficient Fourier mode cutoff has scale
\begin{equation}
\label{eq:dirichlet-cutoff-scale}
N\in\tilde{\mathcal O}\!\left(\frac1{\bar r}
\left[\tau(d+\ceil{\mathcal J})+
\log\left(2+\frac{C(|\mathcal J|\ceil{g}+\sqrt{2d}+1)}
{\mu\epsilon\min\{1,\sigma_\Omega\}}\right)\right]^\tau\right).
\end{equation}
The stability and kernel hypotheses are evaluated at this cutoff.
\end{thm}

\begin{proof}
Fourier transformation, zero-padding, and inverse Fourier transformation implement the interpolation isometry from $\Gamma_N$ to $\Gamma_M$. Coordinate comparisons give $R_N$ and select the interior equation rows. Fourier phase bras in the normal coordinate and interpolation in the tangential coordinates give the boundary block. By \eqref{eq:dirichlet-boundary-norm} and $N\geq d$, the stacked operator has normalization at most $\alpha_N$. Its per-call gates are $\tilde{\mathcal O}(d(|\mathcal J|+\ceil{\mathcal J}))$.

We record a uniform interpolation bound because the boundary points need not lie on $\Gamma_N$. We set $p=d+\ceil{\mathcal J}$, $a=(\tau/e)\bar r^{1/\tau}$, and
\begin{equation}
H=4d3^{d-1}e^\tau\tau\Gamma(\tau p)
\left(\frac{4e}{\tau}\right)^{\tau p}\bar r^{-p}.
\end{equation}
Integration by parts in a coordinate of maximal frequency, with derivative order $\lfloor(\bar r\norm{\omega}_\infty)^{1/\tau}/e\rfloor$, bounds the Fourier coefficient by $\bar C e^\tau e^{-a\norm{\omega}_\infty^{1/\tau}}$. Aliasing changes a derivative of order at most $\ceil{\mathcal J}$ by at most twice its omitted weighted Fourier tail. Counting the shells by $2d3^{d-1}n^{d-1}$ and integrating the remaining moment therefore gives
\begin{equation}
\max_{|\alpha|\leq\ceil{\mathcal J}}
\norm{D^\alpha(\mathcal I_N\bar u_N-\bar u)}_\infty
\leq E_N:=\bar C H e^{-aN^{1/\tau}/2}.
\end{equation}
We choose $c_0$ so that the lower interpolation constant on $\mathcal V_N$ is at least $3/4$, and take $N$ large enough that
\begin{equation}
\label{eq:dirichlet-grid-condition}
E_N\leq\frac{\epsilon\min\{1,\sigma_\Omega\}\norm{u}_{L^2(\Omega)}}
{512(2\pi)^{d/2}(|\mathcal J|\ceil{g}+\sqrt{2d}+1)}.
\end{equation}
Writing $w_N=\bar u_N$, we set
\begin{equation}
y_N=\left(\frac{2N+1}{2M+1}\right)^{d/2}u_M|_\Omega,
\quad U_N=\norm{y_N},
\quad\text{and}\quad \xi_N=\mathbb L_Nw_N-\eta_N.
\end{equation}
The triangle inequality and the lower interpolation bound applied to $\mathcal I_Nw_N$, followed by exact full-period quadrature and \eqref{eq:dirichlet-extension-norm}, give
\begin{equation}
U_N\geq\tfrac12h_N^{-d/2}\norm{u}_{L^2(\Omega)},
\quad\text{and}\quad \norm{w_N}\leq\Lambda U_N.
\end{equation}
Counting at most $(2+h_M)/h_M$ interior points per coordinate, with $h_M\leq1$, gives
\begin{equation}
\norm{R_Nw_N-y_N}\leq\epsilon U_N/128,
\quad\text{and}\quad
\norm{\xi_N}\leq\epsilon\sigma_\Omega U_N/128.
\end{equation}
Thus the discrete bounds follow from uniform Fourier approximation, without a pointwise quadrature estimate for an arbitrary smooth function. An explicit sufficient Fourier cutoff is
\begin{equation}
\label{eq:dirichlet-cutoff}
N\geq\frac{(2e/\tau)^\tau}{\bar r}
\left[\log\frac{512\pi^{d/2}\bar CH(|\mathcal J|\ceil{g}+\sqrt{2d}+1)}
{\epsilon\min\{1,\sigma_\Omega\}\mu}\right]^\tau,
\end{equation}
together with $N\geq d$; logarithms are bounded below by one.

We write $v_t=F_t\eta_N$ and $w_N^\perp=(I-P_{\ker(\mathbb L_N)})w_N$. The regularized solution lies in the row space. On each singular triplet, $\sigma t^2/(\sigma^2+t^2)\leq t/2$ and $\sigma^2/(\sigma^2+t^2)\leq1$, so
\begin{equation}
\norm{\mathbb L_N(v_t-w_N^\perp)}\leq\frac t2\norm{w_N}+\norm{\xi_N}.
\end{equation}
The kernel hypothesis and restricted stability therefore imply
\begin{equation}
\norm{R_Nv_t-y_N}
\leq\frac{t\Lambda U_N/2+\norm{\xi_N}}{\sigma_\Omega}
+\norm{R_Nw_N-y_N}\leq\epsilon U_N/32.
\end{equation}
Also $\norm{\eta_N}\leq\alpha_N\norm{w_N}+\norm{\xi_N}\leq2\alpha_N\Lambda U_N$. We implement $(t/2)F_t$ with uniform polynomial error $\zeta\leq\epsilon t/(512\alpha_N\Lambda)$ and then apply $R_N$. Its successful branch has amplitude at least $tU_N/(4\norm{\eta_N})$, up to the allocated approximation error. Amplitude amplification therefore uses $\mathcal O(\alpha_N\Lambda/t)$ calls to this entire preparation circuit and its inverse. Combining this with the polynomial degree gives
\begin{equation}
\tilde{\mathcal O}\!\left(\frac{\alpha_N^2\Lambda}{t^2}\log\frac1\epsilon\right)
=\tilde{\mathcal O}\!\left(\frac{\alpha_N^2\Lambda^3}{\epsilon^2\sigma_\Omega^2}\log\frac1\epsilon\right)
\end{equation}
queries. Polynomial, block-encoding, and amplification errors can be chosen so that their combined contribution after normalization is at most the remaining error budget; the extra precision changes only logarithmic factors. No exact spectral support or source-consistency assumption is needed. Since $\Lambda$ is an absolute constant, this proves \eqref{eq:dirichlet-query-count}. Finally, the stated bound on $H$ gives \eqref{eq:dirichlet-cutoff-scale}, and the interpolation and arithmetic registers give the workspace estimate, completing the proof.
\end{proof}

\section{Application to the Poisson Equation}
\label{sec:poisson}

Many problems in classical and quantum physics reduce to second-order PDEs. This includes the simulation of material properties, crystals, condensed-matter physics, and molecular dynamics. Here we illustrate our framework with the Poisson equation, in both its isotropic and anisotropic forms, under periodic and Dirichlet boundary conditions.

The (inhomogeneous) steady-state Schr\"odinger equation is of the general form
\begin{equation}
\label{eq:schrodinger}
\left(-\frac{1}2\nabla^2 + V\right)u= \eta
\end{equation}
wherein a nonzero source term $\eta$ is typically used to probe the linear response of the system.

It is therefore useful to first consider the simpler case of the Poisson equation
\begin{equation}
\label{eq:def-poisson}
\nabla^2 u= \eta,
\end{equation}
or its anisotropic variants
\begin{equation}
\label{eq:def-anisotropic-poisson}
\nabla \cdot (\Sigma \nabla u) = \sum_{i, j}  \sigma_{ij}  \tfrac{\partial^2}{\partial x_i \partial x_j} u= \eta.
\end{equation}
Here $\Sigma = (\sigma_{ij})$ is the matrix of second-order coefficients. We refer the reader to \cref{app:elliptic_literature_review} for a summary of prior art in solving the Poisson equation. In the next corollary, we derive its resource complexity as a direct consequence of our general framework.

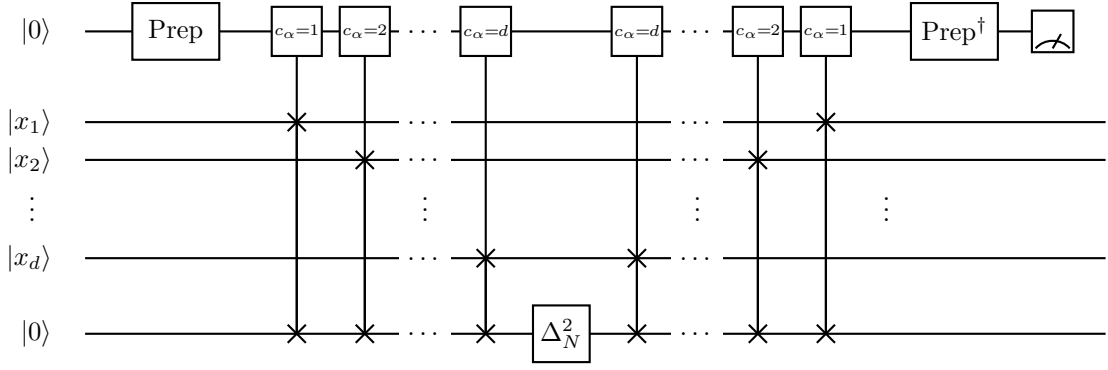
\begin{figure}
\centering
\begin{adjustbox}{max width = \textwidth}

\begin{tikzpicture}[
    thick,
    gate/.style={draw, minimum size=0.65cm, inner sep=1pt, font=\tiny, fill=white},
    dngate/.style={draw, minimum size=0.75cm, inner sep=1pt, font=\small, fill=white},
    statebox/.style={draw, minimum width=1.2cm, minimum height=0.55cm, inner sep=2pt, font=\small, fill=white},
    hgate/.style={draw, minimum size=0.55cm, inner sep=1pt, font=\small, fill=white},
    metergate/.style={draw, minimum width=0.55cm, minimum height=0.55cm, inner sep=0pt, fill=white,
        path picture={
            \draw ([shift={(0.05,0.05)}]path picture bounding box.south west) 
                  to[out=50,in=130] 
                  ([shift={(-0.05,0.05)}]path picture bounding box.south east);
            \draw ([yshift=0.05cm]path picture bounding box.south) -- ++(60:0.22cm);
        }},
    wdots/.style={fill=white, inner sep=3pt, font=\small, outer sep=0pt},
    every node/.style={font=\small},
]

\def\yctrl{0}        
\def\ybone{-1.2}     
\def\ybtwo{-1.7}     
\def\ybd{-3.0}       
\def\yanc{-4}      

\def\cs{0.12}        

\newcommand{\swappair}[3]{%
    \draw (#1-\cs, #2-\cs) -- (#1+\cs, #2+\cs);%
    \draw (#1-\cs, #2+\cs) -- (#1+\cs, #2-\cs);%
    \draw (#1-\cs, #3-\cs) -- (#1+\cs, #3+\cs);%
    \draw (#1-\cs, #3+\cs) -- (#1+\cs, #3-\cs);%
    \draw (#1, #2) -- (#1, #3);%
}

\node[anchor=east] at (-0.3, \yctrl) {$\ket{0}$};
\node[anchor=east] at (-0.3, \ybone) {$\ket{x_1}$};
\node[anchor=east] at (-0.3, \ybtwo) {$\ket{x_2}$};
\node[anchor=east] at (-0.3, \ybd)   {$\ket{x_d}$};
\node[anchor=east] at (-0.3, \yanc)  {$\ket{0}$};


\pgfmathsetmacro{\pH}{1.2}              
\pgfmathsetmacro{\pState}{1.5}         
\pgfmathsetmacro{\pStateEnd}{\pState+0.6}

\pgfmathsetmacro{\pLa}{\pStateEnd+0.7}  
\pgfmathsetmacro{\pLb}{\pLa+0.9}        
\pgfmathsetmacro{\pLdots}{\pLb+0.8}     
\pgfmathsetmacro{\pLc}{\pLdots+0.8}     

\pgfmathsetmacro{\pDelta}{\pLc+1.0}

\pgfmathsetmacro{\pRc}{\pDelta+1.0} 
\pgfmathsetmacro{\pRdots}{\pRc+0.8}      
\pgfmathsetmacro{\pRb}{\pRdots+0.8}      
\pgfmathsetmacro{\pRa}{\pRb+0.9}         

\pgfmathsetmacro{\pPrepDag}{\pRa+1.7}
\pgfmathsetmacro{\pMeasEnd}{\pPrepDag+1.3}

\pgfmathsetmacro{\wend}{\pMeasEnd+0.7}

\draw (0, \yctrl) -- (\pMeasEnd, \yctrl);
\draw (0, \ybone) -- (\wend, \ybone);
\draw (0, \ybtwo) -- (\wend, \ybtwo);
\draw (0, \ybd)   -- (\wend, \ybd);
\draw (0, \yanc)  -- (\wend, \yanc);

\node[draw, minimum width=1.15cm, minimum height=0.75cm, inner sep=1pt, font=\small, fill=white] at (\pH, \yctrl) {Prep};


\node[gate] at (\pLa, \yctrl) {$c_\alpha\!\!=\!\!1$};
\draw (\pLa, {\yctrl-0.325}) -- (\pLa, \yanc);
\swappair{\pLa}{\ybone}{\yanc}

\node[gate] at (\pLb, \yctrl) {$c_\alpha\!\!=\!\!2$};
\draw (\pLb, {\yctrl-0.325}) -- (\pLb, \yanc);
\swappair{\pLb}{\ybtwo}{\yanc}

\node[gate] at (\pLc, \yctrl) {$c_\alpha\!\!=\!\!d$};
\draw (\pLc, {\yctrl-0.325}) -- (\pLc, \yanc);
\swappair{\pLc}{\ybd}{\yanc}

\node[dngate] at (\pDelta, \yanc) {$\Delta_N^2$};


\node[gate] at (\pRc, \yctrl) {$c_\alpha\!\!=\!\!d$};
\draw (\pRc, {\yctrl-0.325}) -- (\pRc, \yanc);
\swappair{\pRc}{\ybd}{\yanc}

\node[gate] at (\pRb, \yctrl) {$c_\alpha\!\!=\!\!2$};
\draw (\pRb, {\yctrl-0.325}) -- (\pRb, \yanc);
\swappair{\pRb}{\ybtwo}{\yanc}

\node[gate] at (\pRa, \yctrl) {$c_\alpha\!\!=\!\!1$};
\draw (\pRa, {\yctrl-0.325}) -- (\pRa, \yanc);
\swappair{\pRa}{\ybone}{\yanc}

\node[draw, minimum width=1.15cm, minimum height=0.75cm, inner sep=1pt, font=\small, fill=white] at (\pPrepDag, \yctrl) {Prep$^\dagger$};
\node[metergate] at (\pMeasEnd, \yctrl) {};

\foreach \yy in {\yctrl, \ybone, \ybtwo, \ybd, \yanc} {
    \node[wdots] at (\pLdots, \yy) {$\dots$};
    \node[wdots] at (\pRdots, \yy) {$\dots$};
}

\pgfmathsetmacro{\vdotY}{(\ybtwo + \ybd) / 2 + 0.1}
\foreach \xp in {-0.7, \pLdots, \pRdots, 10.6} {
    \node at (\xp, \vdotY) {$\vdots$};
}

\end{tikzpicture}
\end{adjustbox}
\caption{The circuit for the block encoding of the (axis-aligned, i.e.\ diagonal-$\Sigma$) Laplacian operator. The \texttt{Prep} box prepares the weighted superposition $\sum_{i=1}^d\sqrt{\sigma_{ii}/\mathrm{tr}(\Sigma)}\ket{c_\alpha\!=\!i}$; \texttt{Prep}$^\dagger$ and the final measurement complete the LCU, heralding success on outcome $\ket{0}$. For general (non-diagonal) $\Sigma$, the cross-derivative terms $\sigma_{ij}\partial_i\partial_j$, $i\neq j$, are instead handled by the general LCU machinery of \cref{fig:L-circuit,fig:D-alpha-circuit}.}
\label{fig:Laplacian-circuit}
\end{figure}

Although the continuous Poisson source has zero mean, its discretization need not, so we rely on the range projector in \cref{thm:linear-pde-complexity} and the corresponding consistent equation becomes
\begin{equation}
\label{eq:projected-discrete-poisson}
\mathbb L_Nv_N=P_N\eta_N,
\qquad \text{where} \quad
P_N=I-\frac{\mathbf1\mathbf1^T}{(2N+1)^d},
\end{equation}
where $\mathbf1$ is the all-ones vector and $P_N$ removes the discrete mean. We seek the minimum-norm solution $v_N=\mathbb L_N^+P_N\eta_N$. The following corollary bounds its approximation error relative to the original discretized solution $u_N$, including the cost of preparing the projected source.

\begin{cor}
\label{cor:poisson}
Suppose that $\Sigma$ is symmetric positive definite, and that the minimum-norm solution to the anisotropic Poisson equation \eqref{eq:def-anisotropic-poisson} is $s$-Gevrey with radius $r$ for some $s \geq 1$. Then, for $0<\epsilon<1$, \cref{alg:linear-pde} applied to \eqref{eq:projected-discrete-poisson} returns, with constant success probability, a normalized state $\ket{v} \in \mathbb{C}^{(2N + 1)^d}$ which is $\epsilon$-close in $\ell^2$-norm to $\frac{u_N}{\norm{u_N}}$ using
\begin{equation}
\label{eq:total-gate-count-poisson}
\tilde{\mathcal O} \left(\frac{\norm{\Sigma}_{1, 1}}{\lambda_{\min}(\Sigma)r^2}
s^{2s} d^{2s+2} \log^{2s} \left(\frac{1}{\norm{\hat u}\, \epsilon} \right) \log (\frac{1}\epsilon)
\right)
\end{equation}
elementary gates, $\mathcal O(d\log N + \log(\kappa_{\mathbb L_N}/\epsilon))$ qubits (excluding the coefficient and source oracle qubits), and at most
\begin{equation}
\label{eq:query-count-poisson}
\tilde{\mathcal O} \left(\frac{\norm{\Sigma}_{1, 1}}{\lambda_{\min}(\Sigma)r^2}
s^{2s} d^{2s+1} \log^{2s} \left(\frac{1}{\norm{\hat u}\, \epsilon}\right) \log(\frac{1}\epsilon)\right)
\end{equation}
queries to $U_{\eta_N}$. Here $\norm{\Sigma}_{1, 1}=\sum_{i,j}\abs{ \sigma_{ij} }$ is the entry-wise $\ell_1$-norm, and $\lambda_{\min}(\Sigma)$ is the smallest eigenvalue of $\Sigma$.
\end{cor}

\begin{proof}
We choose $N$ according to \eqref{eq:complexity-cutoff} with $\epsilon$ replaced by $\epsilon/2$, $\ceil{\mathcal J} = 2$, and $\sigma_{\min}^+$ replaced by $\lambda_{\min}(\Sigma)$. For diagonal $\Sigma$, the circuit in \cref{fig:Laplacian-circuit} gives a
\begin{equation}
\label{eq:laplacian-cost}
\mathcal O(d \log d \log N+ \log^2 N + R_N )
\end{equation}
complexity for implementing the discretized Laplacian operator from elementary gates, which simplifies to $\tilde{\mathcal O}(d \log N)$. For general $\Sigma$, we prepare pair indices $(i,j)$ with amplitudes $\sqrt{|\sigma_{ij}|/\norm{\Sigma}_{1,1}}$ and select two derivative blocks on coordinates $i$ and $j$, with the sign of $\sigma_{ij}$ included in the selected block. This adds at most $\mathcal O(d^2)$ preparation gates and uses two instances of the same routing network, giving a per-query cost $\tilde{\mathcal O}(d^2+d\log N)$. The coefficients $\sigma_{ij}$ constitute the constant functions $g_\alpha$ in our previous notation and therefore
\begin{equation}
\ceil{g}= \sum_{i, j=1}^d |\sigma_{ij}|= \norm{\Sigma}_{1, 1}.
\end{equation}
Additionally, neglecting the single zero eigenvalue of $\mathbb L_N$ that arises from $\mathbf{0} \in \Lambda_N$ (hence finding the \emph{effective} condition number), the nonzero eigenvalues of $\mathbb L_N$ are $-\omega^T\Sigma\omega$ for $\omega \in \Lambda_N \setminus \{0\}$, since the Fourier symbol of $\partial_i\partial_j$ is $(i\omega_i)(i\omega_j)=-\omega_i\omega_j$. Since $\Sigma$ is positive definite, the Rayleigh quotient bound $\omega^T\Sigma\omega \geq \lambda_{\min}(\Sigma) \norm{\omega}^2 \geq \lambda_{\min}(\Sigma)$ holds for every nonzero integer vector $\omega$, giving
\begin{align}
\sigma_{\min}^+(\mathbb L_N) \geq \lambda_{\min}(\Sigma).
\end{align}
This bound is tight when $\Sigma$ is diagonal, attained at $\omega=e_{i^*}$ for the index $i^*$ minimizing $\sigma_{ii}$; for a general (non-diagonal) $\Sigma$, $\sigma_{\min}^+(\mathbb L_N)$ may exceed $\lambda_{\min}(\Sigma)$. The pair-indexed LCU has normalization $\alpha_{\mathrm{LCU}}=\norm{\Sigma}_{1,1}N^2$. For diagonal $\Sigma$, this equals $\mathrm{tr}(\Sigma)N^2=\sigma_{\max}(\mathbb L_N)$, so there is no extra normalization penalty in that case.

The continuous minimum-norm solution has zero mean. Consequently, the zero Fourier coefficient in \cref{cor:qft-func-err} gives
\begin{equation}
\norm{(I-P_N)u_N}\leq (2N+1)^{d/2}C e^{-(rN)^{1/s}/2}.
\end{equation}
Writing $\xi_N=\mathbb L_Nu_N-\eta_N$, the kernel identity and our choice of cutoff give
\begin{equation}
u_N-v_N=\mathbb L_N^+\xi_N+(I-P_N)u_N,
\quad \text{and} \quad
\frac{2\norm{u_N-v_N}}{\norm{u_N}}\leq\frac\epsilon2.
\end{equation}
Here the residual and mean terms each contribute at most $\epsilon/4$, using $\ceil{g}=\norm{\Sigma}_{1,1}\geq\lambda_{\min}(\Sigma)$ and the discretized-norm bound from \cref{cor:ultimate-N}.

The same estimates give $\norm{\xi_N}\leq\lambda_{\min}(\Sigma)\norm{u_N}/8$ and $\norm{(I-P_N)u_N}\leq\norm{u_N}/8$, so \eqref{eq:source-projection-error}, applied to $P_Nu_N$, bounds the projection probability below by $1/2$. Here we use the explicit-projection option: we apply $F_N^{\otimes d}$, flag nonzero frequencies, and undo the Fourier transform. Fixed-point amplitude amplification \cite{yoder2014fixed} prepares the projected source coherently to error $\delta$ with $\mathcal O(\log(1/\delta))$ source-oracle calls and $\mathcal O(d\log^2N\log(1/\delta))$ additional gates. This preserves the original cutoff and the QLSA error budget then proves the corollary.
\end{proof}

\begin{rmk}
As in \cref{rmk:ultimate-N-u}, \cref{cor:poisson} yields
\begin{equation}
\tilde{\mathcal O} \left(\frac{\norm{\Sigma}_{1, 1}}{\lambda_{\min}(\Sigma)r^2} s^{2s} d^{4s + 2} \log^{2s} \left(\frac{1}{\norm{u}\, \epsilon} \right) \log (\frac{1}\epsilon) \right)
\end{equation}
elementary gates and
\begin{equation}
\tilde{\mathcal O} \left(\frac{\norm{\Sigma}_{1, 1}}{\lambda_{\min}(\Sigma)r^2} s^{2s} d^{4s + 1} \log^{2s} \left(\frac{1}{\norm{u}\, \epsilon}\right) \log(\frac{1}\epsilon)\right)
\end{equation}
queries, respectively, in terms of the $\ell^2$-norm of the solution.
\end{rmk}

\begin{rmk}
\label{rmk:poisson-solvability}
The Fourier symbol $-\omega^T\Sigma\omega$ vanishes only at $\omega=0$, so $\ker(\mathbb L_N)=\ker(\mathbb L_N^T)=\mathrm{span}\{\mathbf 1\}$. As a result, the equation $\mathbb L_Nv_N=\eta_N$ has a solution if and only if the sum of all entries of $\eta_N$ is zero; whenever this holds, the solution is unique only up to an additive constant. The projection in the proof enforces this condition and the pseudoinverse selects the zero-mean representative, while the mean discarded from the target is included in the error bound. In the continuous domain, the divergence theorem gives
\begin{equation}
\label{eq:divergence-theorem}
\int_{\mathbb T^d} \nabla \cdot (\Sigma \nabla u)\, dx = \int_{\partial \mathbb T^d} (\Sigma \nabla u) \cdot n\, dS = 0,
\end{equation}
and therefore the Poisson equation has a solution only if the source term has a zero mean, $\int_{\mathbb T^d} \eta = 0$; the physically natural (minimum-norm) solution is likewise the one with $\int_{\mathbb T^d} u = 0$.
\end{rmk}

\begin{rmk}
For $0 < s < 1$, the Poisson construction and its discrete solvability conditions are unchanged. The cutoff uses $r \mapsto (s/2)^s r$ in \eqref{eq:complexity-cutoff}, with $w = \max(d, s(d + 2))$; the query and gate bounds in terms of $N$ then apply with the original Gevrey parameters.
\end{rmk}

\begin{cor}
\label{cor:poisson-dirichlet}
Suppose $\Sigma$ is symmetric positive definite and the solution of \eqref{eq:def-anisotropic-poisson} on $\Omega=[-1,1]^d$, with $u|_{\partial\Omega}=b$, has Gevrey parameters $C$ and $r$ and order $s\geq2$. We take $\mu$, the mixed face-jet parameter $\nu$, and $\bar r$ as in \cref{lem:dirichlet-extension}, use the evaluation density \eqref{eq:evaluation-cutoff}, and write
\begin{equation}
\label{eq:poisson-dirichlet-stability}
\Theta_\Sigma=\max\left\{\frac4{d\pi^2\lambda_{\min}(\Sigma)},\ 1\right\}.
\end{equation}
Then the discretization of \cref{sec:dirichlet} prepares a state $\epsilon$-close to $(u_M|_\Omega)/\norm{u_M|_\Omega}$, with constant success probability, using
\begin{equation}
\label{eq:query-count-poisson-dirichlet}
\tilde{\mathcal O}\!\left(
\frac{\Theta_\Sigma^2(\norm{\Sigma}_{1,1}+1)^2}{\epsilon^2\bar r^4}
\left[s(d+2)+\log\left(2+\frac{C\Theta_\Sigma(\norm{\Sigma}_{1,1}+\sqrt{2d}+1)}{\mu\epsilon}\right)\right]^{4s}
\log\frac1\epsilon\right)
\end{equation}
queries to the combined source and boundary-data oracle, $\tilde{\mathcal O}(d^2)$ times this many elementary gates, and $\mathcal O(d\log(dN)+\log(2+\Theta_\Sigma(\norm{\Sigma}_{1,1}+1)N^2/\epsilon^2))$ qubits excluding oracle workspace. In the isotropic case $\Sigma=I$, for fixed $r$ and $s$, we have $\norm{\Sigma}_{1,1}=d$, $\Theta_\Sigma=1$, and $\bar r=\Theta(1/(d\nu^2))$, giving
\begin{equation}
\label{eq:query-count-poisson-dirichlet-isotropic}
\tilde{\mathcal O}\!\left(\frac{\nu^8}{\epsilon^2}
\left(d^{4s+6}+d^6\log^{4s}\left(2+\frac{C}{\mu\epsilon}\right)\right)\log\frac1\epsilon\right),
\end{equation}
which is polynomial in $d$ when $\nu$ and $\log(C/\mu)$ grow at most polynomially.
\end{cor}

\begin{proof}
Each coordinate of $\Gamma_{M}|_\Omega$ contains at least $2N+1$ points, so a degree-$N$ trigonometric polynomial with zero discrete PDE residual has a constant-coefficient Laplacian vanishing on this tensor grid, and the one-dimensional zero count applied successively shows that the Laplacian vanishes identically. Its Fourier symbol $-\omega^T\Sigma\omega$ then forces the polynomial to be constant, and a zero boundary row forces that constant to vanish. Thus $\mathbb L_N$ has full column rank and \eqref{eq:dirichlet-kernel-condition} holds automatically. Source inconsistency is covered by the residual estimate for the regularized inverse \eqref{eq:dirichlet-filtered-inverse}.

The coefficient-pair construction above gives normalization $\mathcal O(\norm{\Sigma}_{1,1}N^2+1)$ and per-query gate cost $\tilde{\mathcal O}(d^2)$. The cutoff is chosen by \eqref{eq:dirichlet-cutoff} with $|\mathcal J|\ceil{g}$ replaced by $\norm{\Sigma}_{1,1}$. The discrete mass factor \eqref{eq:collar-mass-bound} is an absolute constant.

It remains to bound the stability constant $\Theta_N$ of \eqref{eq:continuous-stability-constant}. Writing $\ell$ for the $\Sigma$-harmonic extension of $v|_{\partial\Omega}$, the difference $v-\ell$ vanishes on $\partial\Omega$ and has the same image under the operator, so the smallest Dirichlet eigenvalue $d\pi^2\lambda_{\min}(\Sigma)/4$ gives
\begin{equation}
\label{eq:poisson-dirichlet-apriori}
\norm{v}_{L^2(\Omega)}\leq\frac{4}{d\pi^2\lambda_{\min}(\Sigma)}\norm{\nabla\cdot(\Sigma\nabla v)}_{L^2(\Omega)}+K_\Sigma\norm{v}_{L^2(\partial\Omega)},
\end{equation}
where $K_\Sigma$ is the norm of the extension as a map $L^2(\partial\Omega)\to L^2(\Omega)$. We claim that $K_\Sigma\leq1$, with no dependence on the dimension. Writing $P_\Sigma$ for the extension operator, the maximum principle gives $\norm{P_\Sigma}_{L^\infty\to L^\infty}\leq1$ and positivity of the kernel gives $\norm{P_\Sigma}_{L^1\to L^1}\leq\sup_{y\in\partial\Omega}\int_\Omega P_\Sigma(x,y)\,dx$, so the Riesz--Thorin theorem bounds $K_\Sigma^2$ by that supremum. Letting $w$ solve $-\nabla\cdot(\Sigma\nabla w)=1$ with $w|_{\partial\Omega}=0$, symmetry of the Green function identifies the inner integral as the conormal derivative $-n\cdot\Sigma\nabla w(y)$. The barrier $\beta_j(x)=(1-x_j^2)/(2\sigma_{jj})$ solves the same equation and is nonnegative on $\partial\Omega$, so $w\leq\beta_j$ on $\Omega$; both vanish on the faces normal to $e_j$, where comparing conormal derivatives gives $-n\cdot\Sigma\nabla w\leq1$. Taking the supremum over all faces proves the claim, so $\Theta_N\leq\Theta_\Sigma$ uniformly in $N$.

For constant coefficients the residual space lies in $\mathcal V_N$, so \cref{rmk:mz-constants} and \eqref{eq:restricted-stability-bound} give $\sigma_\Omega(\mathbb L_N)\geq c/\Theta_\Sigma$ for an absolute constant $c>0$. Applying \cref{thm:dirichlet-complexity}, including its amplification cost, and substituting the cutoff proves the result.
\end{proof}

\section{Discussion}
\label{sec:discussion}

Our periodic complexity bounds in \cref{sec:periodic-boundary-conditions} retain $\sigma_{\min}^+$, the smallest nonzero singular value of the discretized operator $\mathbb L_N$. On a periodic domain, an elliptic operator $\mathcal L = \sum_{\alpha\in \mathcal J}g_{\alpha}(x)D^{\alpha}$ of order $\ceil{\mathcal J}$, meaning one for which its leading symbol
\begin{equation}
p(x,\xi) = \sum_{\substack{\alpha\in\mathcal J\\ |\alpha|=\ceil{\mathcal J}}} g_\alpha(x)\,(i\xi)^\alpha
\end{equation}
never vanishes for real $\xi\neq0$ \cite{evans2010pde}, has closed range and finite-dimensional kernel for smooth coefficients, giving a well-defined $\sigma_{\min}^+(\mathcal L)>0$ on the orthogonal complement of its kernel. Ellipticity rules out the leading-order behavior degenerating at high frequencies. However, we do not know general conditions that guarantee $\sigma_{\min}^+(\mathbb L_N)$ to converge to (or stay bounded near) $\sigma_{\min}^+(\mathcal L)$ as $N\to\infty$ at a useful rate. For the constant-coefficient isotropic and anisotropic Poisson equations with periodic boundary conditions we provide an explicit $N$-independent bound (\cref{cor:poisson}).

The Dirichlet matrix of \cref{sec:dirichlet} behaves differently. Its rows constrain a trial vector only on $\Omega$ and on $\partial\Omega$, so exterior-concentrated trigonometric polynomials give exponentially small singular values even for Poisson. The same effect is known classically for Fourier extension \cite{adcock2014numerical}. We use the restricted constant $\sigma_\Omega(\mathbb L_N)$ of \eqref{eq:dirichlet-restricted-gap}, measured against the weighted output map $R_N$. \Cref{prop:restricted-stability} transfers a continuous a priori estimate to this discrete norm. Quadratic oversampling controls the interpolation constants for constant-coefficient operators; variable coefficients require the corresponding estimate on the residual space separately. The Poisson estimate is explicit in \cref{cor:poisson-dirichlet}. The regularized inverse controls the bias from near-null directions, while amplification accounts for the probability of preparing the restricted state. Together these steps give the $1/\epsilon^2$ dependence of the stated Dirichlet query bound, apart from logarithmic factors; this is not claimed to be optimal. The collar estimate additionally uses the mixed face-jet parameter $\nu$, whose dependence on dimension remains explicit.

Moreover, the discretization of the source and the consistency conditions of the PDE required careful treatment in our work. A nontrivial kernel requires determining which solution the algorithm returns. The range-projected source gives the same pseudoinverse solution $\mathbb L_N^+\eta_N$ as the original source. The proof of \cref{thm:linear-pde-complexity} uses the small Fourier residual to compare their preparation oracles, removing exact discrete consistency as a hypothesis without introducing a general projector oracle. The target's kernel component must still vanish or be included in the error bound. The periodic Poisson equation has a known constant-mode kernel, and \cref{cor:poisson} bounds the discretized solution's component in it. The Dirichlet Poisson matrix has full column rank (\cref{cor:poisson-dirichlet}), so the target-kernel condition is automatic. The residual analysis of the regularized inverse covers source inconsistency without an exact range projection.

The comparison of the performance of our algorithm to Monte Carlo methods also deserves a discussion. For real second-order elliptic operators admitting a Feynman--Kac representation \cite{kakutani1944two, muller1956some, sawhney2020monte, martin2022solving}, the PDE can be treated using a diffusion process. An unbiased Monte Carlo estimator with variance $\mathcal V$ requires $\mathcal O(\mathcal V/\epsilon^2)$ independent samples for root-mean-square error $\epsilon$; however, both variance and trajectory cost can depend on dimension. When estimator variance and trajectory cost grow at most polynomially with $d$, walk-on-spheres-type Monte Carlo methods estimate observables in $\mathcal O(\poly(d)/\epsilon^2)$ time. Our polylogarithmic precision dependence concerns periodic solution-state preparation for a given operator; the Dirichlet bound above has polynomial precision dependence. Estimating a downstream bounded scalar expectation by amplitude estimation requires further $\mathcal O(1/\epsilon)$ coherent calls, so an end-to-end comparison must include preparation, observable normalization and solution-norm costs.

We conclude by outlining several further applications of this framework. In \cite{olivucci2026provable} some of the authors of this paper use the machinery developed here to obtain a provable quantum--classical separation in continuous-domain Gibbs sampling. Another application we have been exploring is torsion-angle dynamics in molecular conformation search \cite{robert2021resource}, where the inherently periodic dihedral-angle coordinates fit naturally into our periodic-domain setting. Periodic boundary conditions are also natural in lattice gauge theory, where the Gauss-law constraint is commonly imposed on a torus to avoid boundary artifacts \cite{sharma2024onedimensional}. The classical XY Hamiltonian, and its generalization to the $n$-vector model, is governed in its low-temperature spin-wave limit by the same class of linear elliptic PDEs \cite{kosterlitz1973ordering}.

Other natural candidates for our framework include the linearized Poisson--Boltzmann equation of electrostatics \cite{kweyu2017fast}, and the (linear) Stokes equations of viscous fluid flow \cite{temam2001navier}, though the latter would require extending our scalar framework to coupled vector-valued systems with a divergence constraint. Another natural extension is to time-dependent PDEs. One route is the method of lines \cite{schiesser1991numerical}: discretizing only the spatial variables with our Fourier techniques while leaving time continuous reduces the problem to a linear system of ordinary differential equations, which would also remove the need to assume periodicity in time. Relatedly, a recent independent work develops a periodic Gevrey-extension approach for simulating slow analytic time-dependent Hamiltonians \cite{zhao2026quantum}, suggesting our Fourier-Gevrey framework may extend naturally to the time-dependent setting.

\begingroup
\renewcommand{\addcontentsline}[3]{}
\section*{Acknowledgement}
\label{sec:ack}

Authors thank Nathan~Baker, Anders~Blom, William~Kim, Jiaqi~Leng, Arsalan~Motamedi, Enrico~Olivucci, Soren~Smidstrup, and Matthias~Troyer for useful technical conversations. This work is supported by NSERC Discovery grant RGPIN-2022-03339, and the Quantum Computing Challenge Program AQC-206 at the National Research Council of Canada (NRC). P.~R.~further acknowledges the financial support of Mike and Ophelia Lazaridis, Innovation, Science and Economic Development Canada (ISED), 1QB Information Technologies (1QBit), the Perimeter Institute for Theoretical Physics, and the Province of Ontario through the Ministry of Colleges and Universities.

\endgroup

\bibliography{main}

\pagebreak
\appendix
\crefalias{section}{appendix}
\crefalias{subsection}{appendix}
\crefalias{subsubsection}{appendix}
\counterwithin{thm}{section}
\counterwithin{lem}{section}
\counterwithin{cor}{section}
\counterwithin{prop}{section}
\counterwithin{defn}{section}
\counterwithin{rmk}{section}
\counterwithin{exm}{section}

\section{The Fourier method}
\label{sec:fourier}

\subsection{Fourier transform}
\label{sub:fourier-transform}

We assume a $2\pi$ period for all multivariate functions in this work with respect to all their variables for a clear exposition and without a loss of generality. So, the domain of definition of a $d$-variate function is the flat torus $\mathbb \mathbb T^d= \mathbb R^d / 2 \pi \mathbb Z^d$. The frequency domain is then $\mathbb Z^d$. For all positive integers $N$, we consider two dual grids. Firstly, the frequency grid
\begin{equation}
\label{dft-lattice-main}
\Lambda_N= \{\omega\in \mathbb Z^d: \norm{\omega}_\infty \leq N \}
\end{equation}
with $(2N + 1)^d$ points. And secondly, a discrete spatial grid with the same number of points:
\begin{equation}
\Gamma_N= \left\{\frac{2\pi n}{2N+1}: n \in \mathbb Z, |n|\leq N \right\}^d.
\end{equation}
The discretization of $f$ on $\Gamma_N$ results in a state vector
\begin{equation}
\label{eq:disc-func-state}
\ket{f_{N}} = \sum_{x\in \Gamma_N} f(x) \ket{x} \in \mathcal{H}_N
\end{equation}
where $\mathcal H_N= \mathbb C^{\Gamma_N}$ is the Hilbert space of dimension $(2N+1)^d$ indexed by the grid points of $\Gamma_N$.

The \emph{Fourier transform} of $f$, denoted by $f \mapsto \hat f$, provides the \emph{Fourier expansion} of it
\begin{equation}
\label{eq:fourier-expansion}
f(x)= \sum_{\omega\in\mathbb{Z}^d} \hat f_\omega \,
e^{i\langle \omega, x\rangle}
\end{equation}
as the sum of the \emph{Fourier coefficients}
\begin{equation}
\label{eq:fourier-coefficients}
\hat f_\omega = \frac{1}{(2\pi)^d}\int_{\mathbb T^d}
f(x) e^{-i\langle \omega,x\rangle} \, dx= \mathbb E_{\mathbb T^d} [f(x) e^{-i \langle \omega, x \rangle}].
\end{equation}
In contrast, the \emph{discrete} Fourier transform coefficients are
\begin{equation}
\label{eq:discrete-fourier-coefficients}
\tilde f_\omega= \mathbb E_{\Gamma_N} [f(x) e^{-i \langle \omega, x\rangle}], \qquad \text{for all} \quad \omega \in \Lambda_N.
\end{equation}
They can be obtained from applying a $d$-fold QFT to $\ket{f_N}$:
\begin{equation}
F_N^{\otimes d} \ket{f_N}
=| \tilde f_N \rangle:= (2N+1)^{d/2} \sum_{\omega\in \Lambda_N} \tilde f_\omega\,\ket{\omega}.
\end{equation}

\subsection{Smoothness and the decay of Fourier coefficients}
\label{sub:smoothness-vs-decay}

It is easy to see by using \eqref{eq:fourier-expansion} and \eqref{eq:discrete-fourier-coefficients} that the continuous and discrete Fourier coefficients satisfy the following relation, known as the aliasing sum:
\begin{equation}
\label{eq:aliasing-sum}
\tilde f_\omega= \sum_{m \in \mathbb Z^d} \hat f_{\omega + (2N+1)m}.
\end{equation}
To represent the function $f$ accurately in the Fourier domain using a quantum state we seek for an approximation $|\hat f_N\rangle \simeq |\tilde f_N\rangle$. This requires a fast decay on the higher frequency Fourier coefficients, so that $\hat f_\omega$ dominates the right hand sum in \eqref{eq:aliasing-sum}. In this section we will relate the degree to which $f$ is ``smooth'' to the rate of decay of Fourier coefficients.

First, we introduce some preliminaries. For a $d$-tuple of non-negative integers $\alpha = (\alpha_1, \cdots, \alpha_d) \in\mathbb{Z}_{\geq 0}^d$ we define $\alpha! := \alpha_1!\cdots \alpha_d!$, and $|\alpha| := \alpha_1 + \cdots + \alpha_d$, and use the following notation for higher order derivatives:
\begin{align}
D^\alpha := \frac{\partial^{|\alpha|}}{\partial x_1^{\alpha_1}
\cdots \partial x_d^{\alpha_d}}.
\end{align}
Recall that imposing the condition $\left|D^{\alpha}f (x_0)\right| \leq \alpha! / r^{|\alpha|}$ on a multivariate function $f$ guarantees the convergence of the Taylor expansion in the hypercube $\prod_{i=1}^d (x_{0,i}-r,x_{0,i}+r)$, where $x_{0,i}$ denotes the $i$-th component of $x_0$.

Note that the Fourier transform of $D^\alpha f$ has Fourier coefficients $(i\omega_1)^{\alpha_1} \cdots (i\omega_d)^{\alpha_d} \hat f_\omega$ which we shorten to the notation $(i\omega)^\alpha \hat f_\omega$. To study the decay rate of $\hat f_\omega$ we repeatedly take advantage of the relation
\begin{equation}
\label{eq:fourier-diff}
\hat{(D^\alpha f)}_\omega = (i\omega)^\alpha \,\hat f_\omega.
\end{equation}

\begin{prop}
\label{prop:all-decay-rates}
If $f$ is merely continuous $\hat f_\omega \to 0$ as $\norm{\omega}_\infty \to \infty$. However, the decay rates in the second column of \cref{tab:fourier-tailedness} hold for differentiable functions. For the Gevrey classes with $s > 0$, the displayed coefficient bounds hold with the derivative-bound parameters $C$ and $r$ whenever $(r\norm{\omega}_\infty)^{1/s} \geq 16$.
\end{prop}

\begin{proof}

\noindent\underline{Case 1}. Continuous functions $f \in \mathscr{C}^0$

Continuity does not establish any decay rate, however it is sufficient for showing that $\hat f_\omega\to 0$ as $\norm{\omega}_\infty\to\infty$. This is also known as the Riemann--Lebesgue lemma \cite{bochner1949fourier}. One path to a proof is to observe that trigonometric polynomials are dense in $L^1(\mathbb T^d)$, so for every $\delta>0$ there exists a trigonometric polynomial $P$ such that $\norm{f-P}_1<\delta$.
Then
\begin{equation}
|\hat f_\omega-\hat P_\omega|
\leq (2\pi)^{-d} \norm{f-P}_1 < C\delta.
\end{equation}
But $\hat P_\omega$ vanishes for large $|\omega|$ because $P$ has finitely many modes.
Therefore
\begin{equation}
|\hat f_\omega| \leq C\delta \qquad \text{for all large} \quad |\omega|.
\end{equation}
Since $\delta$ is arbitrary, $\hat f_\omega\to 0$.

\noindent\underline{Case 2}. Functions with $p$ continuous derivatives $f \in \mathscr{C}^p$

For any index $\alpha$ such that $|\alpha| \leq p$, since $D^\alpha f$ is continuous, hence integrable, we have
\begin{equation}
|\omega^\alpha|\,|\hat f_\omega|
= |\hat{(D^\alpha f)}_\omega|
\leq \norm{D^\alpha f}_1.
\label{eq:apply-fourier-formula}
\end{equation}
Summing over all $|\alpha|\leq p$ results in
\begin{equation}
(1 + \norm{\omega}_\infty^p)|\hat f_\omega|
\leq \sum_{|\alpha|\le p} \|D^\alpha f\|_{L^1},
\end{equation}
where we have used the bound
\begin{equation}
\sum_{|\alpha|\le p} |\omega^\alpha|
\geq 1 + \sum_{j=1}^d |\omega_j|^p
\geq 1 + \max_{1\le j\le d} |\omega_j|^p
= 1 + \norm{\omega}_\infty^p.
\label{eq:basic-lower}
\end{equation}
We can now conclude that
\begin{equation}
|\hat f_\omega|
\leq C (1 + \norm{\omega}_\infty^p)^{-1}
= O(\norm{\omega}_\infty^{-p}).
\end{equation}

\noindent\underline{Case 3.} The H\"older class $f \in \mathscr{C}^{p, q}$ for $p$-smooth functions such that
\begin{equation}
|D^\alpha f(x) - D^\alpha f(y)| \le L |x-y|^q,
\qquad \text{for all} \quad x, y \in \mathbb T^d, \quad \text{and} \quad \alpha: |\alpha|=p.
\end{equation}
We first consider the single-variate case, with $g \in \mathscr{C}^{0, q}(\mathbb T^1)$. Then
\begin{equation}
|\hat g_\omega| \leq C|\omega|^{-q},
\end{equation}
because for $h=\pi/\omega$ we have
\begin{equation}
\hat g_\omega
= \tfrac{1}{4\pi}\int_0^{2\pi} (g(x)-g(x-h)) e^{-i\omega x} dx.
\end{equation}
Then we use $|g(x) - g(x - h)| \leq L|h|^q$. This also implies that for the $d$-variate case $g \in \mathscr{C}^{0, q}$
\begin{equation}
|\hat g_\omega| \leq C\norm{\omega}_\infty^{-q}.
\end{equation}

Now for $f \in \mathscr{C}^{p, q}$ we have $g = D^\alpha f \in \mathscr{C}^{0, q}$. Therefore, $|\hat g_\omega| \leq C\norm{\omega}_\infty^{-q}$. But $\hat g_\omega = (i\omega)^\alpha\hat f_\omega$, so
\begin{equation}
|\omega^\alpha||\hat f_\omega| \leq C\norm{\omega}^{-q}.
\end{equation}
We can now sum over all $\alpha$ with $|\alpha|\leq p$ and repeat the argument of the previous case to conclude that $|\hat f_\omega|=O(\norm{\omega}_\infty^{-p-q})$.

\noindent\underline{Case 4.} Smooth functions $f \in \mathscr{C}^\infty$

This case is obvious since $f \in \mathscr{C}^p$ for all choices of $p$.

\noindent\underline{Case 5}. The Gevrey class ($s > 1$) wherein $f \in \mathscr{C}^\infty$ but additionally
\begin{equation}
\|D^\alpha f\|_\infty \leq C (\alpha!)^s/r^{|\alpha|}, \qquad \text{for all} \quad \alpha \in \mathbb Z_{\geq 0}^d.
\end{equation}
For fixed $\omega \neq 0$ and $m \in \mathbb{N}$ to be chosen later, we pick a multi-index $\alpha$ with $|\alpha| = m$ concentrated in the coordinate where $|\omega_j|$ is maximal, so that $|\omega^\alpha| = \norm{\omega}_\infty^m$. Then using $\hat{D^\alpha f}_\omega = (i\omega)^\alpha \hat f_\omega$ and $|\hat{D^\alpha f}_\omega| \leq \norm{D^\alpha f}_\infty$, we have
\begin{equation}
|\hat f_\omega|
\leq C \left(r \norm{\omega}_\infty\right)^{-m} (m!)^s.
\end{equation}
Writing $t = (r\norm{\omega}_\infty)^{1/s} \geq 16$ and choosing $m = \lfloor t\rfloor$, the Stirling bound $\log(m!) \leq m\log m - m + 1 + \frac12\log m$ gives
\begin{equation}
\label{eq:gevrey-coefficient-quantitative}
\log\frac{|\hat f_\omega|}{C}
\leq s\left(-t + 2 + \tfrac12\log t\right)
\leq -\frac34st \leq -\frac34t.
\end{equation}
Thus
\begin{equation}
|\hat f_\omega| \leq C e^{-\frac34(r\norm{\omega}_\infty)^{1/s}}
\end{equation}
with the original $C$ and $r$, proving the tabulated bound.

\noindent\underline{Case 6.} Analytic functions $f \in \mathscr{C}^\omega$

At $s = 1$, the same calculation applies with $t = r\norm{\omega}_\infty \geq 16$ and gives
\begin{equation}
\log\frac{|\hat f_\omega|}{C} \leq -t + 2 + \tfrac12\log t \leq -\frac34t.
\label{eq:analytic-coefficient-quantitative}
\end{equation}
Thus $|\hat f_\omega| \leq Ce^{-3r\norm{\omega}_\infty/4}$ with the original derivative-bound parameters $C$ and $r$, as claimed.

\noindent\underline{Case 7.} Entire functions $f \in \mathscr{E}$

Entire functions are those that are holomorphic on all of $\mathbb C^d$. Therefore, the quantity $M(\sigma)$ defined in \cref{rmk:sigma-to-r-subtlety} is finite for every $\sigma > 0$. The contour argument there gives
\begin{equation}
|\hat f_\omega| \leq M(\sigma)\, e^{-\sigma\norm{\omega}_\infty}
\quad \text{for every} \quad \sigma>0.
\end{equation}

\noindent\underline{Case 8.} Entire functions of finite order, $\mathscr{G}^s$ for $0 < s < 1$

For $0 < s < 1$, the same Stirling bound gives $\log(|\hat f_\omega|/C) \leq s(-t + 2 + \frac12\log t) \leq -st/2$ for $t = (r\norm{\omega}_\infty)^{1/s} \geq 16$. Hence $|\hat f_\omega| \leq C e^{-\frac{s}{2}(r\norm{\omega}_\infty)^{1/s}}$ with unchanged parameters. This condition also forces $f$ to be entire, because, writing $a_\alpha = D^\alpha f(x_0)/\alpha!$ at any real point $x_0$, we have $|a_\alpha| \leq C(\alpha!)^{s - 1}/r^{|\alpha|}$, and since $s < 1$ we have $(\alpha!)^{s - 1} \to 0$ super-exponentially in $|\alpha|$, so $\limsup_{|\alpha| \to \infty}|a_\alpha|^{1/|\alpha|} = 0$. Therefore, the Taylor series at $x_0$ converges everywhere. As this holds at every real $x_0$ with the same constants $(C, r)$, $f$ extends to an entire function.

\noindent\underline{Case 9.} Finite bandwidth, $s=0$

At $s=0$ we have $\norm{D^\alpha f}_\infty \leq C r^{-|\alpha|}$, with no factorial growth. This holds if and only if $f$ is a trigonometric polynomial. The `if' direction is Bernstein's theorem. For the `only if' direction, the bound implies (by summing the Taylor series) that $f$ extends holomorphically with $\sup_x|f(x-iy)|\leq Ce^{|y|/r}$ for every real $y$ of either sign. For $\omega$ with $\omega_j>1/r$ for some coordinate $j$ (other components/signs are analogous), shifting the contour in that variable to height $-t$, $t>0$, gives $|\hat f_\omega| \leq Ce^{t/r - \omega_j t}$, which tends to $0$ as $t\to\infty$. Hence $\hat f_\omega=0$ whenever $\norm{\omega}_\infty>1/r$, so $f$ is a trigonometric polynomial of degree at most $\lceil 1/r\rceil$.
\end{proof}

\begin{rmk}
\label{rmk:gevrey-coefficient-onset}
For $s > 1$, the stronger coefficient bound $|\hat f_\omega| \leq Ce^{-(r\norm{\omega}_\infty)^{1/s}}$ also retains the original $C$ and $r$, with an $s$-dependent onset. The factorial estimate in Case~5 and Stirling's approximation give
\begin{equation}
\label{eq:gevrey-log-bound}
\log\frac{|\hat f_\omega|}{C}
\leq -m\log (r\norm{\omega}_\infty) + s(m\log m - m) + \mathcal O(s\log m).
\end{equation}
Choosing $m = \lfloor(r\norm{\omega}_\infty)^{1/s}\rfloor$ and absorbing the logarithmic remainder using $s > 1$ proves the stronger bound for sufficiently large modes. An explicit sufficient frequency threshold is
\begin{equation}
\tau_s := \frac1r\left[\dfrac{s}{s - 1}\left(3 + \log\dfrac{s}{s - 1}\right)\right]^s,
\qquad \text{for} \quad s > 1.
\label{eq:gevrey-coefficient-onset}
\end{equation}
Indeed, writing $\delta = (s - 1)/s$ and $t = (r\norm{\omega}_\infty)^{1/s}$, the explicit Stirling bound gives $\log(|\hat f_\omega|/C) \leq s(-t + 2 + \frac12\log t)$. Since $\log t \leq \delta t + \log(1/\delta) - 1$, we have $2 + \frac12\log t \leq \delta t$ for $t \geq (r\tau_s)^{1/s}$. Thus the stronger bound holds whenever $\norm{\omega}_\infty \geq \tau_s$. This onset diverges as $s\downarrow1$, whereas the table's $3/4$ rate requires only $(r\norm{\omega}_\infty)^{1/s} \geq 16$.
\end{rmk}

\begin{rmk}
\label{rmk:sigma-to-r-subtlety}
For analytic functions, the derivative bound gives a holomorphic extension on the strip
\begin{equation}
\{ z \in \mathbb{C}^d : |{\rm Im}\, z_j| < r \text{ for all } j \}.
\end{equation}
For a coordinate $j$ with $|\omega_j| = \norm{\omega}_\infty$, shifting the integration contour to $x - i\sigma\,\mathrm{sgn}(\omega_j)e_j$, for $0 < \sigma < r$, gives
\begin{equation}
\hat f_\omega
= \frac{e^{-\sigma\norm{\omega}_\infty}}{(2\pi)^d}
\int_{\mathbb T^d} f(x - i\sigma\,\mathrm{sgn}(\omega_j)e_j)\, e^{-i\omega \cdot x}\, dx.
\end{equation}
Writing $M(\sigma) := \max_{1 \leq j \leq d}\sup_{x \in \mathbb T^d, \, |t| \leq \sigma}|f(x + it e_j)|$, Taylor expansion in the shifted coordinate gives
\begin{equation}
M(\sigma) \leq C\sum_{m = 0}^\infty\left(\frac\sigma r\right)^m
= \frac{C}{1 - \sigma/r}, \qquad \text{for} \quad 0 \leq \sigma < r.
\label{eq:analytic-strip-bound}
\end{equation}
Consequently,
\begin{equation}
|\hat f_\omega| \leq M(\sigma)e^{-\sigma\norm{\omega}_\infty}
\leq \frac{C}{1 - \sigma/r}\, e^{-\sigma\norm{\omega}_\infty},
\qquad \text{for} \quad 0 \leq \sigma < r.
\end{equation}
This bound retains the original $C$ and $r$, and exhibits the growing prefactor as $\sigma$ approaches $r$. If $f$ has a singularity at height $r$, its modulus on the shifted contour can diverge as $\sigma\uparrow r$, so this argument does not generally give $Ce^{-r\norm{\omega}_\infty}$ with the same $C$. The table uses the $3/4$ rate established in Case~6 to keep its prefactor equal to $C$.
\end{rmk}

We note that smooth functions are not necessarily Gevrey, and entire functions are not necessarily of finite order, as the following examples show.

\begin{exm}
\label{exm:smooth-not-gevrey}
A function with $|\hat f_\omega| \sim e^{-(\log\|\omega\|_\infty)^2}$ (for $\|\omega\|_\infty \geq 2$) decays faster than every polynomial rate $\|\omega\|_\infty^{-k}$, hence is smooth ($\mathscr{C}^\infty$), yet decays slower than $e^{-(r\|\omega\|_\infty)^{1/s}}$ for every finite $s$ (however large), since $(\log\|\omega\|_\infty)^2$ is eventually dominated by $\|\omega\|_\infty^{1/s}$ for any fixed power $1/s > 0$. So it lies in no Gevrey class.
\end{exm}

\begin{exm}
\label{exm:entire-not-finite-order}
The function $f(x) =\mathrm{Re}(e^{e^{ix}}) = \sum_{n \geq 0} \cos(nx)/n!$ is entire (a composition of entire functions), yet $|\hat f_n| = 1/n!$, so $\log(1/|\hat f_n|) \sim n\log n$ by Stirling's approximation. This decay is faster than every fixed rate $rn$, but slower than $n^{1/s}$ for every $s < 1$. So $\mathscr{E}$ is strictly larger than the class of finite-order entire functions of Case 8.
\end{exm}

We also provide several remarks about the proof of \cref{prop:all-decay-rates} and the Gevrey hierarchy.

\begin{rmk}
\label{rmk:entire-periodicity}
Any holomorphic extension of a periodic function to all of $\mathbb C^d$ (Case 7) is automatically $2\pi$-periodic in every imaginary slice, not just on the real axis. This follows from the identity theorem of holomorphic functions since $f(z)$ and $f(z+2\pi e_j)$ agree on $\mathbb R^d$.
\end{rmk}

\begin{rmk}
\label{rmk:entire-no-fixed-rate}
The bound of Case 7 does not by itself pin down a single exponential rate: the Fourier coefficients of an entire function decay faster than any fixed exponential rate, at the cost of a prefactor $M(\sigma)$ that may grow without bound as $\sigma\to\infty$, as \cref{exm:entire-not-finite-order} illustrates.
\end{rmk}

\begin{rmk}
\label{rmk:classical-analytic-precedent}
The analytic contour estimate in \cref{rmk:sigma-to-r-subtlety}, when specialized to $d = 1$, recovers the classical result of \cite{tadmor1986exponential}, who shows that Fourier differencing of a periodic analytic function converges at a rate governed by the width of its strip of analyticity, and likewise that Chebyshev differencing of a (possibly non-periodic) analytic function on an interval converges at a rate governed by its Bernstein regularity ellipse. The same underlying fact, that analyticity yields exponential spectral accuracy, is also the basis for resolving the Gibbs phenomenon for piecewise-analytic functions \cite{gottlieb1997gibbs}.
\end{rmk}

We can now consider the tailedness of the Fourier series. Let $\hat f_N = (\hat f_\omega)_{\omega \in \Lambda_N} \in \mathbb C^{(2N + 1)^d}$ and consider the inclusion $\iota: \mathbb C^{(2N + 1)^d} \hookrightarrow \ell^2(\mathbb Z^d)$ induced by $\Lambda_N \hookrightarrow \mathbb Z^d$. We will reuse the notation $\hat f_N$ for $\iota \hat f_N$. By the \emph{tail} of the Fourier series we mean the truncation error
\begin{equation}
\label{eq:def-tail}
\mathcal T_N = \norm{\hat f - \hat f_N}_2 = \left(\sum_{\omega \not \in \Lambda_N} |\hat f_\omega|^2\right)^{1/2},
\end{equation}
which is identical to
\begin{equation}
\label{eq:tail-parseval}
\mathcal T_N = \left(\frac{1}{(2\pi)^d}\int_{\mathbb T^d} \left|\sum_{\omega \not \in \Lambda_N} \hat f_\omega e^{i \langle \omega, x\rangle}\right|^2 dx\right)^{1/2}
\end{equation}
by Parseval's theorem.

\begin{prop}
\label{prop:tail-decay-rate}
The tail of the Fourier series satisfies $\mathcal T_N \leq \bar{\mathcal T}_N$ according to the third column of \cref{tab:fourier-tailedness}. The fourth column gives sufficient cutoffs for an $\epsilon$-accurate truncated Fourier expansion for sufficiently small $\epsilon$.
\end{prop}

\begin{proof}
For a nonnegative decreasing coefficient envelope $g$, the bound $|\hat f_\omega| \leq Cg(\norm{\omega}_\infty)$ on the omitted modes gives, for $N \geq 1$,
\begin{align}
\mathcal T_N^2
& \leq C^2\sum_{k = N + 1}^\infty \bigl((2k + 1)^d - (2k - 1)^d\bigr)g(k)^2 \nonumber\\
& \leq 2dC^2\int_N^\infty (2x + 1)^{d - 1}g(x)^2\, dx \nonumber\\
& \leq 2d\left(2 + \frac1N\right)^{d - 1}C^2\int_N^\infty x^{d - 1}g(x)^2\, dx \nonumber\\
& \leq 4^d C^2\int_N^\infty x^{d - 1}g(x)^2\, dx.
\label{eq:tail-shell-bound}
\end{align}
The second inequality follows by integrating $2d(2x + 1)^{d - 1}$ over each interval $[k - 1, k]$ and using $g(k) \leq g(x)$. For $x \geq N$, we have $2x + 1 \leq (2 + 1/N)x$, giving the tighter factor $2d(2 + 1/N)^{d - 1}$. Since this is at most $2d\, 3^{d - 1} \leq 4^d$ for $d, N \geq 1$, we use $4^d$ for simplicity in the table and the remaining estimates.

For the H\"older row, $g(x) = x^{-p - q}$ with $p + q > d/2$, so
\begin{equation}
\mathcal T_N^2 \leq 4^d C^2\int_N^\infty x^{d - 1 - 2p - 2q}\, dx
= \frac{4^d C^2}{2p + 2q - d}N^{d - 2p - 2q} = \bar{\mathcal T}_N^2.
\end{equation}
Solving this bound for $\mathcal T_N \leq \epsilon$ gives the corresponding cutoff. The $p$-smooth case follows by setting $q = 0$, provided $p > d/2$. These are coefficient-envelope estimates, not optimal bounds using derivative norms.

For the Gevrey rows, we write $b = 3/4$ when $s \geq 1$ and $b = s/2$ when $0 < s < 1$. By \cref{prop:all-decay-rates}, the coefficient envelope $g(x) = e^{-b(rx)^{1/s}}$ holds on the omitted modes whenever $(rN)^{1/s} \geq 16$. Then \eqref{eq:tail-shell-bound} gives
\begin{equation}
\mathcal T_N^2 \leq 4^d C^2 \int_N^\infty x^{d - 1}e^{-2b(rx)^{1/s}}\, dx.
\end{equation}
With $t = x^{1/s}$, the integral becomes
\begin{equation}
\int_N^\infty x^{\, d - 1}\, e^{-2b(rx)^{1/s}}\, dx
= s\int_{N^{1/s}}^\infty t^{sd - 1}e^{-2br^{1/s}t}\, dt.
\end{equation}
For $I_\beta = \int_a^\infty t^\beta e^{-2br^{1/s}t}\, dt$ and $br^{1/s}a \geq \max\{\beta, 0\}$, we have
\begin{equation}
I_\beta \leq a^\beta e^{-2br^{1/s}a}
\int_0^\infty e^{-br^{1/s}v}\, dv
= \frac{a^\beta}{br^{1/s}}e^{-2br^{1/s}a}.
\end{equation}
Indeed, $(a + v)^\beta \leq a^\beta e^{\max\{\beta, 0\}v/a} \leq a^\beta e^{br^{1/s}v}$. Taking $a = N^{1/s}$ and $\beta = sd - 1$ therefore yields
\begin{equation}
\label{eq:tail-exponential-bound}
\mathcal T_N \leq \bar{\mathcal T}_N = 2^d C\sqrt{\frac{s}{br^{1/s}}N^{d - 1/s}}\, e^{-b(rN)^{1/s}},
\quad \text{provided that} \quad b(rN)^{1/s} \geq sd - 1.
\end{equation}

The choices of $b$ give the three quantitative exponential rows, with the original derivative-bound parameters $C$ and $r$. For entire functions, the contour bound of Case~7 and the same integral estimate at any exponential rate greater than $r$ give $\mathcal T_N = o(e^{-rN})$ for every $r > 0$. Finally, for finite bandwidth functions, $\mathcal T_N = 0$ exactly once $N \geq 1/r$.

For fixed $d, r, s$, the bound $N^{d/2 - 1/(2s)} \leq e^{\frac b2(rN)^{1/s}}$ holds for sufficiently large $N$. Together with \eqref{eq:tail-exponential-bound}, this gives
\begin{equation}
\mathcal T_N \leq \bar{\mathcal T}_N \leq 2^d C\sqrt{\frac{s}{br^{1/s}}}\, e^{-\frac b2(rN)^{1/s}}.
\end{equation}
Solving this inequality gives the table's logarithmic cutoffs. An explicit sufficient regime is
\begin{equation}
\label{eq:tail-cutoff-epsilon-regime}
0 < \epsilon \leq 2^d C\sqrt{\frac{s}{br^{1/s}}}\, e^{-\frac b2(16 + r^{1/s})}
\left(1 + \frac{2ds}{br^{1/s}}\right)^{-ds}.
\end{equation}
Indeed, for $y = N^{1/s}$ the table cutoff satisfies
\begin{equation}
br^{1/s}y \geq 16b + br^{1/s} + 2ds\log(1 + 2ds/(br^{1/s})) \geq 16b + 2ds.
\end{equation}
The last inequality follows from $\log(1 + u) \geq u/(1 + u)$. Thus $N \geq 1$, the coefficient estimate is valid, and the integral condition holds. Also, $\log y \leq br^{1/s}y/(2ds) + \log(1 + 2ds/(br^{1/s}))$ implies $(ds - 1)\log y \leq br^{1/s}y$, establishing the polynomial-prefactor bound. This proves the table's cutoff guarantees for all three exponential rows.
\end{proof}

We can now study the convergence rate of $\tilde f_N \to \hat f$ as previously promised. First we prove a useful lemma.

\begin{lem}
\label{lem:useful}
Let $k \geq d$. We have
\begin{equation}
\max_{x \in [-1, 1]^d} \left( \sum_{m \in \mathbb Z^d \setminus \{0\}} \norm{x + 2m}_\infty^{-2k} \right) \leq 2^{d + 1}.
\end{equation}
\end{lem}

\begin{proof}
For $x \in [-1, 1]^d$ and $m \neq 0$, we have $\norm{x + 2m}_\infty \geq 1$. For $t \geq 1$, the number of such lattice points with $\norm{x + 2m}_\infty \leq t$ is at most $(t + 1)^d \leq 2^d t^d$, since the lattice spacing is two in each coordinate. Integrating this counting bound gives
\begin{align}
\sum_{m \neq 0}\norm{x + 2m}_\infty^{-2k}
& = 2k\int_1^\infty t^{-2k - 1}\#\{m \neq 0:\norm{x + 2m}_\infty \leq t\}\, dt \nonumber\\
& \leq 2^{d + 1}k\int_1^\infty t^{d - 2k - 1}\, dt
= \frac{2^{d + 1}k}{2k - d} \leq 2^{d + 1}.
\end{align}
Since this bound is uniform in $x$, the claimed maximum is at most $2^{d + 1}$.
\end{proof}

\begin{prop}
\label{prop:dft-distance-bound}
Let $f$ be a continuous $2\pi$-periodic $d$-variate function from a regularity class in \cref{tab:fourier-tailedness}, with $N \geq 1$. For the Gevrey rows with parameters $C$ and $r$, assuming $(rN)^{1/s} \geq 16$, if either $s \geq 1$ with $\frac34(rN)^{1/s} \geq 3sd - 1$, or $0 < s < 1$ with $\frac{s}{2}(rN)^{1/s} \geq \max\{sd, 3sd - 1\}$, then
\begin{equation}
\label{eq:dft-envelope-bound}
\norm{\tilde f_N - \hat f}_2 \leq 2^{d/2 + 2}\bar{\mathcal T}_N.
\end{equation}
For the H\"older rows, the same bound holds when $p + q \geq 2d$.
\end{prop}

\begin{proof}
The aliasing sums \eqref{eq:aliasing-sum} imply that
\begin{align}
\norm{\tilde{f}_{N} -\hat{f}}_2^2
&= \underbrace{\sum_{\omega\not\in\Lambda_N}
\left| \hat{f}_\omega \right|^2}_{S_1}
+ \underbrace{\sum_{\omega \in \Lambda_N} \left| \sum_{m \in \mathbb{Z}^d \setminus \{0\}} \hat{f}_{\omega + (2N + 1)m} \right|^2}_{S_2}.
\end{align}
For the first term, \cref{prop:tail-decay-rate} gives $S_1 = \mathcal T_N^2 \leq \bar{\mathcal T}_N^2$ in its stated regimes.

For the second term, we choose the auxiliary exponent $k = d$. The Cauchy--Schwarz inequality gives
\begin{equation}
\begin{split}
S_2 & \leq \sum_{\omega \in \Lambda_N} \left(\sum_{m \in \mathbb{Z}^d \setminus \{0\}} \norm{\omega+ (2N + 1)m}_\infty^{-2k}\right) \left( \sum_{m \in \mathbb{Z}^d \setminus \{0\}}
\norm{\omega+ (2N + 1)m}_\infty^{2k} \left| \hat{f}_{\omega+ (2N + 1)m} \right|^2 \right) \\
& \leq \underbrace{\left(\max_{\omega\in \Lambda_N}
\sum_{m \in \mathbb{Z}^d \setminus \{0\}} \norm{\omega+ (2N + 1)m}_\infty^{-2k}\right)}_{S_3}
\underbrace{\left( \sum_{\omega \not \in \Lambda_N} \norm{\omega}_\infty^{2k} |\hat f_\omega|^2\right)}_{S_4}.
\end{split}
\end{equation}
With $x = \omega/(N + 1/2) \in [-1, 1]^d$, \cref{lem:useful} gives
\begin{equation}
\begin{split}
S_3 \leq (N + 1/2)^{-2k} \max_{x \in [-1, 1]^d} \left(\sum_{m \in \mathbb{Z}^d \setminus \{0\}} \norm{x + 2m}_\infty^{-2k}\right) \leq N^{-2k} 2^{d + 1}.
\label{eq:usefuleq}
\end{split}
\end{equation}
Since $S_1 = \mathcal T_N^2$, these estimates give
\begin{equation}
\label{eq:error-big}
\norm{\tilde f_N - \hat f}_2^2
\leq \mathcal T_N^2 + 2^{d + 1}N^{-2k}\sum_{\omega \notin \Lambda_N}\norm{\omega}_\infty^{2k}|\hat f_\omega|^2.
\end{equation}

By \cref{prop:all-decay-rates}, the coefficient bounds in \cref{tab:fourier-tailedness} apply on the omitted modes, with $(rN)^{1/s} \geq 16$ in the Gevrey cases. Applying \eqref{eq:tail-shell-bound} to $g(x) = x^{\ell - p - q}$ gives, for $\ell \geq 0$ and $p + q > \ell + d/2$,
\begin{equation}
\label{eq:weighted-tail-holder}
\sum_{\omega \notin \Lambda_N}\norm{\omega}_\infty^{2\ell}|\hat f_\omega|^2
\leq \frac{4^d C^2}{2p + 2q - 2\ell - d}N^{2\ell + d - 2p - 2q}.
\end{equation}
For $g(x) = x^\ell e^{-b(rx)^{1/s}}$, with $b = 3/4$ for $s \geq 1$ and $b = s/2$ for $0 < s < 1$, the same argument applies once $b(rN)^{1/s} \geq s\ell$, so that $g$ is decreasing on $[N, \infty)$. The integral estimate in the proof of \cref{prop:tail-decay-rate}, now with $\beta = s(2\ell + d) - 1$, yields
\begin{equation}
\label{eq:weighted-tail-gevrey}
\sum_{\omega \notin \Lambda_N}\norm{\omega}_\infty^{2\ell}|\hat f_\omega|^2
\leq \frac{4^d sC^2}{br^{1/s}}N^{2\ell + d - 1/s}e^{-2b(rN)^{1/s}},
\end{equation}
provided $b(rN)^{1/s} \geq \max\{s\ell, s(2\ell + d) - 1\}$.

Taking $\ell = k$ and combining with the bound on $S_1$ gives
\begin{equation}
\label{eq:dist-bound-holder}
\norm{\tilde f_N - \hat f}_2
\leq 2^d C\left(\frac1{2p + 2q - d} + \frac{2^{d + 1}}{2p + 2q - 2k - d}\right)^{1/2}N^{d/2 - p - q}
\end{equation}
when $p + q > k + d/2$. In the exponential case, $1 + 2^{d + 1} \leq 2^{d + 2}$ gives
\begin{equation}
\label{eq:dist-bound-gevrey}
\norm{\tilde f_N - \hat f}_2
\leq 2^{3d/2 + 1}C\sqrt{\frac{s}{br^{1/s}}N^{d - 1/s}}e^{-b(rN)^{1/s}},
\end{equation}
provided $b(rN)^{1/s} \geq \max\{sd, 3sd - 1\}$. For $s \geq 1$, the maximum is $3sd - 1$, giving the stated condition for $b = 3/4$. The exponential right-hand side is $2^{d/2 + 1}\bar{\mathcal T}_N$, which is below the stated bound. In the algebraic case, $p + q \geq 2d$ implies $(2p + 2q - d)/(2p + 2q - 3d) \leq 3$. Thus the squared prefactor relative to $\bar{\mathcal T}_N$ is at most $1 + 3 \cdot 2^{d + 1} \leq 2^{d + 4}$. The choice $b = s/2$ gives the stated condition for $0 < s < 1$. Finite bandwidth leaves no aliases beyond its cutoff. This proves the claimed bound in every stated regime.
\end{proof}

\begin{cor}
\label{cor:qft-func-err}
Let $N \in \tilde\Omega((ds)^s/r)$ and $f$ be an $s$-Gevrey and $2\pi$-periodic $d$-variate function with Gevrey parameters $C$ and $r$, and $s \geq 1$. We have
\begin{align}
\norm{\tilde f_N - \hat f}_2
\leq C e^{-\frac12(rN)^{1/s}}.
\end{align}
More explicitly, we may choose\,\footnote{Here and in subsequent cutoff formulas, a nonpositive logarithmic cutoff is omitted before taking its power, imposing no additional constraint on the positive integer $N$.}
\begin{equation}
\label{eq:N-bound-for-good-gevrey-fourier}
N \geq \frac1r\left[64ds\log\!\left(\frac{64ds}{r^{1/s}}\right)\right]^s.
\end{equation}
\end{cor}

\begin{proof}
We write $w = ds$, $L = \log(64w/r^{1/s})$, and $t = (rN)^{1/s}$, and use $b = 3/4$ as in \eqref{eq:tail-exponential-bound}. If $r^{1/s} \geq 32w$, then $N \geq 1$ gives $t \geq 32w$; otherwise, the cutoff gives $t \geq 64w\log 2 > 32w$. Thus \cref{prop:all-decay-rates} supplies the table's coefficient bound with the original Gevrey parameters $C$ and $r$, and $bt \geq 24w \geq 3sd - 1$ ensures that \cref{prop:dft-distance-bound} applies.

For $y = N^{1/s} \geq 1$, we have
\begin{equation}
\log y \leq \frac{r^{1/s}y}{16w} + \log\!\left(\frac{16w}{r^{1/s}}\right) - 1
= \frac{t}{16w} + L - \log 4 - 1.
\label{eq:lny-strong}
\end{equation}
Since $d, s \leq w$, $y \geq 1$, and $L \leq t/(64w)$, the logarithm of the resulting squared prefactor satisfies
\begin{align}
\log\!\left(2^{3d + 4}\frac{s}{br^{1/s}}y^{sd - 1}\right)
& \leq 7w + L + 3w\log y \nonumber\\
& \leq \frac{3t}{16} + (3w + 1)L \leq \frac{t}{4} \leq \frac{t}{2},
\label{eq:log-eq}
\end{align}
where we used $7 - 3(\log 4 + 1) < 0$.
Taking square roots, the prefactor contributes at most $e^{(rN)^{1/s}/4}$ against the decay $e^{-b(rN)^{1/s}}$. Since $b \geq 3/4$, this proves the claimed error bound.
\end{proof}

\begin{rmk}
For $0 < s < 1$, Case~8 of \cref{prop:all-decay-rates} and \cref{prop:dft-distance-bound} give $\norm{\tilde f_N - \hat f}_2 \leq Ce^{-\frac{s}{4}(rN)^{1/s}}$. We use $w = d$ in the logarithmic estimate above. In \eqref{eq:N-bound-for-good-gevrey-fourier}, we replace $ds$ by $d$ and substitute $r \mapsto (s/2)^s r$. The resulting cutoff gives $(s/2)(rN)^{1/s} \geq 32d$ and hence satisfies the coefficient onset. The same calculation bounds the squared prefactor by $e^{(s/2)(rN)^{1/s}}$, proving the stated error.
\end{rmk}

\begin{rmk}
\label{rmk:smooth-dft-no-go}
The H\"older case of \cref{prop:dft-distance-bound} gives algebraic convergence when $p + q \geq 2d$. Thus Gevrey smoothness is not necessary for convergence of the discrete Fourier transform, but these finite-smoothness bounds do not provide the exponential accuracy of \cref{cor:qft-func-err}.
\end{rmk}

\subsection{The Fourier approximation of derivatives}
\label{sec:fourier-diff-approx}

Recall that the Fourier transform of $D^\alpha f$ has Fourier coefficients $(i\omega)^\alpha \hat f_\omega$. This means that differentiation is a diagonal operator in the Fourier domain:

\begin{equation}
\label{eq:diff-from-fourier}
D^\alpha f = \mathcal F^{-1} \diag\left\{(i\omega)^\alpha: \omega \in \mathbb Z^d\right\} \mathcal F (f)
\end{equation}
where $\mathcal F$ denotes the Fourier transform as an operator. It is therefore natural to ask whether the truncation of $\mathcal F$ and the diagonal operator up to a cutoff mode $N$ approximates $D^\alpha$ well. Therefore, we define the approximate derivative operator
\begin{equation}
\label{eq:diff-from-discrete-fourier}
\tilde D^\alpha_N = F_N^{\dagger \otimes d} \diag\left\{(i\omega)^\alpha: \omega \in \Lambda_N\right\} F_N^{\otimes d}
\end{equation}
and seek to bound the distance between the approximate derivate applied to the discretization of $f$ and the discretization of the true derivative: $\norm{\tilde D^\alpha_N f_N - (D^\alpha f)_N}_2$.

\begin{prop}
\label{prop:dft-derivative-bound}
Let $f$ be a $2\pi$-periodic real-valued $d$-variate function from a regularity class in \cref{tab:fourier-tailedness}, with continuous derivative $D^\alpha f$ and $N \geq 1$. For the Gevrey rows with parameters $C$ and $r$, assuming $(rN)^{1/s} \geq 16$, if either $s \geq 1$ with $\frac34(rN)^{1/s} \geq s(3d + 2|\alpha|) - 1$, or $0 < s < 1$ with $\frac{s}{2}(rN)^{1/s} \geq \max\{s(d+|\alpha|), s(3d + 2|\alpha|) - 1\}$, then
\begin{equation}
\label{eq:derivative-envelope-bound}
\norm{\tilde D_N^\alpha f_N - (D^\alpha f)_N}_2
\leq 2^{(d + 5)/2}(2N + 1)^{d/2}N^{|\alpha|}\bar{\mathcal T}_N.
\end{equation}
The error also satisfies the grid-uniform bound
\begin{equation}
\label{eq:uniform-derivative-envelope}
\max_{x \in \Gamma_N}|\tilde D_N^\alpha f_N(x) - D^\alpha f(x)|
\leq \frac{2^{d + 2}}{\sqrt d}N^{|\alpha| + d/2}\bar{\mathcal T}_N.
\end{equation}
For the H\"older rows, both bounds hold when $p + q \geq 2d + \frac32|\alpha|$.
\end{prop}

\begin{proof}
For all $x \in \Gamma_N$ we have
\begin{align}
D^\alpha f(x)
&= \sum_{\omega \in \mathbb{Z}^d} (i\omega)^\alpha
\hat{f}_\omega\, e^{i \langle \omega ,x\rangle} \\
&= \sum_{\omega \in \Lambda_N} e^{i \langle \omega ,x\rangle}
\sum_{m\in\mathbb{Z}^d} (i(\omega + (2N+1)m))^\alpha
\hat{f}_{\omega  + (2N+1)m}.
\end{align}
And using the aliaising sums,
\begin{align}
\tilde D^\alpha_N f_N(x)
&= \sum_{\omega\in \Lambda_N} (i\omega)^\alpha
\tilde{f}_\omega\, e^{i\langle \omega,x\rangle} \\
&= \sum_{\omega\in \Lambda_N} e^{i\langle \omega,x\rangle} \,
\sum_{m\in\mathbb{Z}^d} (i\omega)^\alpha \hat{f}_{\omega+(2N+1)m}.
\label{eq:wrap-ap}
\end{align}
Hence
\begin{align}
D^\alpha f(x) - \tilde D^\alpha_N f_N(x)
= \sum_{\omega\in \Lambda_N} e^{i\langle \omega,x\rangle} \,
\sum_{m \in \mathbb{Z}^d \setminus \{0\}} i^{|\alpha|}[(\omega + (2N + 1)m)^\alpha - \omega^\alpha] \hat f_{\omega + (2N + 1)m}.
\end{align}
Parseval's identity for the discrete Fourier transform gives
\begin{align}
\norm{\tilde D^\alpha_N f_N - (D^\alpha f)_N}_2^2
= (2N + 1)^d\sum_{\omega \in \Lambda_N}
\left|\sum_{m \neq 0}i^{|\alpha|}\bigl[(\omega + (2N + 1)m)^\alpha - \omega^\alpha\bigr]
\hat f_{\omega + (2N + 1)m}\right|^2.
\end{align}
For $\omega \in \Lambda_N$ and $m \neq 0$, each coordinate satisfies $|\omega_j| \leq |\omega_j + (2N + 1)m_j|$. Consequently,
\begin{equation}
\left|(\omega + (2N + 1)m)^\alpha - \omega^\alpha\right|
\leq 2\norm{\omega + (2N + 1)m}_\infty^{|\alpha|}.
\end{equation}
We choose the auxiliary exponent $k = d$. The weighted Cauchy--Schwarz inequality and \eqref{eq:usefuleq} therefore imply
\begin{align}
\norm{\tilde D^\alpha_N f_N - (D^\alpha f)_N}_2^2
& \leq 4(2N + 1)^d\left(\max_{\omega \in \Lambda_N}\sum_{m \neq 0}\norm{\omega + (2N + 1)m}_\infty^{-2k}\right)
\sum_{\nu \notin \Lambda_N}\norm{\nu}_\infty^{2k + 2|\alpha|}|\hat f_\nu|^2 \nonumber\\
& \leq 2^{d + 3}(2N + 1)^dN^{-2k}\sum_{\nu \notin \Lambda_N}\norm{\nu}_\infty^{2k + 2|\alpha|}|\hat f_\nu|^2.
\label{eq:usefuleq-again}
\end{align}
Thus the error is bounded by the weighted tail:
\begin{equation}
\label{eq:weighted-derivative-error}
\norm{\tilde D_N^\alpha f_N - (D^\alpha f)_N}_2^2
\leq 2^{d + 3}(2N + 1)^dN^{-2k}
\sum_{\nu \notin \Lambda_N}\norm{\nu}_\infty^{2k + 2|\alpha|}|\hat f_\nu|^2.
\end{equation}
Using \eqref{eq:weighted-tail-holder} with $\ell = k+|\alpha|$ gives
\begin{equation}
\label{eq:diff-dist-bound-holder}
\norm{\tilde D^\alpha_N f_N - (D^\alpha f)_N}_2
\leq \frac{2^{(3d + 3)/2}C(2N + 1)^{d/2}}{\sqrt{2p + 2q - 2k - 2|\alpha| - d}}N^{|\alpha| + d/2 - p - q}
\end{equation}
when $p + q > k+|\alpha| + d/2$. For the Gevrey rows, \cref{prop:all-decay-rates} supplies the coefficient bound since $(rN)^{1/s} \geq 16$, and \eqref{eq:weighted-tail-gevrey} gives
\begin{equation}
\label{eq:diff-dist-bound-gevrey}
\norm{\tilde D^\alpha_N f_N - (D^\alpha f)_N}_2
\leq 2^{(3d + 3)/2}C(2N + 1)^{d/2}\sqrt{\frac{s}{br^{1/s}}N^{2|\alpha| + d - 1/s}}e^{-b(rN)^{1/s}}
\end{equation}
when $b(rN)^{1/s} \geq \max\{s(k+|\alpha|), s(2k + 2|\alpha| + d) - 1\}$, with the same choices of $b$ as in \eqref{eq:weighted-tail-gevrey}. With $k = d$ and $s \geq 1$, this reduces to the stated integral condition. Relative to $\bar{\mathcal T}_N$, the algebraic estimate differs from the exponential estimate by the factor
\begin{equation}
\left(\frac{2p + 2q - d}{2p + 2q - 3d - 2|\alpha|}\right)^{1/2} \leq \sqrt3 < 2,
\quad \text{provided that} \quad p + q \geq 2d + \frac32|\alpha|.
\end{equation}
Using the factor $2$ for both cases gives \eqref{eq:derivative-envelope-bound}.

For the uniform estimate, Cauchy--Schwarz applied to the full Fourier sum gives
\begin{align}
\max_{x \in \Gamma_N}|\tilde D_N^\alpha f_N(x) - D^\alpha f(x)|^2
& \leq 4\left(\sum_{\nu \notin \Lambda_N}\norm{\nu}_\infty^{-2k}\right)
\left(\sum_{\nu \notin \Lambda_N}\norm{\nu}_\infty^{2k + 2|\alpha|}|\hat f_\nu|^2\right) \nonumber\\
& \leq \frac{4^{d + 1}}{2k - d}N^{d - 2k}
\sum_{\nu \notin \Lambda_N}\norm{\nu}_\infty^{2k + 2|\alpha|}|\hat f_\nu|^2.
\end{align}
Here the last inequality uses \eqref{eq:tail-shell-bound} with $g(x) = x^{-k}$. Since $k = d$, substituting the exponential weighted-tail estimate gives the factor $2^{d + 1}/\sqrt d$ relative to $N^{|\alpha| + d/2}\bar{\mathcal T}_N$. The algebraic case adds the same factor bounded by $\sqrt3 < 2$ above, so $2^{d + 2}/\sqrt d$ covers both cases and proves \eqref{eq:uniform-derivative-envelope}. The choice $b = s/2$ gives the finite-order entire condition, while finite bandwidth gives exact differentiation. This establishes both stated envelope bounds.
\end{proof}

\begin{cor}
\label{cor:qft-diff-err}
Let $N \in \tilde\Omega((s(d+|\alpha|))^s/r)$ and $f$ be an $s$-Gevrey and $2\pi$-periodic $d$-variate function with Gevrey parameters $C$ and $r$, and $s \geq 1$. We have
\begin{align}
\norm{\tilde D^\alpha_N f_N - (D^\alpha f)_N}_2
& \leq (2N + 1)^{d/2}\max_{x \in \Gamma_N}|\tilde D_N^\alpha f_N(x) - D^\alpha f(x)| \nonumber\\
& \leq C e^{-\frac12(rN)^{1/s}}.
\label{eq:uniform-derivative-error}
\end{align}
More explicitly, we may choose
\begin{equation}
\label{eq:N-bound-for-good-gevrey-fourier-diff}
N \geq \frac1r\left[64s(d+|\alpha|)\log\!\left(\frac{64s(d+|\alpha|)}{r^{1/s}}\right)\right]^s.
\end{equation}
\end{cor}

\begin{proof}
We write $w = s(d+|\alpha|)$, $L = \log(64w/r^{1/s})$, $y = N^{1/s}$, and $t = r^{1/s}y$, and use $b = 3/4$ as in \eqref{eq:tail-exponential-bound}. The same two cases as in the proof of \cref{cor:qft-func-err} give $t \geq 32w$, establishing the coefficient onset in \cref{prop:all-decay-rates}. Also, $bt \geq 24w$, so the conditions of \cref{prop:dft-derivative-bound} hold.

Substituting the exponential expression for $\bar{\mathcal T}_N$ into \eqref{eq:uniform-derivative-envelope} gives
\begin{equation}
\max_{x \in \Gamma_N}|\tilde D_N^\alpha f_N(x) - D^\alpha f(x)|^2
\leq \frac{16^{d + 1}sC^2}{d\, br^{1/s}}N^{2|\alpha| + 2d - 1/s}e^{-2b(rN)^{1/s}}.
\end{equation}
After multiplication by $(2N + 1)^d \leq 3^dN^d$, the squared prefactor excluding $C^2$ is at most $16 \cdot 48^d s\, y^{s(2|\alpha| + 3d) - 1}/(d\, br^{1/s})$. Since $d, s \leq w$, $b \geq 3/4$, $y \geq 1$, and $L \leq t/(64w)$, \eqref{eq:lny-strong} gives
\begin{align}
\log\!\left(\frac{16 \cdot 48^d s}{d\, br^{1/s}}y^{s(2|\alpha| + 3d) - 1}\right)
& \leq 7w + L + 3w\log y \nonumber\\
& \leq \frac{3t}{16} + (3w + 1)L \leq \frac{t}{4} \leq \frac{t}{2}.
\label{eq:log-ineq}
\end{align}
Taking square roots contributes at most $e^{(rN)^{1/s}/4}$ against $e^{-b(rN)^{1/s}}$. Since $b \geq 3/4$, this proves the stated uniform and Euclidean error bounds.
\end{proof}

\begin{rmk}
For $0 < s < 1$, the decay estimate in Case~8 of \cref{prop:all-decay-rates} gives
\begin{equation}
\norm{\tilde D_N^\alpha f_N - (D^\alpha f)_N}_2 \leq Ce^{-\frac{s}{4}(rN)^{1/s}}
\end{equation}
by the same argument with $w = \max(d, s(d+|\alpha|))$. In \eqref{eq:N-bound-for-good-gevrey-fourier-diff}, we replace $s(d+|\alpha|)$ by this $w$ and substitute $r \mapsto (s/2)^s r$. The resulting cutoff gives $(s/2)(rN)^{1/s} \geq 32w$, so it controls the coefficient onset and the squared prefactor, giving the stated error.
\end{rmk}

\section{Convolutional structure of Fourier derivatives}
\label{sec:diff-convolution}

Recall again that
\begin{align}
\label{eqn:approximate-differential-again}
\tilde D^\alpha= (F_N^{\otimes d})^{-1} \diag\left( \left\{\prod_{j =1}^{d} (i\omega_j)^{\alpha_j} \right\}_{\omega\in \Lambda_N} \right) F_N^{\otimes d}.
\end{align}
Since $\tilde D^\alpha$ is the projection of $D^\alpha$ on the truncated Fourier domain, we still have
\begin{equation}
\label{eq:fourier-derivative-multiplier}
\widehat{(\tilde D^\alpha f)}_{\ \omega}
= \left( \prod_{j=1}^d (i\omega_j)^{\alpha_j} \right)
\hat f_\omega, \qquad \text{for} \quad \omega \in \Lambda_N .
\end{equation}

We now define a convolution kernel $a^\alpha : \Gamma_N \to \mathbb{C}$ as the inverse discrete Fourier transform of the multiplier,
\begin{equation}
\label{eq:kernel-def}
a^\alpha(z)
= \frac{1}{(2N+1)^d}
\sum_{\omega \in \Lambda_N}
\left( \prod_{j=1}^d (i\omega_j)^{\alpha_j} \right)
e^{i \langle \omega , z \rangle},
\qquad \text{for} \quad z \in \Gamma_N .
\end{equation}
Then $\tilde D^\alpha$ acts as a periodic convolution on $\Gamma_N$:
\begin{equation}
\label{eq:convolution-full}
(\tilde D^\alpha f)(x)
= \sum_{y \in \Gamma_N} f(y)\, a^\alpha(x - y),
\qquad \text{for} \quad x \in \Gamma_N,
\end{equation}
where the subtraction on $\Gamma_N$ is modular. Note that $a^\alpha$ factorizes as
\begin{equation}
\label{eq:kernel-separable}
a^\alpha(z)
= \prod_{j=1}^d a_{\alpha_j}(z_j),
\end{equation}
where the one-dimensional kernels are
\begin{equation}
\label{eq:1d-kernel-def}
a_k(t) = \frac{1}{2N+1}
\sum_{\omega=-N}^{N} (i\omega)^k \, e^{i \omega t}.
\end{equation}

\begin{prop}
\label{prop:ak-sums-to-zero}
Given integers $N\geq 1$ and $k\geq 0$
\begin{equation}
\label{eq:ak-sum}
\sum_{n=-N}^{N} a_k[n] \;=\;
\begin{cases}
1, & k=0,\\
0, & k\geq 1.
\end{cases}
\end{equation}
\end{prop}
\begin{proof}
Since
\begin{equation}
\label{eq:sum-ak-swap}
\sum_{n=-N}^{N} a_k[n]
= \frac{1}{2N+1}\sum_{\omega=-N}^{N} (i\omega)^k
\sum_{n=-N}^{N} e^{i\omega t_n}
\end{equation}
using the $(2N+1)$-st roots of unity:
\begin{equation}
\label{eq:roots-unity-sum}
\sum_{n=-N}^{N} e^{i\omega t_n}
=\sum_{n=-N}^{N} \exp\!\left(\frac{2\pi i\,\omega n}{2N+1}\right)
=
\begin{cases}
2N+1, & \omega=0,\\
0, & \omega\neq 0,
\end{cases}
\end{equation}
so only the $\omega=0$ term contributes in \eqref{eq:sum-ak-swap}. Hence
\begin{equation}
\label{eq:ak-sum-eval}
\sum_{n=-N}^{N} a_k[n]
= \frac{1}{2N+1}(i\cdot 0)^k(2N+1) = (i\cdot 0)^k,
\end{equation}
which equals $1$ when $k=0$ and equals $0$ for all $k\geq 1$, proving \eqref{eq:ak-sum}.
\end{proof}

\begin{exm}
\label{exm:first-deriv-kernel}
We can derive a closed formula for $a_1$ as follows. We define the Dirichlet kernel
\begin{equation}
\label{eq:dirichlet-kernel}
D_N(t)
= \sum_{\omega=-N}^{N} e^{i \omega t}
= \frac{\sin\!\left((2N+1) \frac{t}2\right)}{\sin(\frac{t}2)} .
\end{equation}
Since
\begin{equation}
\label{eq:dirichlet-derivative}
D_N'(t) = \sum_{\omega=-N}^{N} (i\omega)\, e^{i \omega t},
\end{equation}
the kernel $a_1$ satisfies
\begin{equation}
\label{eq:a1-dirichlet}
a_1(t) = \frac{1}{2N+1} D_N'(t).
\end{equation}
Evaluating at grid points $t = \frac{2\pi n}{2N+1}$ with $n \in \{-N,\ldots,N\}$ yields
\begin{align}
\label{eq:fourier-a}
a_1[n] =\begin{cases}
0,& \text{if } n=0, \\
\frac{(-1)^{n}}{2\sin\left( \frac{\pi\,n}{2N+1} \right)},&
\text{otherwise,}
\end{cases}
\end{align}
which is the standard odd-grid Fourier spectral first-derivative stencil that was also used in \cite{motamedi2022gibbs}.
\end{exm}

\begin{exm}
\label{exm:higher-order-kernel}
We can treat higher-order kernels in exactly the same way. Differentiating the Dirichlet kernel $k$ times gives
\begin{equation}
\label{eq:dirichlet-kth-derivative}
D_N^{(k)}(t)=\sum_{\omega=-N}^{N}(i\omega)^k e^{i\omega t},
\end{equation}
and therefore
\begin{equation}
\label{eq:ak-dirichlet}
a_k(t)=\frac{1}{2N+1}\,D_N^{(k)}(t).
\end{equation}
In particular,
\begin{equation}
\label{eq:a2-dirichlet}
a_2(t)=\frac{1}{2N+1}\,D_N''(t).
\end{equation}
Evaluating again at the odd grid points $t_n=\frac{2\pi n}{2N+1}$, $n\in\{-N,\ldots,N\}$, yields the standard second-derivative Fourier spectral stencil:
\begin{equation}
\label{eq:a2-offdiag}
a_2[n]=\frac{(-1)^{n+1}\cos\!\left(\frac{\pi n}{2N+1}\right)}{2\sin^2\!\left(\frac{\pi n}{2N+1}\right)},
\qquad \text{for} \quad n\in\{-N,\ldots,N\}\setminus\{0\}.
\end{equation}
The diagonal entry follows directly from \eqref{eq:1d-kernel-def}:
\begin{equation}
\label{eq:a2-diag}
a_2[0]= -\frac1{2N+1}\sum_{\omega=-N}^{N}\omega^2= -\frac{N(N+1)}{3}.
\end{equation}
Combining \eqref{eq:a2-offdiag} and \eqref{eq:a2-diag} gives the closed form
\begin{equation}
\label{eq:a2-closed}
a_2[n]=
\begin{cases}
-\dfrac{N(N+1)}{3}, & n=0,\\[1.0ex]
\dfrac{(-1)^{n+1}\cos\!\left(\dfrac{\pi n}{2N+1}\right)}{2\sin^2\!\left(\dfrac{\pi n}{2N+1}\right)}, & n\neq 0,
\end{cases}
\end{equation}
for $n\in\{-N,\ldots,N\}$. More generally, \eqref{eq:ak-dirichlet} provides a closed form for all $k\geq 0$ in terms of derivatives of $D_N$, and the multi-dimensional kernel remains separable as in \eqref{eq:kernel-separable}.
\end{exm}

\section{Block-encoding of the derivative multiplier via LCU}
\label{app:R_N-calculation}

In this section, we present a circuit that constructs an $(N,\log N+2,0)$-block-encoding $U_{\partial,N}$ of the diagonal multiplier for $\tilde{\partial}_N$ using $R_N = \mathcal O(\log N)$ gates with the ripple-carry comparator. Depending on the available ancilla and circuit-depth budgets, the comparator can use $2\log N-1$ Toffoli gates with one scratch qubit, or $\Theta(\log N)$ Toffoli gates with $\Theta(\log\log N)$ depth and no scratch qubits. The operator $\tilde{\partial}_N$ was originally defined in \eqref{eq:single-diag}:
\begin{equation}
\tilde\partial_N= F_N^{-1} \diag\{(i \omega)_{\omega\in \{-N, \cdots, N\}}\} F_N.
\end{equation}
We choose a power-of-two cutoff $N = 2^n \geq 2$, but keep the odd grid with $2N+1$ points. We first construct the block on frequencies $-N,\ldots,N-1$ and then include the endpoint $+N$ explicitly. The real scaled multiplier on the first block is
\begin{equation}
\bar\partial_N := \frac{1}{N}
\begin{pmatrix}
    -N & & &\\
    & -N+1 & &\\
    & & \ddots &\\
    & & & N-1
\end{pmatrix}.
\end{equation}
We write the two diagonal blocks of $\bar\partial_N$ as $-D_-$ and $D_+$, respectively, where
\begin{align}
    D_-&=\frac{1}{N}\diag\{N,N-1,\ldots,1\},\\
    D_+&=\frac{1}{N}\diag\{0,1,\ldots,N-1\},
\end{align}
so that $\bar\partial_N$ is the direct sum
\begin{align}
    \bar\partial_N=\ket{0}\bra{0}\otimes (-D_-)+\ket{1}\bra{1}\otimes D_+.
\end{align}

Our construction proceeds first by deriving circuits for $U_-$ and $U_+$, block encoding $D_-$ and $D_+$, respectively (\cref{prop:D-pm-block-encoding}). We combine them with the phase $R_z(\pi) = \diag(-i,i)$ on the sign register to obtain a unitary $U_{\partial,N}$ whose encoded block is $i\bar\partial_N$:
\begin{equation}
U_{\partial,N}=
\begin{tikzpicture}[baseline={(m.center)}]
\matrix (m) [matrix of math nodes, left delimiter={(}, right delimiter={)},
    nodes={minimum width=0.9em, minimum height=0.9em, anchor=center},
    column sep=0.2em, row sep=0.2em] {
    -iD_- &[0.05em] &[0.6em] & \\[0.05em]
     & iD_+ & & \\[0.6em]
    & & & \\
    & & & \\
};
\draw[dashed] ($(m-2-1.south west)!0.7!(m-3-1.north west)+(-0.8em,0)$) -- ($(m-2-4.south east)!0.7!(m-3-4.north east)+(0.5em,0)$);
\draw[dashed] ($(m-1-2.north east)!0.7!(m-1-3.north west)+(0,0.45em)$) -- ($(m-4-2.south east)!0.7!(m-4-3.south west)+(0,-0.45em)$);
\node[scale=1.2] at ($(m-1-3.north west)!0.5!(m-2-4.south east)$) {$*$};
\node[scale=1.2] at ($(m-3-1.north west)!0.5!(m-4-2.south east)$) {$*$};
\node[scale=1.2] at ($(m-3-3.north west)!0.5!(m-4-4.south east)$) {$*$};
\end{tikzpicture}
,\quad \text{so that} \quad
\bra{0_{\mathtt{anc}}}\bra{0_{\mathtt{bool}}} U_{\partial,N} \ket{0_{\mathtt{anc}}}\ket{0_{\mathtt{bool}}}=i\bar\partial_N.
\end{equation}
The full circuit is shown in \cref{fig:fullcircuit}. We implement $U_-$ and $U_+$ on registers \texttt{data}, \texttt{bool}, and \texttt{anc}, where \texttt{bool} and \texttt{anc} are auxiliary registers and $D_-$ or $D_+$ is applied to \texttt{data}. The additional system register \texttt{ctrl} supplies the phases $-i$ and $i$ through $R_z(\pi)$, forming the direct sum of $-iD_-$ and $iD_+$. The Cuccaro comparator's scratch qubit is initialized and returned to zero and is omitted from the diagrams and state expressions.

For the odd grid, we add a system flag. Flag zero selects the circuit above; flag one with all other system bits zero represents $+N$ and applies the phase $i$. For flag one with any other system bit nonzero, we flip \texttt{bool}, giving zero successful amplitude on the unused basis states. These controlled operations cost $\mathcal O(\log N)$ gates.

We embed $F_N$ as $F_N\oplus I$ on this padded register. The spatial derivative block is $\Delta_N = (I_{\mathtt{aux}}\otimes F_N^\dagger)U_{\partial,N}(I_{\mathtt{aux}}\otimes F_N)$, whose encoded block is $\tilde\partial_N/N$ on the $2N+1$ valid states. Here \texttt{aux} comprises \texttt{anc}, \texttt{bool}, and the comparator's scratch qubit. The all-zero system bit string represents the frequency $-N$, on which $U_{\partial,N}$ acts with phase $-i$ and zero auxiliary workspace. Consequently, $\Delta_N$ acts in the same way on $\ket\chi = F_N^\dagger\ket{-N}$. We use this eigenstate for the idle routing register in \cref{fig:D-alpha-circuit}, preparing it with $F_N^\dagger$ and unpreparing it with $F_N$.

\begin{figure}[t]
\centering
\begin{adjustbox}{max width=0.65\textwidth}
\begin{tikzpicture}[
    thick,
    gate/.style={draw, minimum size=0.7cm, inner sep=2pt, font=\small, fill=white},
    every node/.style={font=\small},
]
\def\yc{0}      
\def\yb{-1}     
\def\ya{-2}     
\def\yd{-3}     

\def\xone{1.4}   
\def\xtwo{3.1}   
\def\xthree{4.9} 
\def\xfour{6.7}  
\def\xend{8.0}

\draw (0,\yc) -- (\xend,\yc);
\draw (0,\yb) -- (\xend,\yb);
\draw (0,\ya) -- (\xend,\ya);
\draw (0,\yd) -- (\xend,\yd);

\node[anchor=east] at (-0.2,\yc) {$\mathtt{ctrl}\ \ket{c}_1$};
\node[anchor=east] at (-0.2,\yb) {$\mathtt{bool}\ \ket{0}_1$};
\node[anchor=east] at (-0.2,\ya) {$\mathtt{anc}\ \ket{0}_n$};
\node[anchor=east] at (-0.2,\yd) {$\mathtt{data}\ \ket{\phi}_n$};

\node[gate, minimum width=1.1cm] at (\xone,\yc) {$R_z(\pi)$};
\node[gate] at (\xone,\ya) {$H^n$};

\draw (\xtwo,\yc) -- (\xtwo,\yb);
\filldraw (\xtwo,\yc) circle (0.05);
\node[gate] at (\xtwo,\yb) {$X$};

\node[gate, minimum width=1.1cm] at (\xthree,\yb) {$a<d$};
\draw (\xthree,{\yb-0.35}) -- (\xthree,\yd);
\filldraw (\xthree,\ya) circle (0.05);
\filldraw (\xthree,\yd) circle (0.05);

\node[gate] at (\xfour,\ya) {$H^n$};

\end{tikzpicture}
\end{adjustbox}
\caption{Circuit for $U_{\partial,N}$ on the first $2N$ frequency states, combining the block-encodings of \cref{fig:Dcircuit} via the system register \texttt{ctrl}: a controlled bit flip attaches the optional $X$ gate distinguishing $U_+$ from $U_-$, an additional single-qubit gate $R_z(\pi)$ on \texttt{ctrl} supplies the phases $-i$ and $i$, and the shared comparator $C_<$ then acts on \texttt{anc} and \texttt{data} exactly as in \cref{fig:a,fig:b}.}
\label{fig:fullcircuit}
\end{figure}
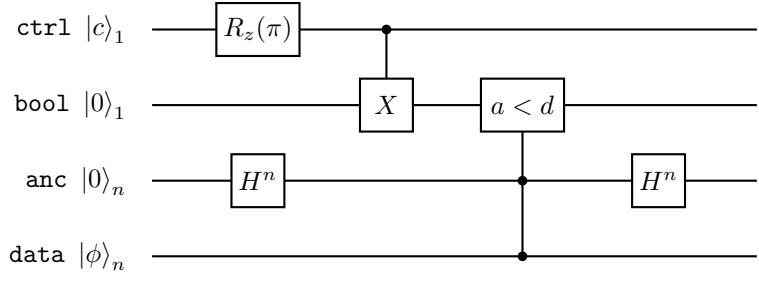

\begin{prop}\label{prop:D-pm-block-encoding}
For each $\sigma\in\{-,+\}$, the circuit in \cref{fig:a,fig:b}, using the ripple-carry comparator, implements $U_\sigma$, a $(1,\log N+2,0)$-block-encoding of $D_\sigma$, with $\mathcal O(\log N)$ elementary gates. The auxiliary count includes the comparator's clean scratch qubit, omitted from the diagrams and state expressions.
\end{prop}

\begin{figure}[htbp]
\centering

\begin{subfigure}{0.48\textwidth}
\centering
\begin{adjustbox}{max width=\linewidth}
\begin{tikzpicture}[
    thick,
    gate/.style={draw, minimum size=0.7cm, inner sep=2pt, font=\small, fill=white},
    every node/.style={font=\small},
]
\def\yb{0}      
\def\ya{-1}     
\def\yd{-2}     

\def\xone{1.4}   
\def\xtwo{3.1}   
\def\xthree{4.9} 
\def\xend{6.2}

\draw (0,\yb) -- (\xend,\yb);
\draw (0,\ya) -- (\xend,\ya);
\draw (0,\yd) -- (\xend,\yd);

\node[anchor=east] at (-0.2,\yb) {$\mathtt{bool}\ \ket{0}_1$};
\node[anchor=east] at (-0.2,\ya) {$\mathtt{anc}\ \ket{0}_n$};
\node[anchor=east] at (-0.2,\yd) {$\mathtt{data}\ \ket{\phi}_n$};

\node[gate] at (\xone,\ya) {$H^n$};

\node[gate, minimum width=1.1cm] at (\xtwo,\yb) {$a<d$};
\draw (\xtwo,{\yb-0.35}) -- (\xtwo,\yd);
\filldraw (\xtwo,\ya) circle (0.05);
\filldraw (\xtwo,\yd) circle (0.05);

\node[gate] at (\xthree,\ya) {$H^n$};

\end{tikzpicture}
\end{adjustbox}
\caption{}
\label{fig:a}
\end{subfigure}
\hfill
\begin{subfigure}{0.48\textwidth}
\centering
\begin{adjustbox}{max width=\linewidth}
\begin{tikzpicture}[
    thick,
    gate/.style={draw, minimum size=0.7cm, inner sep=2pt, font=\small, fill=white},
    every node/.style={font=\small},
]
\def\yb{0}      
\def\ya{-1}     
\def\yd{-2}     

\def\xone{1.4}   
\def\xtwo{3.1}   
\def\xthree{4.9} 
\def\xend{6.2}

\draw (0,\yb) -- (\xend,\yb);
\draw (0,\ya) -- (\xend,\ya);
\draw (0,\yd) -- (\xend,\yd);

\node[anchor=east] at (-0.2,\yb) {$\mathtt{bool}\ \ket{0}_1$};
\node[anchor=east] at (-0.2,\ya) {$\mathtt{anc}\ \ket{0}_n$};
\node[anchor=east] at (-0.2,\yd) {$\mathtt{data}\ \ket{\phi}_n$};

\node[gate] at (\xone,\yb) {$X$};
\node[gate] at (\xone,\ya) {$H^n$};

\node[gate, minimum width=1.1cm] at (\xtwo,\yb) {$a<d$};
\draw (\xtwo,{\yb-0.35}) -- (\xtwo,\yd);
\filldraw (\xtwo,\ya) circle (0.05);
\filldraw (\xtwo,\yd) circle (0.05);

\node[gate] at (\xthree,\ya) {$H^n$};

\end{tikzpicture}
\end{adjustbox}
\caption{}
\label{fig:b}
\end{subfigure}

\caption{Block-encodings of $D_-$ (a) and $D_+$ (b) (\cref{prop:D-pm-block-encoding}). In both circuits a uniform superposition over \texttt{anc} is compared against \texttt{data} using $C_<$, and the comparison bit is written onto \texttt{bool}; the two are identical except that, in (b), an $X$ gate is applied to \texttt{bool} before the comparator, flipping which branch is selected.}
\label{fig:Dcircuit}
\end{figure}
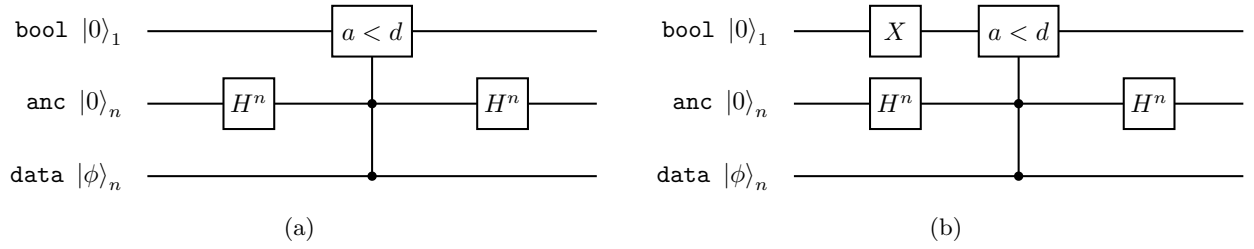

\begin{proof}
Let $\tau=0$ if $\sigma=-$ and $\tau=1$ if $\sigma=+$, so that $\tau$ counts the number of $X$ gates applied to \texttt{bool} before the comparator. We track the action of the circuit for $U_\sigma$ on an arbitrary computational basis state $\ket{x}_d$ of the \texttt{data} register:
\begin{itemize}
\item Generate a uniform superposition with $H^{\otimes \log N}$ on \texttt{anc}:
\begin{equation}
\frac{1}{\sqrt{N}}\sum_{j=0}^{N-1} \ket{0}_b\ket{j}_a\ket{x}_d.
\end{equation}
\item Apply $X^\tau$ to \texttt{bool} (a no-op for $U_-$, a bit flip for $U_+$), then the reversible comparator circuit $C_<$, which writes the comparison bit onto \texttt{bool}:
\begin{equation}
C_<:\ket{\tau}_b\ket{j}_a\ket{x}_d\mapsto \ket{\tau\oplus 1_{j<x}}_b\ket{j}_a\ket{x}_d,
\end{equation}
yielding the state
\begin{equation}
\frac{1}{\sqrt{N}}\sum_{j=0}^{x-1} \ket{\tau\oplus1}_b\ket{j}_a\ket{x}_d+\frac{1}{\sqrt{N}}\sum_{j=x}^{N-1} \ket{\tau}_b\ket{j}_a\ket{x}_d.
\end{equation}
\item Apply $H^{\otimes \log N}$ to \texttt{anc}. The $\ket{0}_b$ branch above contains the $j=x,\ldots,N-1$ terms when $\tau=0$, and the $j=0,\ldots,x-1$ terms when $\tau=1$, so the coefficient of $\ket{0}_b\ket{0}_a$ is $\frac{N-x}{N}$ if $\sigma=-$ and $\frac{x}{N}$ if $\sigma=+$, i.e.\ exactly $D_\sigma(x)$ in both cases.
\end{itemize}

This confirms that $U_\sigma$ is a correct block-encoding of $D_\sigma$:
\begin{equation}
\bra{0}_b\bra{0}_a U_\sigma\ket{0}_b\ket{0}_a\ket{x}_d= D_\sigma(x)\,\ket{x}_d.
\end{equation}
The resources used are the same for both values of $\sigma$:
\begin{itemize}[noitemsep, topsep=5pt]
\item The comparator stores its output in the \texttt{bool} qubit. The ripple-carry implementation requires $2\log N-1$ Toffoli gates, $\mathcal O(\log N)$ CNOT and $X$ gates, and one scratch qubit \cite{cuccaro2004new}, while a depth-optimized implementation achieves $\Theta(\log\log N)$ depth with $\Theta(\log N)$ Toffoli gates and no scratch qubits \cite[Theorem 2]{vandaele2026asymptotically}. The choice of comparator can be adapted to the available ancilla and depth budgets.
\item $2\log N$ Hadamard gates; and,
\item at most one additional $X$ gate.
\end{itemize}
Therefore, using the ripple-carry comparator, $U_\sigma$ is a $(1,\log N+2,0)$-block-encoding of $D_\sigma$ for each $\sigma\in\{-,+\}$.
\end{proof}

\section{Prior literature on solving the Poisson equation}
\label{app:elliptic_literature_review}

We review quantum algorithms for the Poisson equation with periodic and Dirichlet boundary conditions, including the shifted equation $-\Delta u+u=b$ considered in \cite{tong2021fast}. The methods are drawn from three papers: a Hamiltonian-simulation-based approach \cite{cao2013quantum}; a finite-difference and a spectral method \cite{childs2021high}; and a quantum-singular-value-transformation (QSVT) approach, together with its accelerated ``fast inversion'' variant, \cite{tong2021fast}. We specialize \cref{cor:poisson,cor:poisson-dirichlet} to $\Sigma=I$ for the comparison in \cref{tab:poisson-comparison-isotropic}. In the periodic case, $\norm{\Sigma}_{1,1}=d$ and $\sigma_{\min}^+=1$. The diagonal implementation has normalization $dN^2$ and $\tilde{\mathcal O}(d)$ gates per block call by \eqref{eq:laplacian-cost}, giving $\tilde{\mathcal O}(d^2N^2\log(1/\epsilon))$ gates. Unlike the reflected FDM of \cite{childs2021high}, our Dirichlet construction does not require a smooth odd periodic extension.

\begin{table}[b]
\centering
\renewcommand{\arraystretch}{1.4}
\begin{adjustbox}{max width=\textwidth}
\begin{tabular}{@{}llll@{}}
\toprule
Method & Equation, boundary condition & Regularity assumption & Leading gate complexity \\
\midrule
Hamiltonian sim.\ \cite{cao2013quantum} & $-\nabla^2u = f$, Dirichlet & $u \in \mathscr{C}^r$ (our extension) & $\epsilon^{-4/(r-2)}\max\{d, \log(1/\epsilon)\}(\log d\!+\!\log(1/\epsilon))^3$ \\
Adaptive FDM \cite{childs2021high} & $-\nabla^2u=f$, periodic & $k=\Theta\!\left(d\log(|u^{(2k+1)}|/\epsilon)\right)$ & $d^2k^{4.5}\sqrt{\log(dk^3/\epsilon)}$ \\
Adaptive FDM \cite{childs2021high} & $-\nabla^2u=f$, Dirichlet ($u=0$) & same $k$, for the reflected solution & $d^2k^{4.5}\sqrt{\log(dk^3/\epsilon)}$ \\
Spectral \cite{childs2021high} & $-\nabla^2u=f$, periodic\,/\,Dirichlet & finite $g$ and $g'$ (defined below) & $(d+qd^2)\polylog(dg'/(g\epsilon))$ \\
QSVT \cite{tong2021fast} & $-\nabla^2u + u = b$, periodic & $u \in \mathscr{G}^s$ (our inference) & $d^{2s+1}\log^{2s+1}(1/\epsilon)$ \\
Fast inversion \cite{tong2021fast} & $-\nabla^2u + u = b$, periodic & $u \in \mathscr{G}^s$ (our inference) & $d\log^2(d\log(1/\epsilon))+\log^2(1/\epsilon)$ \\
Ours (\cref{cor:poisson}, isotropic) & $-\nabla^2u = \eta$, periodic & $u \in \mathscr{G}^s$ & $d^{2s + 2}\log^{2s+1}(1/\epsilon)$ \\
Ours (\cref{cor:poisson-dirichlet}, isotropic) & $-\nabla^2u = \eta$, Dirichlet & $u \in \mathscr{G}^s$, $s\geq2$, mixed jets $\nu$ & $\epsilon^{-2} \nu^8 d^{4s+8} \log^{4s+1}(1/\epsilon)$ \\
\bottomrule
\end{tabular}
\end{adjustbox}
\caption{A comparison of isotropic Poisson solvers in elementary gates, excluding implementation of the supplied input-state oracles. The row for \cite{cao2013quantum} is our higher-order extension for fixed integer $r\geq3$, with corrected-source preparation excluded and independent repetitions included for constant success probability. The Dirichlet FDM row for \cite{childs2021high} requires the derivative bound for the odd periodic extension; the spectral factor $q$ is defined below. The rows for \cite{tong2021fast} assume $\xi=\Theta(1)$ and use the Gevrey cutoff $N=\tilde{\mathcal O}(d^s\log^s(1/\epsilon))$ with the gate implementations described below. Solution regularity parameters and relative norms are fixed. We retained the mixed face-jet parameter $\nu$ for our Dirichlet row.}
\label{tab:poisson-comparison-isotropic}
\end{table}

\medskip
\noindent{\bf Hamiltonian simulation \cite{cao2013quantum}.}
The earliest quantum circuit for the Poisson equation \cite{cao2013quantum} solves the isotropic equation $-\nabla^2 u = f$ with zero Dirichlet boundary values by directly simulating the Hamiltonian of the discretized Laplacian, rather than by querying it through an oracle: since the discretized Laplacian is diagonalized by a sine transform implemented using the quantum Fourier transform (QFT), its exponential can be implemented from the analytically known eigenvalues alone, without ever invoking an oracle $O_A$ for the matrix entries. Their second-order stencil gives $N=\Theta(\epsilon^{-1/2})$ in the $r=4$ example. For fixed integer $r\geq3$, we extend the construction to obtain $N=\Theta(\epsilon^{-1/(r-2)})$ as follows.

On $[0,1]^d$, we set $h=1/N$, $k=\lceil(r-2)/2\rceil$, and let $T_j$ act as $\operatorname{tridiag}(-1,2,-1)$ in coordinate $j$. The higher-order matrix is
\begin{equation}
A_{h,k}=h^{-2}\sum_{j=1}^d p_k(T_j),
\quad \text{where} \quad
p_k(z)=\sum_{\ell=1}^k\frac{2z^\ell}{\ell^2\binom{2\ell}{\ell}}.
\end{equation}
This polynomial truncates $4\arcsin^2(\sqrt z/2)$ and retains the tensor sine eigenbasis. For $z=4\sin^2(\theta/2)$, $0\leq\theta\leq\pi$, positivity of its coefficients gives $z\leq p_k(z)\leq\theta^2$. Hence $4d\leq\lambda_{\min}(A_{h,k})$, $\lambda_{\max}(A_{h,k})\leq\pi^2dN^2$, and $\kappa(A_{h,k})\leq\pi^2N^2/4$.

The boundary rows implement odd reflection, which alone need not preserve $\mathscr C^r$ regularity. At a face $x_j=0$, we add $2\sum_{1\leq2\ell\leq r-1}t^{2\ell}\partial_j^{2\ell}u(0,x')/(2\ell)!$ to the reflected ghost value $-u(t,x')$, and move its stencil contribution to the right-hand side; we do the same at $x_j=1$. These corrections depend only on the source, since $-\Delta u=f$ and the zero boundary trace give
\begin{equation}
\left.\partial_j^{2\ell}u\right|_{\mathrm{face}}
=-\left.\sum_{q=0}^{\ell-1}(-\Delta_{\mathrm{tan}})^q
\partial_j^{2\ell-2-2q}f\right|_{\mathrm{face}}.
\end{equation}
Here $\Delta_{\mathrm{tan}}$ is the tangential Laplacian. One-sided differences of the grid values of $f$ approximate each derivative of total order $2\ell-2$ to error $\mathcal O(h^{r-2\ell})$. Its boundary contribution is multiplied by $h^{2\ell-2}$, so these errors and the Taylor remainders give residual $\mathcal O(h^{r-2})$ throughout the grid. The spectral gap then gives the same solution-error order, including normalized states for controlled relative derivative bounds. For fixed $r$, evaluating $p_k$ adds only a fixed number of arithmetic operations to each eigenvalue computation, preserving the leading per-attempt gate count of \cite{cao2013quantum}. We assume preparation of the boundary-corrected source state, rather than a free conversion from the original source state. This proves the higher-order extension used in the table.

\begin{thm}[{Higher-order extension of \cite{cao2013quantum}}]
\label{thm:litrev-cao-hamsim}
There exists a quantum circuit that solves the $d$-dimensional Poisson equation on a Dirichlet grid to error $\epsilon$ using, per attempt,
\begin{align*}
\max\{d,\log(1/\epsilon)\}\,\mathcal O\!\left((\log d+\log(1/\epsilon))^3\right)
&\quad \text{elementary gates and}\\
\max\{d,\log(1/\epsilon)\}\,\mathcal O\!\left((\log d+\log(1/\epsilon))^2\right)
&\quad \text{qubits (including ancillae).}
\end{align*}
More precisely, before simplifying using $\log d = \mathcal O(\log(1/\epsilon))$, the gate count is
\begin{equation*}
d\, \mathcal O\!\left((\log d + \log(1/\epsilon))^3\right) + \mathcal O\!\left(\log^4(1/\epsilon)\right),
\end{equation*}
where the $\log^4(1/\epsilon)$ term comes from implementing the controlled rotation at the end of the linear-systems algorithm. The success probability is $\Omega(N^{-4})$, so $\mathcal O(\epsilon^{-4/(r-2)})$ repetitions suffice for constant success probability, with constants depending on the fixed order $r$.
\end{thm}
The $\log^4(1/\epsilon)$ cost of this controlled rotation is a recurring bottleneck. As mentioned when we introduce the bit-oracle \eqref{eq:pde-coeff-oracle} for the coefficient functions, this can be improved to $2.5$ using the diagonal block-encoding of \cite[Lemma~48]{gilyen2019quantum}.

\medskip
\noindent{\bf Finite-difference and spectral methods \cite{childs2021high}.}
The finite-difference method in \cite{childs2021high} uses the Fourier-series linear solver of \cite{childs2017quantum}. Its sparse-access formulation queries the coefficient matrix $A$ of the linear system $Ax=b$ entrywise, through an oracle $\mathcal P:\ket{j,k,z}\mapsto\ket{j,k,A_{jk}\oplus z}$, together with a nonzero-location oracle and an oracle $\mathcal U_b$ that prepares $\ket b$ in time $\mathcal O(\polylog N)$. In the following quoted theorem only, $d$ denotes matrix sparsity and $N$ the matrix dimension, with $\norm{A}=1$:
\begin{thm}[{\cite[Theorem~3]{childs2017quantum}}]
The quantum linear systems problem can be solved with $\mathcal O(\kappa\sqrt{\log(\kappa/\epsilon)})$ uses of a Hamiltonian simulation algorithm that approximates $\exp(-iAt)$ for $t=\mathcal O(\kappa\log(\kappa/\epsilon))$ to precision $\mathcal O(\epsilon/(\kappa\sqrt{\log(\kappa/\epsilon)}))$, and $\mathcal O(\kappa\sqrt{\log(\kappa/\epsilon)})$ uses of $\mathcal U_b$. Using the Hamiltonian simulation algorithm cited there, this amounts to $\mathcal O(d\kappa^2\log^{2.5}(\kappa/\epsilon))$ queries to the sparse-access oracles and gate complexity $\mathcal O\!\left(d\kappa^2\log^{2.5}(\kappa/\epsilon)\left(\log N+\log^{2.5}(\kappa/\epsilon)\right)\right)$.
\end{thm}
The two algorithms in \cite{childs2021high} are an adaptive finite difference method (FDM) for the \emph{isotropic} Poisson equation and a spectral method for constant-coefficient anisotropic elliptic equations. The FDM discretizes the isotropic Laplacian into a linear combination of tensor products of circulant matrices and identities, whose eigenvalues are known in closed form, giving condition number $\mathcal O(dN^2)$. Its Hamiltonian simulation exploits this Fourier diagonalization rather than the generic sparse-access implementation. A \emph{fixed}-order stencil would incur a truncation error vanishing only polynomially in $N$, so the authors instead let the stencil order $k$ satisfy the self-consistency relation
\begin{equation}
k=\Theta\!\left(d\log\!\left(\left|\tfrac{d^{2k+1}u}{dx^{2k+1}}\right|\Big/\epsilon\right)\right)
\end{equation}
obtained by optimizing their error bound over both the grid size and the stencil order (their Section 3.4, subject to Lemma 2). This is not resolved to a closed form in $\epsilon$: since the right-hand side depends on the $(2k+1)$-th derivative of $u$, evaluating $k$ requires an independent bound on how $|u^{(m)}|$ grows with the order $m$, which \cite{childs2021high} does not supply for a general regularity class. Their Theorem~1, reproduced below, accordingly states the algorithm's complexity directly in terms of this still-implicit $k$ rather than as an explicit function of $\epsilon$ alone:
\begin{thm}[{\cite[Theorem~1]{childs2021high}}]
\label{thm:litrev-childs-fdm}
There is a quantum algorithm that outputs a state $\epsilon$-close to $\ket u$ using
\begin{equation}
\tilde{\mathcal O}\Biggl(d^{6.5}\log^{4.5}\!\left(\left|\frac{d^{2k+1}u}{dx^{2k+1}}\right|\Big/\epsilon\right)\sqrt{\log\!\left[d^4\log^3\!\left(\left|\frac{d^{2k+1}u}{dx^{2k+1}}\right|\Big/\epsilon\right)\Big/\epsilon\right]}\Biggr)
\end{equation}
elementary gates and
\begin{equation}
\tilde{\mathcal O}\Biggl(d^4\log^3\!\left(\left|\frac{d^{2k+1}u}{dx^{2k+1}}\right|\Big/\epsilon\right)\sqrt{\log\!\left[d^4\log^3\!\left(\left|\frac{d^{2k+1}u}{dx^{2k+1}}\right|\Big/\epsilon\right)\Big/\epsilon\right]}\Biggr)
\end{equation}
queries to the oracle for $f$, where $k$ is the adaptively chosen order of the finite difference stencil.
\end{thm}
Section 3.5 of \cite{childs2021high} also treats zero Dirichlet boundary values by the method of images, restricting the periodic difference matrix to an odd-reflection subspace. The same resource bound applies with the derivative term evaluated on the reflected periodic solution. Zero boundary values alone do not guarantee the required smoothness across the reflection faces; this is the qualification in the table's Dirichlet FDM row. The fixed-order boundary correction above is part of our extension of \cite{cao2013quantum}, not a claim made for the adaptive method of \cite{childs2021high}.

The same paper's second algorithm instead expands the solution in a truncated Fourier (periodic case) or Chebyshev (Dirichlet case) series and solves for the unknown spectral coefficients using \cite[Theorem~5]{childs2017quantum}. For the Poisson equation, the resulting condition number satisfies \cite[Corollary~1]{childs2021high}
\begin{equation}
\kappa_{L_{\mathrm{Poisson}}} \le (2N)^4,
\end{equation}
Their Theorem~2 includes both periodic and Dirichlet boundary conditions, including nonzero Dirichlet data. We specialize it to the Poisson operator:

\begin{thm}[Poisson case of {\cite[Theorem~2]{childs2021high}}]
\label{thm:litrev-childs-spectral}
Consider $-\Delta u=f$ with periodic or Dirichlet boundary conditions. There is a quantum algorithm that produces a state whose amplitudes are proportional to an approximation $v(\bm x)$ on a set of interpolation nodes $\bm x$, where $v(\bm x)/\norm{v(\bm x)}$ is $\epsilon$-close to $u(\bm x)/\norm{u(\bm x)}$ in $\ell^2$-norm, succeeding with probability $\Omega(1)$ and flagging success, using
\begin{equation}
\left(d+qd^2\right)\poly\!\left(\log(g'/(g\epsilon))\right)
\end{equation}
queries to the source and boundary-data oracles, where $q$ is their source/boundary state-combination factor \cite[Eq.~(4.3)]{childs2021high}, equal to $1$ for zero boundary data, and
\begin{equation}
g = \min_{\bm x}\norm{u(\bm x)}, \quad \text{and} \quad g' := \max_{\bm x}\max_{n\in\mathbb N}\norm{u^{(n+1)}(\bm x)}.
\end{equation}
The gate complexity exceeds the query complexity by a factor of $\poly(\log(d/\epsilon))$.
\end{thm}
The stated hypothesis is that the inhomogeneity $f$ merely lies in $\mathscr{C}^\infty$ (their Problem 1). However, the truncation order $n$ appearing throughout \cref{thm:litrev-childs-spectral} is set, in their proof, via a classical Chebyshev/Fourier truncation-error bound to $n = \left\lceil\log\Omega/\log\log\Omega\right\rceil$ with $\Omega := g'(1 + \epsilon)/(g\epsilon)$, and this in turn requires $g' = \max_{\bm x}\max_{n \in \mathbb N}\norm{u^{(n + 1)}(\bm x)}$ to be finite. This pins $u$ very narrowly: for a $2\pi$-periodic univariate function, writing $u(x) = \sum_k\hat u_k e^{ikx}$ and integrating $\hat u_k = \frac1{2\pi}\int u(x)e^{-ikx}dx$ by parts $m$ times gives $|\hat u_k|\, |k|^m \le \norm{u^{(m)}}_\infty$, so any nonzero coefficient at $|k| \ge 2$ forces $\norm{u^{(m)}}_\infty \to \infty$ as $m \to \infty$. Hence a uniform bound over derivative orders in this coordinate holds if and only if $u$'s Fourier support is confined to $|k| \le 1$, i.e.\ $u(x) = \hat u_0 + \hat u_1e^{ix} + \hat u_{-1}e^{-ix}$, just the constant and the fundamental mode. The frequency threshold changes under coordinate rescaling, but the all-orders bound remains much stronger than $\mathscr{C}^\infty$ smoothness or finite bandwidth. The table therefore confines to the cited derivative-dependent spectral bound, for both boundary conditions.

\medskip
\noindent{\bf QSVT and fast inversion \cite{tong2021fast}.}
The third paper solves linear systems via QSVT applied directly to a block-encoding of the coefficient matrix. Its worked PDE example \cite[Proposition~8]{tong2021fast} is periodic; it does not supply a Dirichlet PDE discretization or truncation-error analysis. We therefore confine to periodic boundary conditions for its two rows:
\begin{thm}[Standard QSVT linear system solver, {\cite{tong2021fast}}]
\label{thm:qsvt_qlsp}
Let $U_A$ be an $(\alpha,m,0)$-block-encoding of $A$ with condition number $\kappa$, let $U_b$ be the oracle preparing $\ket b$, and let $\xi=\norm{A^{-1}\ket b}\in[1/\norm{A},\norm{A^{-1}}]$. Then $\ket x\propto A^{-1}\ket b$ can be obtained to precision $\epsilon$ with success probability at least $1/2$ using $\mathcal O\!\left(\alpha\kappa^2/(\norm{A}^2\xi)\log(\kappa/(\norm{A}\xi\epsilon))\right)$ queries to $U_A$ and $U_A^\dagger$, and $\mathcal O(\kappa/(\norm{A}\xi))$ queries to $U_b$.
\end{thm}
For the \emph{isotropic} $d$-dimensional elliptic equation $-\nabla^2u+u=b$ on the torus, the discretized operator $A:=-\nabla^2+I$ is diagonalized exactly by the QFT, so only a block-encoding of the (known, diagonal) eigenvalue matrix is needed. The identity term $+u$ in their equation is added to keep the operator invertible. Their worked example further assumes $\xi:=\norm{A^{-1}\ket b}=\Theta(1)$, i.e., that the solution-to-source norm ratio neither vanishes nor blows up as more planewaves are added, the best possible scaling, since \cref{thm:qsvt_qlsp} only guarantees $\xi\ge1/\norm{A}$ in general. In the worst case $\xi=1/\norm{A}$, both methods' query counts instead pick up an extra factor of $\kappa$, but fast inversion retains its advantage over standard QSVT. Under $\xi=\Theta(1)$, combining the diagonalization with \cref{thm:qsvt_qlsp} gives:

\begin{thm}[{\cite[Proposition~8]{tong2021fast}}]
\label{prop:cost_elliptic}
Counting the supplied block encodings, diagonal-entry and reciprocal-rotation oracles, QFTs, and source preparation as unit-cost operations, a planewave discretization with grid size $\frac1N$ per dimension solves the $d$-dimensional elliptic equation with $\xi=\Theta(1)$ to precision $\epsilon$ with success probability at least $1/2$ with circuit depth and query count to the operator block encoding (and its adjoint) both $\mathcal O(dN^2\log(1/\epsilon))$, and query count to $U_b$ equal to $\mathcal O(1)$, via standard QSVT. Using ``fast inversion'', a direct block-encoding of the inverse diagonal matrix that bypasses QSVT's polynomial approximation of $1/x$, the circuit depth and the number of calls to the supplied operations and their inverses are all $\mathcal O(1)$ instead.
\end{thm}
Neither part of \cref{prop:cost_elliptic} resolves $N$ as a function of $\epsilon$. We can, however, infer a resolution ourselves, matching our own standing hypothesis: if the \emph{solution} $u=A^{-1}b$ itself (not $b$) is $s$-Gevrey, the tail-decay argument of \cref{prop:tail-decay-rate} applies to its planewave truncation with a Fourier-projected source. At fixed Gevrey parameters and relative solution norm this gives $N=\tilde{\mathcal O}(d^s\log^s(1/\epsilon))$, retaining the dimension dependence. Here $\xi=\norm{\hat u_N}/\norm{\hat b_N}$, so the $\xi=\Theta(1)$ promise is separate.

For a gate-level comparison, we use our $2\pi$-periodic torus convention, where the eigenvalues are integers $\lambda_\omega=1+\norm{\omega}^2$. Reversible squaring and a balanced addition tree compute $\lambda_\omega$ once in $\mathcal O(d\log^2N)$ gates. We retain this register throughout diagonal-basis QSVT and uncompute it before the final QFT, so Fourier transforms are not repeated at every query. With $2^m>1+dN^2$ and $m=\mathcal O(\log(dN))$, comparing a uniform $m$-bit register with $\lambda_\omega$ and undoing its Hadamards gives an exact block encoding of $\lambda_\omega/2^m$ in $\mathcal O(m)$ gates, as in \cref{app:R_N-calculation}. Including the ancilla phase operations, \cref{thm:qsvt_qlsp} therefore gives $\mathcal O(d\log^2N+dN^2\log(dN)\log(1/\epsilon))$ gates. Substituting the cutoff yields the QSVT gate bound in \cref{tab:poisson-comparison-isotropic}.

For fast inversion, we compute $\lfloor2^{b-1}/\lambda_\omega\rfloor$ by reversible binary division and use the same comparator construction on a uniform $b$-bit register. Its successful amplitude is $2^{-b}\lfloor2^{b-1}/\lambda_\omega\rfloor$, within $2^{-b}$ of $1/(2\lambda_\omega)$. Since $\xi=\Theta(1)$, choosing $b=\lceil\log_2(1/\epsilon)\rceil+\mathcal O(1)$ gives the required normalized-state accuracy with constant success probability. Division and uncomputation cost $\mathcal O((m+b)^2)$ gates, giving total gate complexity $\mathcal O(d\log^2N+\log^2(dN/\epsilon))$. Substituting the Gevrey cutoff gives $\mathcal O(d\log^2(d\log(1/\epsilon))+\log^2(1/\epsilon))$, including the reciprocal implementation rather than treating it as a unit-cost oracle.

\end{document}